\documentclass[
a4paper,reqno]{amsart}
\usepackage[T1]{fontenc}
\usepackage[english]{babel}
\usepackage{amssymb,bbm,enumerate}
\usepackage{color}
\usepackage[linktocpage=true,colorlinks=true, linkcolor=blue, citecolor=red, urlcolor=green]{hyperref}

\usepackage{float}
\restylefloat{table}

\usepackage{tikz}

\renewcommand{\phi}{\varphi}
\renewcommand{\ker}{\Ker}

\newcommand{\mc}[1]{\mathcal{#1}}
\newcommand{\mf}[1]{\mathfrak{#1}}
\newcommand{\mb}[1]{\mathbb{#1}}

\newcommand{\id}{\mathbbm{1}}

\newcommand{\tint}{{\textstyle\int}}

\DeclareMathOperator{\Mat}{Mat}
\DeclareMathOperator{\Hom}{Hom}
\DeclareMathOperator{\End}{End}

\DeclareMathOperator{\tr}{tr}
\DeclareMathOperator{\res}{Res}
\DeclareMathOperator{\ad}{ad}
\DeclareMathOperator{\im}{Im}

\DeclareMathOperator{\Ker}{Ker}
\DeclareMathOperator{\Span}{Span}

\DeclareMathOperator{\rk}{rk}
\DeclareMathOperator{\Res}{Res}

\DeclareMathOperator{\Cond}{\Gamma}

\theoremstyle{plain}
\newtheorem{theorem}{Theorem}[section]
\newtheorem{lemma}[theorem]{Lemma}

\newtheorem{corollary}[theorem]{Corollary}

\theoremstyle{definition}
\newtheorem{definition}[theorem]{Definition}
\newtheorem{example}[theorem]{Example}

\theoremstyle{remark}
\newtheorem{remark}[theorem]{Remark}

\numberwithin{equation}{section}

\definecolor{light}{gray}{.9}

\begin{document}

\title[\texorpdfstring{$\mc W$}{W}-algebras and integrable systems for classical Lie algebras]{Classical affine \texorpdfstring{$\mc W$}{W}-algebras
and the associated integrable Hamiltonian hierarchies
for classical Lie algebras}

\author{Alberto De Sole}
\address{Dipartimento di Matematica, Sapienza Universit\`a di Roma,
P.le Aldo Moro 2, 00185 Rome, Italy}
\email{desole@mat.uniroma1.it}
\urladdr{www1.mat.uniroma1.it/\$$\sim$\$desole}

\author{Victor G. Kac}
\address{Dept of Mathematics, MIT,
77 Massachusetts Avenue, Cambridge, MA 02139, USA}
\email{kac@math.mit.edu}

\author{Daniele Valeri}
\address{Yau Mathematical Sciences Center, Tsinghua University, 100084 Beijing, China}
\email{daniele@math.tsinghua.edu.cn}



\begin{abstract}
We prove that any classical affine $W$-algebra $\mc W(\mf g,f)$,
where $\mf g$ is a classical Lie algebra and $f$ is an arbitrary 
nilpotent element of $\mf g$, 
carries an integrable Hamiltonian hierarchy of Lax type equations.
This is based on the theories of generalized Adler type operators
and of generalized quasideterminants, which we develop in the paper.
Moreover, we show that under certain conditions,
the product of two generalized Adler type operators
is a Lax type operator.
We use this fact to construct a large number of integrable
Hamiltonian systems,
recovering, as a special case,
all KdV type hierarchies constructed by Drinfeld and Sokolov.
\end{abstract}


\maketitle

\tableofcontents

\section{Introduction}\label{sec:0}

In our paper \cite{DSKVnew} we proposed a new method of constructing integrable
(bi)Hamiltonian hierarchies of PDE's of Lax type.
It combined two well-known approaches.
The first one is the Gelfand-Dickey fractional powers of pseudodifferential
operators technique, based on the Lax pair method \cite{GD76,Dic03}.
The second one is the classical Hamiltonian reduction technique,
combined with the Zakharov-Shabat dressing method,
as developed by Drinfeld and Sokolov \cite{DS85}.

The central notion of the paper \cite{DSKVnew}
is that of a matrix pseudodifferential operator of Adler type
over a Poisson vertex algebra (PVA) $\mc V$,
introduced in \cite{AGD}.
It was derived there starting from Adler's formula \cite{Adl79} 
for the second Poisson structure for the $M$-th KdV hierarchy.

One of the important properties of an Adler type operator $L(\partial)$
over a PVA $\mc V$ is that it provides a hierarchy of compatible equations
of Lax type in terms of its fractional powers:
\begin{equation}\label{eq:intro1}
\frac{dL(\partial)}{dt_{n,k}}
=
[L(\partial)^{\frac nk}_+,L(\partial)]
\,\,,\,\,\,\,
n,k\in\mb Z_{\geq1}
\,,
\end{equation}
where the subscript $+$ stands, as usual, for the differential part
of the pseudodifferential operator $L(\partial)^{\frac nk}$.
Furthermore, $L(\partial)$ provides an infinite set of conserved densities
$h_{n,k}$ for the hierarchy \eqref{eq:intro1}, defined by
\begin{equation}\label{eq:intro2}
h_{n,k}
=
-\frac{k}{n}\res_{\partial}\tr L(\partial)^{\frac nk}
\,,\,\,
n,k\in\mb Z_{\geq1}
\,,
\end{equation}
so that \eqref{eq:intro1} is a hierarchy of Hamiltonian equations
with the corresponding Hamiltonian functionals $\tint h_{n,k}$
in involution.
Moreover, if $L(\partial)+\epsilon\id$ is of Adler type for every $\epsilon\in\mb F$
(the base field),
then these conserved densities satisfy the generalized 
Lenard-Magri scheme \cite{Mag78}
and \eqref{eq:intro1} is a bi-Hamiltonian hierarchy.
See \cite{AGD} and \cite{DSKVnew} for details.

The ancestors of all Adler type operators constructed in \cite{DSKVnew}
are given by the following family of first order $N\times N$ matrix
differential operators:
\begin{equation}\label{eq:intro3}
A_S(\partial)
=
\id_N\partial+\sum_{i,j=1}^N e_{ji}E_{ij} +S
\,\in\Mat_{N\times N}\mc V[\partial]
\,,
\end{equation}
where $\mc V$ is the algebra of differential polynomials in the generators 
$e_{ij}$, $i,j=1,\dots,N$, and $S\in\Mat_{N\times N}\mb F$.
It is easy to see that the operator \eqref{eq:intro3}
is of Adler type over the affine PVA $\mc V(\mf{gl}_N)=S(\mb F[\partial]\mf{gl}_N)$
with the $\lambda$-bracket
\begin{equation}\label{eq:intro4}
\{a_\lambda b\}
=
[a,b]+\tr(ab)\lambda+\tr(S[a,b])
\,,\,\,
a,b\in\mf{gl}_N
\,.
\end{equation}

The key observation in \cite{DSKVnew}
is that any generalized quasideterminant of an Adler type operator over a PVA $\mc V$
is again of Adler type.
In particular, the ``ancestors'' $A_S(\partial)$ produce a large number of ``descendent''
Adler type operators by taking generalized quasideterminants 
(for the theory of quasideterminants see \cite{GGRW05},
and for the definition of generalized quasideterminant see \cite{DSKVnew}).

Recall that to a reductive Lie algebra $\mf g$
and its nilpotent element $f$,
one associates a PVA $\mc W(\mf g,f)$,
which is a subquotient of the affine PVA $\mc V(\mf g)$
(see e.g. Section \ref{sec:2} of the present paper).
The key observation of our paper \cite{DSKV16b}
is that for any nilpotent element $f$ of $\mf{gl}_N$
a certain generalized quasideterminant of the differential operator $A_S(\partial)$
produces a pseudodifferential operator $L(\partial)$
whose coefficients are elements of $\mc W(\mf{gl}_N,f)$.
Since $L(\partial)$ is an operator of Adler type,
we thus obtain an integrable hierarchy of (bi)Hamiltonian
Lax type equations \eqref{eq:intro1} over the PVA $\mc W(\mf{gl}_N,f)$,
with the infinitely many conserved densities \eqref{eq:intro2}.

This gave, for $\mf g=\mf{gl}_N$, 
an affirmative answer to the longstanding problem
whether any $\mc W$-algebra $\mc W(\mf g,f)$
carries an integrable Hamiltonian hierarchy.

In the present paper we solve the same problem for all classical Lie algebras $\mf g$
and all their nilpotent elements $f$.
Drinfeld and Sokolov solved this problem in their seminal paper \cite{DS85}
for an arbitrary reductive $\mf g$ and its principal nilpotent element $f$
and, extending a series of previous papers 
\cite{BdGHM93,dGHM92,DF95,FHM93,FGMS95,FGMS96} etc.,
we solved this problem in \cite{DSKV13}
for arbitrary reductive $\mf g$ and its nilpotent element $f$ of ``semisimple type''.

The method used in the present paper is a development of our previous papers
\cite{DSKVnew} and \cite{DSKV16b}.
First, we construct a generalization of the ``ancestor'' operator $A_S(\partial)$
as follows.
Given a reductive Lie algebra $\mf g$ and its faithful representation $\varphi$
in a finite-dimensional vector space $V$,
we chose a basis $\{u_i\}_{i\in I}$ of $\mf g$
and let $U^i=\varphi(u^i)$,
where $\{u^i\}_{i\in I}$ is the dual basis of $\mf g$
with respect to the trace form of $V$.
We then define the generalized ``ancestor'' operator by
\begin{equation}\label{eq:intro5}
A_{S,V}(\partial)
=
\partial\id_V+\sum_{i\in I} u_iU^i +S
\,\in\mc V(\mf g)[\partial]\otimes\End V
\,,
\end{equation}
where $S=\varphi(s)$, $s\in\mf g$.
Our first main result is Theorem \ref{thm:main}
which states that a certain generalized quasideterminant 
of the matrix $A_{S,V}(\partial)$
produces a pseudodifferential operator $L(\partial)$
with coefficients in $\mc W(\mf g,f)$.
Also, from formula \eqref{eq:Lw} in this theorem,
one can read off the generators of the differential algebra $\mc W(\mf g,f)$.

In order to apply the ideas of \cite{DSKVnew}
to arbitrary $\mc W$-algebras we need to introduce the notion
of a generalized Adler type operator for an arbitrary pair $(\mf g,V)$.
It seems, however, 
that this is possible only for classical reductive Lie algebras $\mf g$
and their standard representations $V$,
i.e. for the linear Lie algebras $\mf g=\mf{gl}_N,\mf{sl}_N,\mf{so}_N,\mf{sp}_N$.

We define a \emph{generalized Adler type} operator
as an element $L(\partial)\in\mc V((\partial^{-1}))\otimes\End V$,
where $\mc V$ is a PVA and $V$ is an $N$-dimensional vector space,
satisfying the following identity for some constants $\alpha,\beta,\gamma\in\mb F$:
\begin{equation}\label{eq:intro6}
\begin{array}{l}
\displaystyle{
\vphantom{\Big(}
\{L(z)_\lambda L(w)\}
=
\alpha
(\id\otimes L(w+\lambda+\partial))(z-w-\lambda-\partial)^{-1}
(L^*(\lambda-z)\otimes\id)\Omega
} \\
\displaystyle{
\vphantom{\Big(}
-\alpha
\Omega\,\big(L(z)\otimes(z-w-\lambda-\partial)^{-1}L(w)\big)
} \\
\displaystyle{
\vphantom{\Big(}
-\beta
(\id\otimes L(w+\lambda+\partial))
\Omega^\dagger(z+w+\partial)^{-1}(L(z)\otimes\id)
} \\
\displaystyle{
\vphantom{\Big(}
+\beta
(L^*(\lambda-z)\otimes\id)\Omega^\dagger(z+w+\partial)^{-1}(\id\otimes L(w))
} \\
\displaystyle{
\vphantom{\Big(}
+\gamma\big(\id\otimes \big(L(w+\lambda+\partial)-L(w)\big)\big)
(\lambda+\partial)^{-1}
\big(\big(L^*(\lambda-z)-L(z)\big)\otimes\id\big)
\,.}
\end{array}
\end{equation}
Here $\Omega=\sum_{i,j=1}^NE_{ij}\otimes E_{ji}$,
$\Omega^\dagger=\sum_{i,j=1}^NE_{ij}^\dagger\otimes E_{ji}$,
and we assume, if $\beta\neq0$,
that $V$ carries a non-degenerate symmetric or skewsymmetric bilinear form,
and denote by $A^\dagger$ the adjoint of $A\in\End V$ with respect to this form.
Also, $*:\,\mc V((\partial^{-1}))\to \mc V((\partial^{-1}))$ denotes the formal adjoint of
a scalar pseudodifferential operator,
and it is extended to $\mc V((\partial^{-1}))\otimes\End V$ by acting only on the first factor.
The examples when $A_{S,V}(\partial)$ is a generalized Adler type operator
that we know of correspond to
$(\mf g,V)=\mf{gl}_N,\mf{sl}_N,\mf{so}_N,\mf{sp}_N$,
for the values of $\alpha,\beta,\gamma$ given in Table \eqref{table}.
In the case of $\mf{gl}_N$ we recover the notion of an Adler type operator
studied in \cite{AGD,DSKVnew}.

In the present paper we develop a theory of generalized Adler type operators
and their applications to the theory of integrable Hamiltonian systems
along the lines described above.
First, we prove Theorem \ref{thm:main2}
which states, in particular, that the generalized quasideterminant
of an operator of generalized Adler type with parameters $\alpha,\beta,\gamma$
is again of generalized Adler type with the same parameters.
Second, we prove Theorem \ref{thm:hn},
which states that if $L(\partial)$ is an operator of generalized Adler type,
then the $h_{n,k}\in\mc V$ defined by \eqref{eq:intro2}
are densities of Hamiltonian functionals in involution,
defining a compatible hierarchy of Lax type Hamiltonian equations
\begin{equation}\label{eq:intro7}
\frac{dL(\partial)}{dt_{n,k}}
=
\big[\alpha L(\partial)^{\frac nk}_+
-\beta \big(L(\partial)^{\frac nk}\big)^{\star\dagger}_+
,L(\partial)\big]
\,.
\end{equation}
We also prove Theorem \ref{thm:main-bi}
which states that, if $L(\partial)+\epsilon\id$ is of generalized Adler type
for every constant $\epsilon\in\mb F$,
then the densities $h_{n,k}$ satisfy the generalized Lenard-Magri scheme,
and \eqref{eq:intro7} is a bi-Hamiltonian hierarchy.

In the present paper
we discover some new ways of constructing 
integrable Hamiltonian hierarchies using
generalized Adler type operators.
First, in Section \ref{sec:5.4}
we classify all scalar constant coefficients pseudodifferential operators
of generalized Adler type.
But, what is most remarkable,
it turns out that, under some conditions,
products of generalized Adler type operators
produce compatible hierarchies of Lax type Hamiltonian equations,
similar to \eqref{eq:intro7}, see Theorem \ref{thm:hn-product}.

In Section \ref{sec:8}
we list the resulting integrable hierarchies of Hamiltonian equations
associated to all $\mc W$-algebras for classical Lie algebras
and their pairwise tensor products.
In particular, we recover all Drinfeld-Sokolov integrable KdV type hierarchies
that they attach to a classical affine Lie algebra
(including the twisted ones)
and a node on its Dynkin diagram \cite{DS85}.

In Section \ref{sec:9} we describe a number of explicit examples
of Lax operators.
First, we compute them in the case of a principal nilpotent
of all classical Lie algebras recovering thereby the operators $P$, $Q$ and $R$
of Drinfeld and Sokolov \cite{DS85}.
Second, we compute the Lax operators in the case of a minimal nilpotent
of all classical Lie algebras,
recovering thereby our results from \cite{DSKV16b}
in the case of $\mf{gl}_N$.
The next most interesting case is the distinguished nilpotent element
in $\mf{so}_{4n}$ corresponding to the partition $(2n+1,2n-1)$.
The Lax operator in this case is given by formula \eqref{L_distinguished}.

Finally, in Section \ref{sec:10}
we write down explicitly in many cases
the first non-trivial equations of the constructed integrable hierarchies.
First, we consider all cases with one unknown function.
All possibilities for the Lax operator are the Lax operator $L(\mf g)$
for $\mf{g}=\mf{sp}_2$ or $\mf g=\mf{so_3}$
multiplied by $1$ or by $\partial^{\pm1}$.
All these Lax operators produce the KdV equation,
with the following two exceptions:
$L=L(\mf{sp}_2)\partial$, which produces the Sawada-Kotera equation,
and $L=L(\mf{so}_3)$, which produces the Kaup-Kupershmidt equation.
In conclusion of the section we treat the case of $\mf g=\mf{sl}_N$ and $\mf{sp}_N$
and the minimal nilpotent $f$.
For $\mf g=\mf{sl}_N$ we recover the equations obtained in \cite{DSKV14a,DSKV15-cor},
which, after Dirac reduction, produce the $N$-component Yajima-Oikawa equation,
of which $N=3$ corresponds to the classical Y-O equation discovered in \cite{YO76}.
For $\mf g=\mf{sp}_N$ there are three choices for the Lax operator:
$L(\mf{sp}_N,f_{\text{min}})$,
$L(\mf{sp}_N,f_{\text{min}})\partial$ 
and $L(\mf{sp}_N,f_{\text{min}})\partial^{-1}$.
The corresponding first non-trivial equation for the first Lax operator
was studied in \cite{DSKV14a},
while for the last two Lax operators we find, after Dirac reduction, 
some apparently new integrable system in $N-1$ unknown functions.
In the case $N=4$ these equations are:
$$
\frac{du}{dt}
=
\frac34(-v_1v_2^{\prime\prime}+v_2v_1^{\prime\prime})
\,\,,\,\,\,\,
\frac{d}{dt}
\Big(\begin{array}{l} v_1 \\ v_2 \end{array}\Big)
=
\Big(\begin{array}{l} 
v_1^{\prime\prime\prime}-uv_1^\prime \\ 
v_2^{\prime\prime\prime}-uv_2^\prime 
\end{array}\Big)
$$
for $L=L(\mf{sp}_4,f_{\text{min}})\partial$, and
$$
\frac{du}{dt}
=
u^{\prime\prime\prime}-6uu^\prime
+\frac34(-v_1v_2^{\prime\prime}+v_2v_1^{\prime\prime})
\,\,,\,\,\,\,
\frac{d}{dt}
\Big(\begin{array}{l} v_1 \\ v_2 \end{array}\Big)
=
\Big(\begin{array}{l} 
v_1^{\prime\prime\prime}-3(uv_1)^\prime \\ 
v_2^{\prime\prime\prime}-3(uv_2)^\prime 
\end{array}\Big)
$$
for $L=L(\mf{sp}_4,f_{\text{min}})\partial^{-1}$.

The quantum finite analogue of an operator of Adler type
is an operator of Yangian type, introduced in \cite{DSKV17}.
The defining identity for such operators is the same as the identity
defining the Yangian of $\mf{gl}_N$, \cite{Mol07}.
Such operators were used in \cite{DSKV17}
to describe the quantum finite $W$-algebras associated to $\mf{gl}_N$,
in a way similar to the description of the classical affine $W$-algebras 
for $\mf{gl}_N$ using Adler type operators.
In our subsequent paper \cite{DSKV18}
we shall use the twisted Yangian identity, similar to the generalized Adler identity,
and related to the theory of twisted Yangians \cite{Mol07},
to describe the quantum finite $W$-algebras associated to all classical Lie algebras.

Throughout the paper the base field $\mb F$ is a a field of characteristic $0$.

\subsubsection*{Acknowledgments} 

We would like to thank L\'aszl\'o Feh\'er for posing the problem of constructing
the Lax operator for the classical affine $\mc W$-algebra $\mc W(\mf{so}_{4N},f)$,
where $f$ is a nilpotent element associated to the partition $(2n+1,2n-1)$.
The first author would like to acknowledge
the hospitality of MIT, where he was hosted during the spring semester of 2017.
The second author would like to acknowledge
the hospitality of the University of Rome La Sapienza
during his visit in Rome in January 2017 and 2018.
The third author is grateful to the University of Rome La Sapienza
for its hospitality during his several visits in 2016, 2017 and 2018.
All three authors are grateful to SISSA, where this work has started,
and IHES, where the paper has been completed,
for their kind hospitality during the summers of 2016 and 2017 respectively.
The first author is supported by National FIRB grant RBFR12RA9W,
National PRIN grant 2015ZWST2C,
and University grant C26A158K8A,
the second author is supported by an NSF grant
and the Italian grant C26V1735TJ,
and the third author is supported by an NSFC 
``Research Fund for International Young Scientists'' grant
and Tsinghua University startup grant.

\section{Preliminaries on pseudodifferential operators and Poisson vertex algebras}\label{sec:1}

Let $\mc V$ be a differential algebra, i.e. a unital commutative associative algebra
with a derivation $\partial$.
As usual, we denote by $\tint:\,\mc V\to\mc V/\partial\mc V$ the canonical quotient map
of vector spaces.

We denote by $\mc V((\partial^{-1}))$ the algebra of scalar pseudodifferential operators
with coefficients in $\mc V$.
Given $a(\partial)=\sum_{n=-\infty}^Na_n\partial^n$, $a_n\in\mc V$,
we denote 
by $a^*(\partial)=\sum_n(-\partial)^n\circ a_n\in\mc V((\partial^{-1}))$ its 
formal adjoint,
by $a(\partial)_+=\sum_{n=0}^Na_n\partial^n\in\mc V[\partial]$ 
its differential part,
by $a(\partial)_-=\sum_{n=-\infty}^{-1}a_n\partial^n\in\mc V[[\partial^{-1}]]$
its singular part,
and by $a(z)=\sum_na_nz^{n}\in\mc V((z^{-1}))$ its symbol.

The following notation will be used throughout the paper:
given $a(\partial)\in\mc V((\partial^{-1}))$ as above
and $b,c\in \mc V$, we let:
\begin{equation}\label{eq:notation}
a(z+x)\big(\big|_{x=\partial}b)c
=\sum_{n=-\infty}^Na_n((z+\partial)^nb)c\,\in\mc V
\,,
\end{equation}
where in the RHS we expand, for negative $n$, in the domain of large $z$.
For example, with this notation, we have
$a^*(z)=\big(\big|_{x=\partial}a(-z-x)\big)$.
Furthermore, for $a(z)\in\mc V((z^{-1}))$, we call the coefficient of $z^{-1}$ its residue, 
and we denote it by $\Res_za(z)$.

Let $M$ be a unital associative algebra.
By an $M$-valued pseudodifferential operator over $\mc V$
we mean an element $A(\partial)\in\mc V((\partial^{-1}))\otimes M$.
We shall omit the tensor product sign for such operators:
for $a(\partial)\in\mc V((\partial^{-1}))$ and $A\in M$, we let
$a(\partial)A$ be the corresponding monomial in $\mc V((\partial^{-1}))\otimes M$.
The symbol $A(z)$ of an $M$-valued pseudodifferential operator over $\mc V$
is defined as above,
and its formal adjoint is defined by taking formal adjoint of the first factor:
if $A(\partial)=a(\partial) A$, then $A^*(\partial)=a^*(\partial)A$.
Its symbol, with the notation \eqref{eq:notation}, is
\begin{equation}\label{20170406:eq1}
A^*(z)=\big(\big|_{x=\partial}A(-z-x)\big)
\,.
\end{equation}
\begin{lemma}\label{20170404:lem1}
Given $A(\partial),B(\partial)\in\mc V((\partial^{-1}))\otimes M$,
we have:
\begin{enumerate}[(a)]
\item
$(AB)(z)=A(z+\partial)B(z)$ (in the RHS $\partial$ is applied to the coefficients of $B(z)$);
\item
$(AB)^*(z)=\big(\big|_{x=\partial}A^*(z)\big)B^*(z+x)$.
\end{enumerate}
\end{lemma}
\begin{proof}
Part (a) follows from the definition of the product of pseudodifferential operators. 
For part (b), we have
$$
\begin{array}{l}
\displaystyle{
\vphantom{\Big(}
(AB)^*(z)
=
\big(\big|_{x=\partial} (AB)(-z-x)\big)
=
\big(\big|_{x=\partial} A(-z-x+\partial)B(-z-x)\big)
} \\
\displaystyle{
\vphantom{\Big(}
=
\big(\big|_{x_1=\partial} A(-z-x_1)\big)\big(\big|_{x_2=\partial}B(-z-x_1-x_2)\big)
=
\big(\big|_{x_1=\partial} A^*(z)\big)B^*(z+x_1)
\,.}
\end{array}
$$
\end{proof}

A Laurent series involving negative powers of $z\pm\partial$ or $z\pm\lambda$ is always considered 
to be expanded using geometric series expansion in the domain of large $z$,
and similarly for $w$.
On the other hand, for a series involving negative powers of $z\pm w$ we shall use the notation $\iota_z$
or $\iota_w$ to denote geometric series expansion in the domain of large $z$
or of large $w$ respectively.
For example, 
$\iota_z(z-w)^{-1}=\sum_{n\in\mb Z_+}z^{-n-1}w^n$.
For $a(z)\in\mc V((z^{-1}))$ as above, we have
\begin{equation}\label{eq:positive}
\Res_z a(z)\iota_z(z-w)^{-1}
=
a(w)_+
\,\,,\quad
\Res_z a(z)\iota_w(z-w)^{-1}
=
-a(w)_-
\,.
\end{equation}
\begin{lemma}\label{20170411:lem1}
Let $\dagger:\,M\to M$, $A\mapsto A^\dagger$ 
be an anti-involution of the associative algebra $M$, i.e. $(AB)^\dagger=B^\dagger A^\dagger$.
Then, for $A(\partial),B(\partial)\in\mc V((\partial^{-1}))\otimes M$,
we have
$(A^*(\partial))^\dagger(B^*(\partial))^\dagger=((BA)^*(\partial))^\dagger$.
\end{lemma}
\begin{proof}
Since $*$ is an anti-involution of $\mc V((\partial^{-1}))$
and $\dagger$ is an anti-involution of $M$, the claim follows.
\end{proof}
Note that in the present paper $M$ will be usually $\End V$,
where $V$ is a vector space.
In this case $A^*(\partial)$ is NOT the formal adjoint
of the matrix pseudodifferential operator $A(\partial)$,
which is, in fact, $A^*(\partial)^\dagger$.

Recall from \cite{BDSK09} that a $\lambda$-\emph{bracket} on the differential algebra $\mc V$ 
is a bilinear (over $\mb F$) map $\{\cdot\,_\lambda\,\cdot\}:\,\mc V\times\mc V\to\mc V[\lambda]$, 
satisfying the following
axioms ($a,b,c\in\mc V$):
\begin{enumerate}[(i)]
\item
sesquilinearity:
$\{\partial a_\lambda b\}=-\lambda\{a_\lambda b\}$,
$\{a_\lambda\partial b\}=(\lambda+\partial)\{a_\lambda b\}$;
\item
Leibniz rules (see notation \eqref{eq:notation}):
$$
\{a_\lambda bc\}=\{a_\lambda b\}c+\{a_\lambda c\}b
\,\,,\,\,\,\,
\{ab_\lambda c\}=\{a_{\lambda+x} c\} \big(\big|_{x=\partial}b\big)
+\{b_{\lambda+x} c\} \big(\big|_{x=\partial}a\big)
\,.
$$
\end{enumerate}
A \emph{Poisson vertex algebra} (PVA) $\lambda$-bracket on $\mc V$ 
satisfies the following additional axioms ($a,b,c\in\mc V$)
\begin{enumerate}[(i)]
\setcounter{enumi}{2}
\item
skewsymmetry:
$\{b_\lambda a\}=-\big(\big|_{x=\partial}\{a_{-\lambda-x} b\}\big)$;
\item
Jacobi identity:
$\{a_\lambda \{b_\mu c\}\}-\{b_\mu\{a_\lambda c\}\}
=\{\{a_\lambda b\}_{\lambda+\mu}c\}$.
\end{enumerate}
Recall that, if $\mc V$ is a Poisson vertex algebra,
then $\mc V/\partial\mc V$ carries a well defined Lie algebra structure given by
$\{\tint f,\tint g\}=\tint\{f_\lambda g\}|_{\lambda=0}$,
and we have a representation of the Lie algebra $\mc V/\partial\mc V$ on $\mc V$
given by $\{\tint f,g\}=\{f_\lambda g\}|_{\lambda=0}$.
A \emph{Hamiltonian equation} on $\mc V$ associated to a \emph{Hamiltonian functional} 
$\tint h\in\mc V/\partial\mc V$ is the evolution equation 
\begin{equation}\label{ham-eq}
\frac{du}{dt}=\{\tint h,u\}\,\,, \,\,\,\, u\in\mc V\,.
\end{equation}
An \emph{integral of motion} for the Hamiltonian equation \eqref{ham-eq}
is a local functional $\tint f\in\mc V/\partial\mc V$ such that $\{\tint h,\tint f\}=0$,
and two integrals of motion $\tint f,\tint g$ are \emph{in involution} if $\{\tint f,\tint g\}=0$.

Let $\mc V$ be a unital differential algebra with a $\lambda$-bracket $\{\cdot\,_\lambda\,\cdot\}$
and let $M$ be a unital associative algebra.
Given the $M$-valued pseudodifferential operators over $\mc V$
$A(\partial),B(\partial)\in\mc V((\partial^{-1}))\otimes M$,
we define the $\lambda$-bracket of their symbols 
$\{A(z)_\lambda B(w)\}$ as the element of $\mc V((z^{-1},w^{-1}))\otimes M\otimes M$
obtained by taking the $\lambda$-bracket of the first factors.
In other words, if $A(\partial)=a(\partial)A$ and $B(\partial)=b(\partial)B$,
with $a(\partial),b(\partial)\in\mc V((\partial^{-1}))$ and $A,B\in M$,
we have
\begin{equation}\label{20170404:eq2}
\{A(z)_\lambda B(w)\}=\{a(z)_\lambda b(w)\}\,A\otimes B
\,.
\end{equation}
(As usual, we omit the tensor product sign after the first factor.)
In the sequel we shall use the following properties of such $\lambda$-brackets
(which, in matrix element form, appeared in \cite[Eq.(2.12)--(2.15)]{DSKVnew}):
\begin{lemma}\label{20170404:lem2}
Let $A(\partial),B(\partial),C(\partial),A_\ell(\partial),B_\ell(\partial)\in\mc V((\partial^{-1}))\otimes M$,
$\ell=1,\dots,s$.
\begin{enumerate}[(a)]
\item
We have
$$
\begin{array}{l}
\displaystyle{
\vphantom{\Big(}
\{A(z)_\lambda (BC)(w)\}
} \\
\displaystyle{
\vphantom{\Big(}
=
\{A(z)_\lambda B(w+x)\}\big(1\otimes\big|_{x=\partial}C(w)\big)
+\big(1\otimes B(w+\lambda+\partial)\big)\{A(z)_\lambda C(w)\}
\,.}
\end{array}
$$
\item
We have
$$
\begin{array}{l}
\displaystyle{
\vphantom{\Big(}
\{(AB)(z)_\lambda C(w)\}
} \\
\displaystyle{
\vphantom{\Big(}
=
\{A(z+x)_{\lambda+x}C(w)\}\big(\big|_{x=\partial}B(z)\otimes1\big)
+\big(\big|_{y=\partial}A^*(\lambda-z)\otimes1\big)\{B(z)_{\lambda+y}C(w)\}
\,.}
\end{array}
$$
\item
We have
$$
\begin{array}{l}
\displaystyle{
\vphantom{\Big(}
\{A(z)_\lambda (B_1\dots B_s)(w)\}
} \\
\displaystyle{
\vphantom{\Big(}
=
\sum_{\ell=1}^{s}
\big(1\otimes (B_1\!\dots\! B_{\ell-1})(w\!+\!\lambda\!+\!\partial)\big)
\{A(z)_\lambda B_\ell(w\!+\!x)\}\big(1\otimes\big|_{x=\partial}(B_{\ell+1}\!\dots\! B_s)(w)\big)
\,.}
\end{array}
$$
\item
We have
$$
\begin{array}{l}
\displaystyle{
\vphantom{\Big(}
\{(A_1\dots A_s)(z)_\lambda B(w)\}
} \\
\displaystyle{
\vphantom{\Big(}
=\!
\sum_{\ell=1}^{s}
\!
\big(\big|_{y=\partial}\!(A_1\!\dots\! A_{\ell\!-\!1})^*(\lambda\!-\!z)\!\otimes\!1\big)
\{A_\ell(z\!+\!x)_{\lambda\!+\!x\!+\!y}B(w)\}
\big(\big|_{x=\partial}\!(A_{\ell\!+\!1}\!\dots\! A_s)(z)\!\otimes\!1\big)
\,.}
\end{array}
$$
\item
For every $n\in\mb Z_{\geq1}$, we have
$$
\begin{array}{l}
\displaystyle{
\vphantom{\Big(}
\{A(z)_\lambda B^n(w)\}
} \\
\displaystyle{
\vphantom{\Big(}
=
\sum_{\ell=0}^{n-1}
\big(1\otimes B^{n-\ell-1}(w+\lambda+\partial)\big)
\{A(z)_\lambda B(w+x)\}\big(1\otimes\big|_{x=\partial}B^\ell(w)\big)
\,.}
\end{array}
$$
\item
For every $n\in\mb Z_{\geq1}$, we have
$$
\begin{array}{l}
\displaystyle{
\vphantom{\Big(}
\{A^n(z)_\lambda B(w)\}
} \\
\displaystyle{
\vphantom{\Big(}
=
\sum_{\ell=0}^{n-1}
\big(\big|_{y=\partial}(A^\ell)^*(\lambda-z)\otimes1\big)
\{A(z+x)_{\lambda+x+y}B(w)\}
\big(\big|_{x=\partial}A^{n-1-\ell}(z)\otimes1\big)
\,.}
\end{array}
$$
\item
If $B(\partial)$ is invertible in $\mc V((\partial^{-1}))\otimes M$, then
$$
\begin{array}{l}
\displaystyle{
\vphantom{\Big(}
\{A(z)_\lambda B^{-1}(w)\}
} \\
\displaystyle{
\vphantom{\Big(}
=-\big(1\otimes B^{-1}(w+\lambda+\partial)\big)
\{A(z)_{\lambda}B(w+x)\}\big(1\otimes \big|_{x=\partial}B^{-1}(w)\big)
\,.}
\end{array}
$$
\item
If $A(\partial)$ is invertible in $\mc V((\partial^{-1}))\otimes M$, then
$$
\begin{array}{l}
\displaystyle{
\vphantom{\Big(}
\{A^{-1}(z)_\lambda B(w)\}
} \\
\displaystyle{
\vphantom{\Big(}
=
-\big(\big|_{y=\partial}(A^{-1})^*(\lambda-z)\otimes1\big)
\{A(z+x)_{\lambda+x+y}B(w)\}\big(\big|_{x=\partial}A^{-1}(z)\otimes1\big)
\,.}
\end{array}
$$
\item
If $B(\partial)$ is invertible in $\mc V((\partial^{-1}))\otimes M$ and $n\in\mb Z_{\leq-1}$, then
$$
\begin{array}{l}
\displaystyle{
\vphantom{\Big(}
\{A(z)_\lambda B^n(w)\}
} \\
\displaystyle{
\vphantom{\Big(}
=-\sum_{\ell=n}^{-1}
\big(1\otimes B^{n-1-\ell}(w+\lambda+\partial)\big)
\{A(z)_{\lambda}B(w+x)\}\big(1\otimes \big|_{x=\partial}B^\ell(w)\big)
\,.}
\end{array}
$$
\item
If $A(\partial)$ is invertible in $\mc V((\partial^{-1}))\otimes M$ and $n\in\mb Z_{\leq-1}$, then
$$
\begin{array}{l}
\displaystyle{
\vphantom{\Big(}
\{A^n(z)_\lambda B(w)\}
} \\
\displaystyle{
\vphantom{\Big(}
=
-\sum_{\ell=n}^{-1}
\big(\big|_{y=\partial}(A^\ell)^*(\lambda-z)\otimes1\big)
\{A(z+x)_{\lambda+x+y}B(w)\}
\big(\big|_{x=\partial}A^{n-1-\ell}(z)\otimes1\big)
\,.}
\end{array}
$$
\end{enumerate}
\end{lemma}
\begin{proof}
Formulas (a) and (b) are just the Leibniz rules (ii) above 
in the notation \eqref{20170404:eq2}.
Formulas (c) and (d) follow from (a) and (b) respectively,
by induction and Lemma \ref{20170404:lem1}.
Formulas (e) and (f) are a special case of (c) and (d).
Formulas (g) and (h) follow from (a) and (b) respectively,
using the identity $AA^{-1}=A^{-1}A=1$.
Finally, formulas (i) and (j) are obtained by (e)-(g) and (f)-(h) respectively.
\end{proof}

\section{Classical affine \texorpdfstring{$\mc W$}{W}-algebras}\label{sec:2}

\subsection{Construction of the classical affine $\mc W$-algebras}\label{sec:def}

We review here the construction of the classical affine $\mc W$-algebra
following \cite{DSKV13}.
Let $\mf g$ be a reductive Lie algebra with a non-degenerate symmetric 
invariant bilinear form $(\cdot\,|\,\cdot)$,
and let $\{f,2x,e\}\subset\mf g$ be an $\mf{sl}_2$-triple in $\mf g$.
We have the corresponding $\ad x$-eigenspace decomposition
\begin{equation}\label{eq:grading}
\mf g=\bigoplus_{k\in\frac{1}{2}\mb Z}\mf g_{k}
\,\,\text{ where }\,\,
\mf g_k=\big\{a\in\mf g\,\big|\,[x,a]=ka\big\}
\,,
\end{equation}
so that $f\in\mf g_{-1}$, $x\in\mf g_{0}$ and $e\in\mf g_{1}$.
We let $d$ be the \emph{depth} of the grading, i.e. the maximal eigenvalue of $\ad x$.
For a subspace $\mf a\subset\mf g$ we denote by $\mc V(\mf a)$
the algebra of differential polynomials over $\mf a$,
i.e. $\mc V(\mf a)=S(\mb F[\partial]\mf a)$.

Consider the pencil of affine Poisson vertex algebras $\mc V_\epsilon(\mf g,s)$,
where $\epsilon\in\mb F$ and $s\in\mf g_d$, defined as follows.
The underlying differential algebra is the algebra $\mc V(\mf g)$ of differential polynomials over $\mf g$,
and the PVA $\lambda$-bracket is given by 
\begin{equation}\label{lambda}
\{a_\lambda b\}_\epsilon=[a,b]+(a| b)\lambda+\epsilon(s|[a,b])
\quad \text{ for }\quad a,b\in\mf g\,,
\end{equation}
and extended to $\mc V(\mf g)$ by the sesquilinearity axioms and the Leibniz rules.

The $\mb F[\partial]$-submodule
$\mb F[\partial]\mf g_{\geq\frac12}\subset\mc V(\mf g)$ 
is a Lie conformal subalgebra of $\mc V_\epsilon(\mf g,s)$
with the $\lambda$-bracket $\{a_\lambda b\}_\epsilon=[a,b]$, $a,b\in\mf g_{\geq\frac12}$
(it is independent of $\epsilon$).
Consider the differential subalgebra
$\mc V(\mf g_{\leq\frac12})$ of $\mc V(\mf g)$,
and denote by $\rho:\,\mc V(\mf g)\twoheadrightarrow\mc V(\mf g_{\leq\frac12})$,
the differential algebra homomorphism defined on generators by
\begin{equation}\label{rho}
\rho(a)=\pi_{\leq\frac12}(a)+(f| a),
\qquad a\in\mf g\,,
\end{equation}
where $\pi_{\leq\frac12}:\,\mf g\to\mf g_{\leq\frac12}$ denotes 
the projection with kernel $\mf g_{\geq1}$.
We have a representation of the Lie conformal algebra $\mb F[\partial]\mf g_{\geq\frac12}$ 
on the differential subalgebra $\mc V(\mf g_{\leq\frac12})\subset\mc V(\mf g)$,
defined by
$$
a_\lambda(g)=\rho\{a_\lambda g\}_\epsilon
\quad\text{ for }\quad a\in\mf g_{\geq\frac12}\,,\,\,g\in\mc V(\mf g_{\leq\frac12})
$$
(note that the RHS is independent of $\epsilon$ since, by assumption, $s\in\mf g_d$).

The \emph{classical} $\mc W$-\emph{algebra} $\mc W_\epsilon(\mf g,f,s)$ is, by definition,
the differential algebra
\begin{equation}\label{20120511:eq2}
\mc W=\mc W(\mf g,f)
=\big\{w\in\mc V(\mf g_{\leq\frac12})\,\big|\,\rho\{a_\lambda w\}_\epsilon=0\,
\text{ for all }a\in\mf g_{\geq\frac12}\}\,,
\end{equation}
endowed with the following pencil of PVA $\lambda$-brackets
\cite[Lemma 3.2]{DSKV13}
\begin{equation}\label{20120511:eq3}
\{v_\lambda w\}^{\mc W}_{\epsilon}=\rho\{v_\lambda w\}_\epsilon,
\qquad v,w\in\mc W\,.
\end{equation}
With a slight abuse of notation,
we shall denote by $\mc W(\mf g,f)$ also the $\mc W$-algebra $\mc W_\epsilon(\mf g,f,s)$
for $\epsilon=0$ (or, equivalently, $s=0$).
\subsection{Structure Theorem for classical affine $\mc W$-algebras}
\label{sec:3.2}

Fix a subspace $U\subset\mf g$ complementary to $[f,\mf g]$,
which is compatible with the grading \eqref{eq:grading}.
For example, we could take $U=\mf g^e$, as we did in \cite{DSKV13} and \cite{DSKV16a},
or a different, more convenient, choice for $U$ as we did for $\mf g=\mf{gl}_N$ in \cite{DSKV16b}.
Since $\ad f:\,\mf g_{j}\to\mf g_{j-1}$ is surjective for $j\leq\frac12$, 
we have $\mf g_{\leq-\frac12}\subset[f,\mf g]$.
In particular, we have the direct sum decomposition
\begin{equation}\label{eq:U}
\mf g_{\geq-\frac12}=[f,\mf g_{\geq\frac12}]\oplus U\,.
\end{equation}
Note that, by the non-degeneracy of $(\cdot\,|\,\cdot)$, the orthocomplement to $[f,\mf g]$
is $\mf g^f$, the centralizer of $f$ in $\mf g$.
Hence, the direct sum decomposition dual to \eqref{eq:U} is
\begin{equation}\label{eq:Uperp}
\mf g_{\leq\frac12}=U^\perp\oplus\mf g^f\,.
\end{equation}
As a consequence of \eqref{eq:Uperp}
we have the decomposition in a direct sum of subspaces
\begin{equation}\label{eq:decomp}
\mc V(\mf g_{\leq\frac12})=\mc V(\mf g^f)\oplus\langle U^\perp\rangle\,,
\end{equation}
where $\langle U^\perp\rangle$
is the differential algebra ideal of $\mc V(\mf g_{\leq\frac12})$ generated by $U^\perp$.
Let $\pi_{\mf g^f}:\,\mc V(\mf g_{\leq\frac12})\twoheadrightarrow\mc V(\mf g^f)$
be the canonical quotient map, with kernel $\langle U^\perp\rangle$.

\begin{theorem}[{\cite[Cor.4.1]{DSKV16a}, \cite[Rem.3.4]{DSKV16b}}]
\label{thm:structure-W}
The map $\pi_{\mf g^f}$ restricts to a differential algebra isomorphism
$$
\pi:=\pi_{\mf g^f}|_{\mc W}:\,\mc W\,\stackrel{\sim}{\longrightarrow}\,\mc V(\mf g^f)
\,,
$$
hence we have the inverse differential algebra isomorphism
$$
w:\,\mc V(\mf g^f)\,\stackrel{\sim}{\longrightarrow}\,\mc W
\,,
$$
which associates to every element $q\in\mf g^f$ the (unique) element $w(q)\in\mc W$
of the form $w(q)=q+r$, with $r\in\langle U^\perp\rangle$.
\end{theorem}

\section{The pseudodifferential operator \texorpdfstring{$L_\epsilon(\partial)$}{L_eps(d)} for the \texorpdfstring{$\mc W$}{W}-algebra 
\texorpdfstring{$\mc W_\epsilon(\mf g,f,s)$}{W_eps(g,f,s)} associated to a \texorpdfstring{$\mf g$}{g}-module
\texorpdfstring{$V$}{V}}\label{sec:3}

Let $\varphi:\mf g\to\End V$ be a faithful representation of $\mf g$
on an $N$-dimensional vector space $V$.
Throughout the paper we shall often use the following convention:
we denote by lowercase Latin letters elements of the Lie algebra $\mf g$,
and by uppercase letters the corresponding elements of $\End V$.
For example, $F=\varphi(f)$ is a nilpotent endomorphism of $V$
and we denote by $p=(p_1\geq p_2\geq\dots\geq p_r>0)$
the corresponding partition of $N$.
Moreover, $X=\varphi(x)$ is a semisimple endomorphism of $V$ with half-integer eigenvalues.
The corresponding $X$-eigenspace decomposition of $V$ is
\begin{equation}\label{eq:grading_V}
V=\bigoplus_{k\in\frac12\mb Z}V[k]
\,,
\end{equation}
with largest eigenvalue $\frac D2$, where $D=p_1-1$.
The eigenspace associated to the largest eigenvalue
has dimension $\dim V[\frac D2]=r_1$, the multiplicity of $p_1$ in the partition $p$.
We also have the corresponding $\ad X$-eigenspace 
decomposition of $\End V$:
\begin{equation}\label{eq:grading_EndV}
\End V=\bigoplus_{k\in\frac12\mb Z}(\End V)[k]
\,,
\end{equation}
which has largest eigenvalue $D$.
\begin{lemma}\label{20170317:lem1}
\begin{enumerate}[(a)]
\item
For every $k\in\frac12\mb Z$ s.t. $- D\leq k\leq D$, we have a canonical isomorphism
$$
(\End V)[k]
\simeq
\bigoplus_{j\in\frac12\mb Z}\Hom(V[j],V[j+k])
\,,
$$
where the direct sum is over $j$ such that $-\frac D2\leq j,j+k\leq\frac D2$.
\item
In particular, we have a canonical isomorphism
$$
(\End V)[D]
\simeq
\Hom(V[-\frac D2],V[\frac D2])
\,.
$$
\item
Identifying $V[-\frac D2]$ with $V[\frac D2]$ 
via the isomorphism $F^D:\,V[\frac D2]\stackrel{\sim}{\rightarrow}V[-\frac D2]$,
we get the corresponding isomorphisms
$$
(\End V)[D]
\simeq\End(V[\frac D2])
\,\,,\,\,\,
A\mapsto A F^D\,\big|_{V[\frac D2]}
\,,
$$
and
$$
(\End V)[D]
\simeq\End(V[-\frac D2])
\,\,,\,\,\,
A\mapsto F^D A\,\big|_{V[-\frac D2]}
\,.
$$
\end{enumerate}
\end{lemma}
\begin{proof}
Denote by $r^k_j:\,(\End V)[k]\to\Hom(V[j],V[j+k])$ the restriction map.
Clearly, $r^k:=\oplus_j r^k_j:\,(\End V)[k]\to\bigoplus_j\Hom(V[j],V[j+k])$
is injective for every $k$,
and by dimension counting it is easy to see that it is surjective too.
This proves part (a).
Part (b) is the special case $k=D$ of (a),
and part (c) is an obvious consequence of (b).
\end{proof}
Recall that the trace form of the representation $V$ 
is, by definition,
\begin{equation}\label{20170317:eq1}
(a|b)=\tr_V(\varphi(a)\varphi(b))\,,
\qquad
a,b\in\mf g
\,,
\end{equation}
and we assume that it is non-degenerate.
Let $\{u_i\}_{i\in I}$ be a basis of $\mf g$ 
compatible with the $\ad x$-eigenspace decomposition \eqref{eq:grading},
i.e. $I=\sqcup_k I_k$ where $\{u_i\}_{i\in I_k}$ is a basis of $\mf g_k$.
We also denote $I_{\leq\frac12}=\sqcup_{k\leq\frac12}I_k$,
and similarly for $I_{\leq0}$, $I_{\geq\frac12}$, etc.
Moreover, we assume that $\{u_i\}_{i\in I}$
contains a basis $\{u_i\}_{i\in I_f}$ of $\mf g^f$.
Let $\{u^i\}_{i\in I}$ be the basis of $\mf g$ dual to $\{u_i\}_{i\in I}$ with respect 
to the form \eqref{20170317:eq1},
i.e. $(u_i|u^j)=\delta_{i,j}$.
According to our convention,
we denote by $U_i$ and $U^i$, $i\in I$, the corresponding endomorphisms of $V$.

Associated to the element $s\in\mf g_d$ we have the element $S=\varphi(s)\in(\End V)[d]$.
Let $T\in(\End V)[D]$.
If $D=d$, we can take $T=S$, but in general $d\leq D$, 
so it is not always possible to let $T$ and $S$ be the same endomorphism.
Consider the canonical decomposition
$T=I J$, 
where 
\begin{equation}\label{IJ}
J=T:V\twoheadrightarrow\im T
\,\,\text{ and }\,\,
I:\im T\hookrightarrow V
\,\,\text{ is the inclusion map}.
\end{equation}
Clearly, 
$\im T\subset V[\frac{D}{2}]$
and $\bigoplus_{k>-\frac D2} V[k]\subset\ker T$.
If $T$ is of maximal rank ($=r_1$), then both inclusions become equalities:
$\im T=V[\frac{D}{2}]$, $\ker T=\bigoplus_{k>-\frac D2} V[k]$.
We shall assume that $T$ satisfies the following condition:
\begin{equation}\label{20170317:eq2}
V=\ker T\oplus F^D(\im T)
\,.
\end{equation}
This is equivalent to say that the endomorphisms 
$T_+=T F^D\big|_{V[\frac D2]}\,\in\End(V[\frac D2])$
and/or $T_-=F^D T \big|_{V[-\frac D2]}\,\in\End(V[-\frac D2])$
(cf. Lemma \ref{20170317:lem1}(c)) are such that $\ker T_{\pm}\cap\im T_{\pm}=0$.
Of course, if $\rk(T)=r_1$, then $T_\pm$ are invertible and condition \eqref{20170317:eq2}
automatically holds.

Consider the following $\End V$-valued differential operator
with coefficients in $\mc V(\mf g)$ (depending on the parameter $\epsilon\in\mb F$):
\begin{equation}\label{eq:A}
A_\epsilon(\partial)=\partial\id_V+\sum_{i\in I}u_i U^i+\epsilon S
\,\,
\in\mc V(\mf g)[\partial]\otimes\End(V)
\,.
\end{equation}
Here and further, we drop the tensor product sign when writing an element of $\mc V\otimes\End V$.
If we apply the map $\rho$ ($=\rho\otimes 1$), defined by \eqref{rho}, to $A(\partial)$,
we get
$$
\rho(A_\epsilon(\partial))=\partial\id_V+F+\sum_{i\in I_{\leq\frac12}}u_i U^i+\epsilon S
\,\,\in\mc V(\mf g_{\leq\frac12})[\partial]\otimes\End V
\,.
$$
We shall consider its $(I,J)$-quasideterminant \cite{DSKVnew},
namely
\begin{equation}\label{eq:L}
L_\epsilon(\partial)=L_\epsilon(\mf g,f,s,V,T)(\partial)
:=
\Big(J\big(\partial\id_V+F+\sum_{i\in I_{\leq\frac12}}u_iU^i+\epsilon S\big)^{-1} I\Big)^{-1}
\,.
\end{equation}

Associated to the basis $\{u_i\}_{i\in I}$ we have the subspace
$$
U
=\Span\{u^i\mid i\in I_{f}\}
\,\subset\mf g_{\geq-\frac12}
\,,
$$
which is complementary to $[f,\mf g]$ in $\mf g$,
and to $[f,\mf g_{\geq\frac12}]$ in $\mf g_{\geq-\frac12}$,
and its orthocomplement in $\mf g_{\leq\frac12}$,
$$
U^\perp=\Span\{u_i\mid i\in I_{\leq\frac12}\setminus I_f\}\,\subset\mf g_{\leq\frac12}
\,,
$$
which is complementary to $\mf g^f$ in $\mf g_{\leq\frac12}$.
Recall that, by Theorem \ref{thm:structure-W},
we have the corresponding differential algebra isomorphism 
$w:\,\mc V(\mf g^f)\stackrel{\sim}{\rightarrow}\mc W(\mf g,f)$,
and let us denote by $w_i:=w(u_i),\,i\in I_f$,
the corresponding free generators of the the $\mc W$-algebra (as a differential algebra).
\begin{theorem}\label{thm:main}
$L_\epsilon(\partial)$ is well defined and 
\begin{equation}\label{eq:Lw}
L_\epsilon(\partial)
=
\Big(J\big(\partial\id_V+F+\sum_{i\in I_{f}}w_i U^i+\epsilon S\big)^{-1}I\Big)^{-1}
\,\in\mc W(\mf g,f)((\partial^{-1}))\otimes\End(\im T)
\,.
\end{equation}
\end{theorem}
The above theorem consists of three statements.
First, it claims that $L_\epsilon(\partial)$ is well defined, i.e. both inverses in formula
\eqref{eq:L} can be carried out in the algebra
of pseudodifferential operators with coefficients in $\mc V(\mf g_{\leq\frac12})$.
This is the content of Lemma \ref{20170303:prop1} below.
Next, it claims that, in fact, the coefficients of $L_\epsilon(\partial)$
lie in the $\mc W$-algebra $\mc W(\mf g,f)$,
which is proved in Lemma \ref{prop:L2_general} below.
Finally, it gives a formula, equation \eqref{eq:Lw}, 
for $L_\epsilon(\partial)$ in terms of the generators $w_i,\,i\in I_f$,
of the $\mc W$-algebra $\mc W(\mf g,f)$.
This formula is proved in Lemma \ref{lemma:last}.
\begin{remark}\label{20170317:rem}
Note that $U^\perp$, and hence $U$, depend on the choice of the basis elements 
$u_i$, for  $i\in I_{\leq\frac12}\backslash I_f$.
As a consequence, the map $w:\,\mc V(\mf g^f)\to\mc W(\mf g,f)$,
as well as the generators $w_i$, $i\in I_f$,
change when we change the basis elements $u_i$, $i\in I_{\leq\frac12}\backslash I_f$.
However, as a consequence of Theorem \ref{thm:main},
the RHS of formula \eqref{eq:Lw} is independent of the choice of the basis of $\mf g$.
\end{remark}
\begin{lemma}\label{20170303:prop1}
\begin{enumerate}[(a)]
\item $\rho(A_\epsilon(\partial))$ 
is invertible in $\mc V(\mf g_{\leq\frac12})((\partial^{-1}))\otimes\End(V)$.
\item $J(\rho(A_\epsilon(\partial)))^{-1} I$ 
is invertible in $\mc V(\mf g_{\leq\frac12})((\partial^{-1}))\otimes\End(\im T)$.
\end{enumerate}
\end{lemma}
\begin{proof}
The differential operator $\rho(A_\epsilon(\partial))$
is of order one with leading coefficient $\id_V$.
Hence it is invertible in the algebra $V(\mf g_{\leq\frac12})((\partial^{-1}))\otimes\End(V)$,
and its inverse can be computed by geometric series expansion:
\begin{equation}\label{eq:thm1-pr1}
\rho(A_\epsilon(\partial))^{-1}
=
\sum_{\ell=0}^\infty (-1)^\ell \partial^{-1}
\Big(\big(F+\sum_{i\in I_{\leq\frac12}}u_iU^i+\epsilon S\big)\partial^{-1}\Big)^\ell
\,.
\end{equation}
This proves part (a). 
Recall that $F\in(\End V)[-1]$, $U^i\in(\End V)[\geq-\frac12]$ for $i\in I_{\leq\frac12}$
and $S\in(\End V)[d]\subset(\End V)[\geq-\frac12]$.
Recall also that 
$\im I=\im T\subset V[\frac{D}{2}]$
and
$\ker J=\ker T\supset V[>-\frac D2]$.
Hence, by keeping track of the $X$-eigenvalues, we immediately get that
$$
J\partial^{-1}\Big(\big(F+\sum_{i\in I_{\leq\frac12}}u_iU^i+\epsilon S\big)\partial^{-1}\Big)^\ell I
\left\{\begin{array}{l}
\vphantom{\Big(}
=0\,\,\text{ if }\,\, \ell<D\,, \\
\vphantom{\Big(}
=JF^DI\,\partial^{-p_1}\,\,\text{ if }\,\, \ell=D\,, \\
\vphantom{\Big(}
\!\!\in\!\!\mc V(\mf g_{\leq\frac12})[[\partial^{-1}]]\partial^{\!\!-\!p_1\!-\!1}\!\!\otimes\End(\im T)
\text{ if } \ell>D.
\end{array}\right.
$$
It follows from \eqref{eq:thm1-pr1} that
\begin{equation}\label{eq1:20170303}
J\rho(A_\epsilon(\partial))^{-1} I=(-1)^{D}J F^{D} I\,\partial^{-p_1}\,+\,\text{ lower order terms}\,.
\end{equation}
By the assumption \eqref{20170317:eq2} on $T$,
we have that $JF^DI=TF^D|_{\im T}$ is an invertible endomorphism 
of $\im T$.
Hence, \eqref{eq1:20170303} can be inverted by geometric series expansion,
proving (b).
\end{proof} 
\begin{lemma}[{\cite[{Lem.3.1(b)}]{DSKV13}}]\label{lem:app2}
Consider the pencil of affine Poisson vertex algebras $\mc V=\mc V_\epsilon(\mf g,s)$
with $s\in\mf g_d$.
For $a\in\mf g_{\geq\frac12}$ and $g\in\mc V(\mf g)$, we have
$\rho\{a_\lambda \rho(g)\}_\epsilon=\rho\{a_\lambda g\}_\epsilon$.
\end{lemma}
\begin{lemma}\label{lem:app1a}
For $a\in\mf g$, 
we have
$\displaystyle{\sum_{i\in I}[a,u_i] U^i=\sum_{i\in I}u_i [U^i,\varphi(a)]
\in\mc V(\mf g)\!\otimes\!\End V}$.
\end{lemma}
\begin{proof}
By the completeness relation for $\mf g$ and the invariance 
of the trace form $(\cdot\,|\,\cdot)$, we have
$$
\sum_{i\in I}[a,u_i] U^i
=
\sum_{i,j\in I}([a,u_i]|u^j)u_j\varphi(u^i)
=
\sum_{i,j\in I}([u^j,a]|u_i)u_j\varphi(u^i)
=
\sum_{j\in I}u_j[U^j,\varphi(a)]
\,.
$$
\end{proof}
\begin{lemma}\label{lem:app1b}
For $a\in\mf g$, we have
$$
\{a_\lambda A_\epsilon(z)\}_\epsilon=A_\epsilon(z+\lambda)\varphi(a)-\varphi(a)A_\epsilon(z)
\,,
$$
where $A_\epsilon(z)$ is the symbol of the pseudodifferential operator 
$A_\epsilon(\partial)\in\mc V(\mf g)[\lambda]\otimes\End(V)$ defined in \eqref{eq:A}.
\end{lemma}
\begin{proof}
By definition \eqref{lambda}
of the $\lambda$-bracket in $\mc V(\mf g)$, we have
$$
\{a_\lambda A_\epsilon(z)\}_\epsilon
=
\sum_{i\in I}\{a_\lambda u_i\}_\epsilon U^i
=
\sum_{i\in I}[a,u_i]U^i+\lambda\varphi(a)+\epsilon\varphi([s,a])
\,.
$$
On the other hand, by \eqref{eq:A} we have
$$
A_\epsilon(z+\lambda)\varphi(a)-\varphi(a)A_\epsilon(z)
=
\lambda\varphi(a)+\sum_{i\in I}u_i[U^i,\varphi(a)]+\epsilon[S,\varphi(a)]
\,.
$$
The claim follows by Lemma \ref{lem:app1a} (recalling that $S=\varphi(s)$).
\end{proof}

\begin{lemma}\label{prop:L2_general}
We have
$\rho\{a_\lambda L_\epsilon^{-1}(z)\}_\epsilon=0$
for every $a\in\mf g_{\geq\frac12}$,
where $L_\epsilon^{-1}(z)$ is the symbol 
of the pseudodifferential operator $L_\epsilon^{-1}(\partial)=J\circ(\rho(A_\epsilon(\partial)))^{-1}\circ I$.
Equivalently, $L_\epsilon^{-1}(\partial)$, and hence $L_\epsilon(\partial)$,
has coefficients in the $\mc W$-algebra $\mc W(\mf g,f)$.
\end{lemma}
\begin{proof}
By the definition \eqref{eq:L} of $L_\epsilon(\partial)$
and recalling that the map $\rho$ is a differential algebra homomorphism, we have
\begin{equation}\label{eq:app3_general}
\rho\{a_\lambda L_\epsilon^{-1}(z)\}_\epsilon
=
\rho\{a_\lambda J(\rho A_\epsilon)^{-1}(z)I\}_\epsilon
=
\rho J\{a_\lambda \rho(A_\epsilon^{-1}(z))\}_\epsilon I
\,.
\end{equation}
We then apply Lemma \ref{lem:app2} to rewrite the RHS of \eqref{eq:app3_general} as
$$
\rho J\{a_\lambda (A_\epsilon^{-1}(z))\}_\epsilon I
\,,
$$
we use Lemma \ref{20170404:lem2}(g) to rewrite it as
$$
-\rho J A_\epsilon^{-1}(z+\lambda+\partial)\{a_\lambda A_\epsilon(z+x)\}_\epsilon
\big(\big|_{x=\partial}A_\epsilon^{-1}(z)\big) I
\,,
$$
and finally we use Lemma \ref{lem:app1b} to equal it to
$$
-\rho( J \varphi(a) A_\epsilon^{-1}(z) I )
+\rho( J A_\epsilon^{-1}(z+\lambda)\varphi(a) I )
\,.
$$
To conclude we observe that, since $\varphi(a)\in(\End V)[\geq\frac12]$,
both $J\varphi(a)$ and $\varphi(a) I$ vanish.
\end{proof}
\begin{lemma}\label{lemma:last}
Equation \eqref{eq:Lw} holds.
\end{lemma}
\begin{proof}
By Lemma \ref{prop:L2_general} we have 
$L_\epsilon(\partial)\in\mc W(\mf g,f)((\partial^{-1}))\otimes\im T$. Hence,
by Theorem \ref{thm:structure-W} we have 
$L_\epsilon(\partial)=(w\circ\pi_{\mf g^f})L_\epsilon(\partial)$. 
The claim follows by the definition 
of $U$ and the fact that $w$ and $\pi_{\mf g^f}$ are differential algebra homomorphisms.
\end{proof}

\section{Generalized Adler identity}\label{sec:4}

\subsection{Some facts from linear algebra}\label{sec:4.1}

Given a vector space $V$ of dimension $N$, 
we denote by $\Omega_V\in\End V\otimes\End V$
the permutation map:
\begin{equation}\label{Omega}
\Omega_V(v_1\otimes v_2)=v_2\otimes v_1
\,\,\text{ for all }\,\,v_1,v_2\in V
\,.
\end{equation}
Often, if confusion may not arise, we shall drop 
the index $V$ and we shall denote $\Omega=\Omega_V$.
We shall also sometimes write $\Omega=\Omega^\prime\otimes\Omega^{\prime\prime}$
to denote, as usual, a sum of monomials in $\End V\otimes\End V$.
In fact, we can write an explicit formula:
$\Omega=\sum_{i,j=1}^NE_{ij}\otimes E_{ji}$,
where $E_{ij}$ is the ``standard'' basis of $\End V$ consisting of elementary matrices 
w.r.t. any basis of $V$
(obviously, $\Omega$ does not depend on the choice of this basis).
By the completeness relation, we have
\begin{equation}\label{20170404:eq1}
\tr(\Omega^\prime A)\Omega^{\prime\prime}=A\,.
\end{equation}
As an immediate consequence of \eqref{Omega} we have:
\begin{lemma}\label{20170322:lem1}
Let $U$ and $V$ be vector spaces,
and let $A,B\in U\to V$ be linear maps.
We have
\begin{equation}\label{eq:permut}
\Omega_V(A\otimes B)=(B\otimes A)\Omega_U
\end{equation}
\end{lemma}
\begin{proof}
It is an obvious consequence of \eqref{Omega}.
\end{proof}

Let $\langle\cdot\,|\,\cdot\rangle$ be a non-degenerate symmetric or skewsymmetric
bilinear form on $V$:
\begin{equation}\label{eq:epsilon}
\langle u|v\rangle=\epsilon\langle v|u\rangle
\,\,,\,\,\,\,u,v\in V
\,\,,\,\,\text{ where }\epsilon\in\{\pm1\}
\,.
\end{equation}
Let $\{v_k\}_{k=1}^N$ be a basis of $V$ and let $\{v^k\}_{k=1}^N$ be the dual basis 
with respect to $\langle\cdot\,|\,\cdot\rangle$:
$$
\langle v^k|v_h\rangle=\epsilon\langle v_k|v^h\rangle=\delta_{h,k}
\,\,\text{ for all }h,k=1,\dots,N
\,.
$$
By the symmetry of the inner product, we have
\begin{equation}\label{eq:epsbasis}
\sum_{k=1}^Nv^k\otimes v_k=\epsilon\sum_{k=1}^Nv_k\otimes v^k
\,.
\end{equation}
Recall also that we have the following completeness relations:
\begin{equation}\label{eq:completeness}
\sum_{k=1}^N \langle v^k|v\rangle v_k
=
\sum_{k=1}^N \langle v|v_k\rangle v^k
=v
\,\,\text{ for all } v\in V
\,.
\end{equation}

For $A\in\End V$ we denote by $A^\dagger$ its adjoint with respect to $\langle\cdot\,|\,\cdot\rangle$:
$$
\langle u|A^\dagger(v)\rangle
=
\langle A(u)|v\rangle
\,\,\text{ for all }\,\, u,v\in V
\,.
$$
It immediately follows from the completeness relation \eqref{eq:completeness}
and the definition of adjoint that, 
for every $A\in\End V$, we have
\begin{equation}\label{eq:abasis}
\sum_{k=1}^NA(v^k)\otimes v_k=\sum_{k=1}^Nv^k\otimes A^\dagger(v_k)
\,.
\end{equation}

We shall denote by $\Omega_V^\dagger$ (or simply $\Omega^\dagger$)
the element of $\End V\otimes\End V$ obtained taking the adjoint on the first factor of $\Omega$:
\begin{equation}\label{Omega-tau}
\Omega^\dagger
=
(\Omega^\prime)^\dagger\otimes\Omega^{\prime\prime}
\,\in\End V\otimes\End V
\,.
\end{equation}
The following Lemma gives an explicit formula for the action of $\Omega^\dagger$ on $V\otimes V$:
\begin{lemma}\label{20170322:lem2}
For every $v_1,v_2\in V$, we have
\begin{equation}\label{eq:omega+}
\Omega^\dagger(v_1\otimes v_2)=\langle v_1|v_2\rangle\sum_{k=1}^N v^k\otimes v_k
\,.
\end{equation}
\end{lemma}
\begin{proof}
It suffices to take inner product of both sides of \eqref{eq:omega+} with $v_3\otimes v_4$.
For the LHS we have
$$
\begin{array}{l}
\displaystyle{
\vphantom{\Big(}
\langle v_3\otimes v_4|\Omega^\dagger(v_1\otimes v_2)\rangle
=
\langle v_3|(\Omega^\prime)^\dagger(v_1)\rangle\langle v_4|\Omega^{\prime\prime}(v_2)\rangle
=
\langle \Omega^\prime(v_3)|v_1\rangle\langle v_4|\Omega^{\prime\prime}(v_2)\rangle
} \\
\displaystyle{
\vphantom{\Big(}
=
\langle v_2|v_1\rangle\langle v_4|v_3\rangle
=
\langle v_1|v_2\rangle\langle v_3|v_4\rangle
\,.}
\end{array}
$$
For the third equality we used \eqref{eq:permut}
and for the last equality we used \eqref{eq:epsilon} (and the fact that $\epsilon^2=1$).
Doing the same computation with the RHS of \eqref{eq:omega+} we get,
by \eqref{eq:completeness},
$$
\langle v_1|v_2\rangle
\sum_{k=1}^N\langle v_3\otimes v_4|v^k\otimes v_k\rangle
=
\langle v_1|v_2\rangle
\sum_{k=1}^N\langle v_3|v^k\rangle\langle v_4|v_k\rangle
=
\langle v_1|v_2\rangle
\langle v_3|v_4\rangle
\,.
$$
\end{proof}
\begin{lemma}\label{20170322:lem3}
For every $A\in\End V$, we have
\begin{equation}\label{eq:permut-tau}
(A\otimes\id)\Omega^\dagger=(\id\otimes A^\dagger)\Omega^\dagger
\,\,,\,\,\,\,
\Omega^\dagger(A\otimes\id)=\Omega^\dagger(\id\otimes A^\dagger)
\,.
\end{equation}
\end{lemma}
\begin{proof}
We can combine both identities in \eqref{eq:permut-tau} in
\begin{equation}\label{eq:permut-tau2}
(A\otimes\id)\Omega^\dagger(B\otimes\id)=(\id\otimes A^\dagger)\Omega^\dagger(\id\otimes B^\dagger)
\,\,\text{ for }\,A,B\in\End V\,.
\end{equation}
If we apply the LHS of \eqref{eq:permut-tau2} to $v_1\otimes v_2$ we get
$$
\langle B(v_1)|v_2\rangle\sum_{k=1}^NA(v^k)\otimes v_k
\,,
$$
while if we apply the RHS of \eqref{eq:permut-tau2} to $v_1\otimes v_2$ we get
$$
\langle v_1|B^\dagger(v_2)\rangle\sum_{k=1}^Nv^k\otimes A^\dagger(v_k)
\,.
$$
Hence, \eqref{eq:permut-tau2} follows from the definition of adjoint operator 
and from equation \eqref{eq:abasis}.
\end{proof}

Fix an endomorphism $T\in\End V$ and let $T=IJ$,
with $I:\,\im T\hookrightarrow V$ and $J:\,V\twoheadrightarrow\im T$,
be its canonical decomposition given by \eqref{IJ}.
We shall assume that $T$ is either selfadjoint or skewadjoint 
with respect to the inner product $\langle\cdot\,|\,\cdot\rangle$:
\begin{equation}\label{eq:delta}
T^\dagger=\delta T
\,\,\text{ where }\,\,\delta\in\{\pm1\}
\,.
\end{equation}
\begin{lemma}\label{20170322:lem4}
We have a well defined non-degenerate bilinear form $\langle\cdot\,|\,\cdot\rangle^T$ on $\im T$,
depending on the endomorphism $T$,
given by the following formula
\begin{equation}\label{eq:productU}
\langle u_1|u_2\rangle^T
=
\langle J^{-1}(u_1)|I(u_2)\rangle
\,\,\text{ for all }\, u_1,u_2\in\im T
\,.
\end{equation}
The inner product $\langle\cdot\,|\,\cdot\rangle^T$
has the following parity, depending on the parity \eqref{eq:epsilon} of $\langle\cdot\,|\,\cdot\rangle$
and the self/skew-adjointness \eqref{eq:delta} of $T$:
\begin{equation}\label{eq:epsdelta}
\langle u_1|u_2\rangle^T
=
\epsilon\delta\,
\langle u_2|u_1\rangle^T
\,\,\text{ for all }\, u_1,u_2\in\im T
\,.
\end{equation}
\end{lemma}
\begin{proof}
First, we need to show that the RHS of formula \eqref{eq:productU}
is well defined, i.e., for $u_1,u_2\in\im T$,
$\langle\tilde{u_1}|u_2\rangle$ does not depend on the choice 
of the representative $\tilde{u_1}\in J^{-1}(u_1)$.
This is the same as saying that
$\ker J\perp \im I$ with respect to $\langle\cdot\,|\,\cdot\rangle$.
But $\ker J=\ker T$ and $\im I=\im T$ and,
for a self or skew adjoint operator $T$, $\ker T$ and $\im T$ are automatically orthogonal.

Next, we prove that $\langle\cdot\,|\,\cdot\rangle^T$ is non-degenerate.
Indeed, if the RHS of \eqref{eq:productU} vanishes for all $u_1\in\im T$,
then $\langle v|I(u_2)\rangle=0$ for all $v\in V$,
which implies $u_2=0$, since $\langle\cdot\,|\,\cdot\rangle$ is non-degenerate
and $I$ is injective.

Finally, we prove the symmetry property \eqref{eq:epsdelta}.
Let $u_1,u_2\in\im T$ and let $\tilde{u_1},\tilde{u_2}\in V$
be their preimages via $J$: $J(\tilde{u_1})=u_1$ and $J(\tilde{u_2})=u_2$.
We have
$$
\begin{array}{l}
\displaystyle{
\vphantom{\Big(}
\langle u_1|u_2\rangle^T
=
\langle J^{-1}(u_1)|I(u_2)\rangle
=
\langle \tilde{u_1}|u_2\rangle
=
\epsilon\langle u_2|\tilde{u_1}\rangle
=
\epsilon\langle J(\tilde{u_2})|\tilde{u_1}\rangle
=
\epsilon\delta\langle \tilde{u_2}|J(\tilde{u_1})\rangle
} \\
\displaystyle{
\vphantom{\Big(}
=
\epsilon\delta\langle \tilde{u_2}|u_1\rangle
=
\epsilon\delta\langle u_2|u_1\rangle^T
\,.}
\end{array}
$$
\end{proof}
\begin{lemma}\label{20170411:lem2}
Let $T\in\End V$ satisfy condition \eqref{eq:delta},
and let $\langle\cdot\,|\,\cdot\rangle^T$ be the inner product on $\im T$ defined by \eqref{eq:productU}.
For $A\in\End V$, we have
$$
(JAI)^{\dagger_T}=\delta JA^\dagger I\,,
$$
where $\dagger$ denotes the adjoint in $\End V$ w.r.t. $\langle\cdot\,|\,\cdot\rangle$,
and $\dagger_T$ denotes the adjoint in $\End(\im T)$ w.r.t. $\langle\cdot\,|\,\cdot\rangle^T$.
\end{lemma}
\begin{proof}
By the definition \eqref{eq:productU} of $\langle\cdot\,|\,\cdot\rangle^T$,
the assumption \eqref{eq:epsilon} and condition \eqref{eq:epsdelta}, 
we have
$$
\begin{array}{l}
\displaystyle{
\vphantom{\Big(}
\langle u_1|(JAI)^{\dagger_T}u_2\rangle^T
=
\langle JAI u_1|u_2\rangle^T
=
\langle AI u_1| Iu_2\rangle
=
\langle I u_1| A^\dagger Iu_2\rangle
} \\
\displaystyle{
\vphantom{\Big(}
=
\epsilon\langle A^\dagger Iu_2 | I u_1 \rangle
=
\epsilon\langle JA^\dagger Iu_2 | u_1 \rangle^T
=
\epsilon^2\delta\langle u_1 | JA^\dagger Iu_2 \rangle^T
\,.}
\end{array}
$$
\end{proof}
\begin{lemma}\label{20170322:lem5}
Let  $\{t_h\}_{h=1}^M$ and $\{t^h\}_{h=1}^M$ be bases of $\im T$
dual with respect to $\langle\cdot\,|\,\cdot\rangle^T$.
We have:
\begin{equation}\label{eq:complete2}
\begin{array}{l}
\displaystyle{
\vphantom{\big(}
\sum_{h=1}^Mt^h\otimes I(t_h)
=
\sum_{k=1}^NJ(v^k)\otimes v_k
\,,} \\
\displaystyle{
\vphantom{\big(}
\sum_{h=1}^MI(t^h)\otimes t_h
=
\delta\sum_{k=1}^Nv^k\otimes J(v_k)
\,.}
\end{array}
\end{equation}
\end{lemma}
\begin{proof}
We first observe that $\sum_{k=1}^NJ(v^k)\otimes v_k\in\im T\otimes\im T$.
Indeed, if we pair the first factor with $v\in V$, we have
$$
\sum_{k=1}^N\langle J(v^k)|v\rangle \otimes v_k
=\delta\sum_{k=1}^N\langle v^k|J(v)\rangle \otimes v_k
=\delta J(v)\in\im T
\,.
$$
Next, 
if we take the inner product $\langle\cdot\,|\,\cdot\rangle^T$
of $\sum_{h=1}^Mt^h\otimes t_h$
with $u_1\otimes u_2\in\im T\otimes\im T$ we have,
by definition of dual bases, 
$$
\sum_{h=1}^M\langle t^h|u_1\rangle^T\langle t_h|u_2\rangle^T=\langle u_1|u_2\rangle^T
\,.
$$
On the other hand, 
if we take the same inner product 
of the RHS of the first equation in \eqref{eq:complete2} 
with $u_1\otimes u_2\in\im T\otimes\im T$, we have,
by \eqref{eq:completeness},
$$
\begin{array}{l}
\displaystyle{
\vphantom{\Big(}
\sum_{k=1}^N\langle J(v^k)|u_1\rangle^T\langle v_k|u_2\rangle^T
=
\sum_{k=1}^N\langle v^k|I(u_1)\rangle\langle J^{-1}(v_k)|I(u_2)\rangle
} \\
\displaystyle{
\vphantom{\Big(}
=
\langle J^{-1}(u_1)|I(u_2)\rangle
=
\langle u_1|u_2\rangle^T
\,.}
\end{array}
$$
This proves the first equation in \eqref{eq:complete2}.
If we permute the two factors in both sides of the first equation in\eqref{eq:complete2},
we get
$$
\sum_{h=1}^M I(t_h)\otimes t^h
=
\sum_{k=1}^Nv_k\otimes J(v^k)
\,.
$$
But $\sum_ht_h\otimes t^h=\epsilon\delta\sum_ht^h\otimes t_h$,
and $\sum_kv_k\otimes v^k=\epsilon\sum_kv^k\otimes v_k$.
The second equation in \eqref{eq:complete2} follows.
\end{proof}
The following result will be essential in Section \ref{sec:4.3}:
\begin{lemma}\label{20170322:lem6}
Consider the operator $\Omega_V^\dagger$
associated to the vector space $V$ and its inner product $\langle\cdot\,|\,\cdot\rangle$,
and the operator $\Omega_{\im T}^\dagger$
associated to the subspace $\im T$ and its inner product $\langle\cdot\,|\,\cdot\rangle^T$.
The following identity holds in $\Hom(V,\im T)\otimes\Hom(\im T,V)$:
\begin{equation}\label{permut-omega-dagger1}
(J\otimes\id_V)\Omega_V^\dagger(\id_V\otimes I)
=
(\id_{\im T}\otimes I)\Omega_{\im T}^\dagger(J\otimes\id_{\im T})
\,,
\end{equation}
and the following identity holds in $\Hom(\im T,V)\otimes\Hom(V,\im T)$:
\begin{equation}\label{permut-omega-dagger2}
(\id_V\otimes J)\Omega_V^\dagger(I\otimes\id_V)
=
(I\otimes\id_{\im T})\Omega_{\im T}^\dagger(\id_{\im T}\otimes J)
\,.
\end{equation}
\end{lemma}
\begin{proof}
If we apply the LHS of \eqref{permut-omega-dagger1} to $v\otimes u\in V\otimes\im T$,
we get
$$
\langle v|I(u)\rangle \sum_{k=1}^N
J(v^k)\otimes v_k\,,
$$
while we apply the RHS of \eqref{permut-omega-dagger1} to $v\otimes u\in V\otimes\im T$,
we get
$$
\langle J(v)|u\rangle^T
\sum_{h=1}^M t^h\otimes I(t_h)
\,.
$$
Hence, equation \eqref{permut-omega-dagger1} follows by the definition \eqref{eq:productU}
of the inner product $\langle\cdot\,|\,\cdot\rangle^T$ and by the first equation
in  \eqref{eq:complete2}.

Next, let us apply the LHS of \eqref{permut-omega-dagger2} to $u\otimes v\in\im T\otimes V$.
As a result we get
$$
\langle I(u)|v\rangle
\sum_{k=1}^Nv^k\otimes J(v_k)
\,.
$$
On the other hand, 
if we apply the RHS of \eqref{permut-omega-dagger2} to $u\otimes v\in\im T\otimes V$,
we get
$$
\langle u|J(v)\rangle^T
\sum_{h=1}^M I(t^h)\otimes t_h
\,.
$$
Equation \eqref{permut-omega-dagger2} follows by the second equation in \eqref{eq:complete2}
and by the following identity, $\langle u|J(v)\rangle^T=\delta\langle I(u)|v\rangle$,
which is easily checked.
\end{proof}

\subsection{Formula for the $\lambda$-bracket of $A_\epsilon(z)$}\label{sec:4.2}

Consider the pencil of Poisson vertex algebras $\mc V_\epsilon(\mf g,s)$, $\epsilon\in\mb F$,
associated to the Lie algebra $\mf g$ and its element $s\in\mf g$,
with $\lambda$-bracket \eqref{lambda}.
Recall that, 
given a faithful representation $\varphi:\,\mf g\hookrightarrow\End V$ of $\mf g$,
we constructed the differential operator 
\begin{equation}\label{Aepsilon}
A_\epsilon(\partial)
=
\partial\id+\sum_{i\in I}u_iU^i+\epsilon S
\,\in\mc V(\mf g)[\partial]\otimes\End V
\,.
\end{equation}
Recall the definition \eqref{Omega} of $\Omega_V\in\End V\otimes\End V$.
We shall also denote by $\Omega^{\mf g}_V$ 
(or simply $\Omega^{\mf g}$, if confusion may not arise),
the following operator:
$$
\Omega^{\mf g}_V=\sum_{i\in I}U_i\otimes U^i\,\in\End V\otimes\End V
\,.
$$
We shall mainly be interested in the following three cases:
\begin{description}
\item[Case 1]
$\mf g=\mf{gl}_N$ and $V=\mb F^N$ is the defining representation.
In this case $\Omega^{\mf g}=\Omega$.
\item[Case 2]
$\mf g=\mf{sl}_N$ and $V=\mb F^N$ is the defining representation.
In this case $\Omega^{\mf g}=\Omega-\frac1N\id\otimes\id$.
\item[Case 3]
$\mf g=\{A\in\End V\,|\,A^\dagger=-A\}$, 
where $A^\dagger$ denotes the adjoint
w.r.t. a symmetric or skewsymmetric non-degenerate bilinear form 
$\langle\cdot\,|\,\cdot\rangle$ on $V$;
in other words, $\mf g\simeq\mf{so}_N$ if the form is symmetric, 
and $\mf g\simeq\mf{sp}_N$ if the form is skewsymmetric,
and $V\simeq\mb F^N$ is the defining representation of $\mf g$.
In this case $\Omega^\mf g=\frac12(\Omega-\Omega^\dagger)$.
\end{description}
\begin{lemma}\label{20170318:lem}
\begin{enumerate}[(a)]
\item
The following identity holds:
$$
\{A_\epsilon(z)_\lambda A_\epsilon(w)\}_\epsilon
=
\sum_{i\in I}u_i[\id\otimes U^i,\Omega^{\mf g}]
+\lambda\Omega^{\mf g}+\epsilon[\id\otimes S,\Omega^{\mf g}]
\,.
$$
\item
The following identity holds:
$$
\begin{array}{l}
\displaystyle{
\vphantom{\Big(}
(\id\otimes A_\epsilon(w+\lambda+\partial))(z-w-\lambda-\partial)^{-1}
(A_\epsilon^*(\lambda-z)\otimes\id)\Omega
} \\
\displaystyle{
\vphantom{\Big(}
-\Omega\,(A_\epsilon(z)\otimes(z-w-\lambda-\partial)^{-1}A_\epsilon(w))
} \\
\displaystyle{
\vphantom{\Big(}
=
\sum_{i\in I}u_i[\id\otimes U^i,\Omega]
+\lambda\Omega+\epsilon[\id\otimes S,\Omega]
\,.}
\end{array}
$$
\item
The following identity holds:
$$
\big(\id\otimes \big(A_\epsilon(w+\lambda+\partial)-A_\epsilon(w)\big)\big)
(\lambda+\partial)^{-1}
\big(\big(A_\epsilon^*(\lambda-z)-A_\epsilon(z)\big)\otimes\id\big)
=
-\lambda\id\otimes\id
\,.
$$
\item
If $\mf g$,$V$ are as in Case 3 above,
then the following identity holds:
$$
\begin{array}{l}
\displaystyle{
\vphantom{\Big(}
(\id\otimes A_\epsilon(w+\lambda+\partial))
\Omega^\dagger(z+w+\partial)^{-1}(A_\epsilon(z)\otimes\id)
} \\
\displaystyle{
\vphantom{\Big(}
-(A_\epsilon^*(\lambda-z)\otimes\id)\Omega^\dagger(z+w+\partial)^{-1}(\id\otimes A_\epsilon(w))
} \\
\displaystyle{
\vphantom{\Big(}
=
\sum_{i\in I}u_i[\id\otimes U^i,\Omega^\dagger]
+\lambda\Omega^\dagger+\epsilon[\id\otimes S,\Omega^\dagger]
\,.}
\end{array}
$$
Moreover, in this case we have 
$(A^*(\partial))^\dagger=-A(\partial)$.
\end{enumerate}
\end{lemma}
\begin{proof}
By the definition \eqref{Aepsilon} of $A_\epsilon(z)$ 
and the definition \eqref{lambda} of the $\lambda$-bracket  in $\mc V_\epsilon(\mf g,s)$,
we have
\begin{equation}\label{20170320:eq1}
\begin{array}{l}
\displaystyle{
\vphantom{\Big(}
\{A_\epsilon(z)_\lambda A_\epsilon(w)\}_\epsilon
=
\sum_{i,j\in I}\{{u_i}_\lambda{u_j}\}_\epsilon U^i\otimes U^j
} \\
\displaystyle{
\vphantom{\Big(}
=
\sum_{i,j\in I}
\big(
[u_i,u_j]+\lambda(u_i|u_j)
+\epsilon(s|[u_i,u_j])
\big)U^i\otimes U^j
}
\end{array}
\end{equation}
On the other hand, by the completeness relation in $\mf g$, we have
\begin{equation}\label{20170320:eq2}
\sum_{i,j\in I}
[u_i,u_j]
U^i\otimes U^j
=
\sum_{i,k\in I}
u_k
U^i\otimes [U^k,U_i]
=
\sum_{k\in I}
u_k
[\id\otimes U^k,\Omega^{\mf g}]
\,,
\end{equation}
we have
\begin{equation}\label{20170320:eq3}
\sum_{i,j\in I}
(u_i|u_j)
U^i\otimes U^j
=
\sum_{i\in I}
U^i\otimes U_i
=
\Omega^{\mf g}
\,,
\end{equation}
and we have
\begin{equation}\label{20170320:eq4}
\sum_{i,j\in I}
(s|[u_i,u_j])
U^i\otimes U^j
=
\sum_{i\in I}
U^i\otimes \varphi([s,u_i])
=
[\id\otimes S,\Omega^{\mf g}]
\,.
\end{equation}
Combining equations \eqref{20170320:eq1}--\eqref{20170320:eq4}, we get part (a).

Next, let us prove part (b).
We have, by a straightforward computation,
\begin{equation}\label{20170320:eq5}
\begin{array}{l}
\displaystyle{
\vphantom{\Big(}
A_\epsilon(w+\lambda+\partial)
\otimes(z-w-\lambda-\partial)^{-1}
A_\epsilon^*(\lambda-z)
} \\
\displaystyle{
\vphantom{\Big(}
-A_\epsilon(z)
\otimes(z-w-\lambda-\partial)^{-1}
A_\epsilon(w)
} \\
\displaystyle{
\vphantom{\Big(}
=
\lambda\id\otimes\id
+\sum_{i\in I}u_i(U^i\otimes\id-\id\otimes U^i)
+\epsilon(S\otimes\id-\id\otimes S)
\,.}
\end{array}
\end{equation}
Claim (b) is obtained multiplying both sides of \eqref{20170320:eq5} on the left by $\Omega$
and applying \eqref{eq:permut}.

Part (c) is immediate by definition \eqref{Aepsilon}, 
since $A_\epsilon(w+\lambda+\partial)-A_\epsilon(w)=(\lambda+\partial)\id$
and $A_\epsilon^*(\lambda-z)-A_\epsilon(z)=-\lambda\id$.

Finally, for part (d), we have,
denoting $\tilde{u_i}=u_i+\epsilon(s|u_i)$,
\begin{equation}\label{20170320:eq6}
\begin{array}{l}
\displaystyle{
\vphantom{\Big(}
(\id\otimes A_\epsilon(w+\lambda+\partial))
\Omega^\dagger(z+w+\partial)^{-1}(A_\epsilon(z)\otimes\id)
} \\
\displaystyle{
\vphantom{\Big(}
-(A_\epsilon^*(\lambda-z)\otimes\id)\Omega^\dagger(z+w+\partial)^{-1}(\id\otimes A_\epsilon(w))
} \\
\displaystyle{
\vphantom{\Big(}
=
(z(w+\lambda)-(z-\lambda)w)(z+w)^{-1}\Omega^\dagger
} \\
\displaystyle{
\vphantom{\Big(}
+\sum_{i\in I}\Big(
(w+\lambda+\partial)(z+w+\partial)^{-1}\tilde{u_i} \Omega^\dagger(U^i\otimes\id)
} \\
\displaystyle{
\vphantom{\Big(}
-(z-\lambda)(z+w+\partial)^{-1}\tilde{u_i} \Omega^\dagger(\id\otimes U^i)
} \\
\displaystyle{
\vphantom{\Big(}
+z(z+w)^{-1}\tilde{u_i} (\id\otimes U^i)\Omega^\dagger
-w(z+w)^{-1}\tilde{u_i} (U^i\otimes\id)\Omega^\dagger
\Big)
} \\
\displaystyle{
\vphantom{\Big(}
+\sum_{i,j\in I}
\tilde{u_i}(z+w+\partial)^{-1}\tilde{u_j}
\Big(
(\id\otimes U^i)\Omega^\dagger(U^j\otimes\id)
-(U^i\otimes\id)\Omega^\dagger(\id\otimes U^j)
\Big)
\,.}
\end{array}
\end{equation}
By equations \eqref{eq:permut-tau}, we have
$\Omega^\dagger(U^i\otimes\id)=-\Omega^\dagger(\id\otimes U^i)$,
$(U^i\otimes\id)\Omega^\dagger=-(\id\otimes U^i)\Omega^\dagger$,
and $(\id\otimes U^i)\Omega^\dagger(U^j\otimes\id)=(U^i\otimes\id)\Omega^\dagger(\id\otimes U^j)$.
Hence, the RHS of \eqref{20170320:eq6} becomes
$$
\lambda\Omega^\dagger
+\sum_{i\in I}\tilde{u_i} [\id\otimes U^i,\Omega^\dagger]
\,,
$$
proving the first assertion in claim (d). The last assertion in claim (d) is obvious.
\end{proof}
As a consequence of Lemma \ref{20170318:lem},
in the three cases 1-3 described above
we have the following formulas for $\{A_\epsilon(z)_\lambda A_\epsilon(w)\}_\epsilon$:
\begin{description}
\item[Case 1]
For $\mf g=\mf{gl}_N$ and $V=\mb F^N$, 
$A_\epsilon(\partial)$ satisfies the following \emph{Adler identity} (cf. \cite[Eq.(5.1)-(5.2)]{AGD}):
\begin{equation}\label{eq:adler-glN}
\begin{array}{c}
\displaystyle{
\vphantom{\Big(}
\{A_\epsilon(z)_\lambda A_\epsilon(w)\}_\epsilon
=
(\id\otimes A_\epsilon(w+\lambda+\partial))(z-w-\lambda-\partial)^{-1}
(A_\epsilon^*(\lambda-z)\otimes\id)\Omega
} \\
\displaystyle{
\vphantom{\Big(}
-\Omega\,(A_\epsilon(z)\otimes(z-w-\lambda-\partial)^{-1}A_\epsilon(w))
\,.}
\end{array}
\end{equation}
\item[Case 2]
For $\mf g=\mf{sl}_N$ and $V=\mb F^N$,
$A_\epsilon(\partial)$ satisfies the following \emph{modified Adler identity}:
\begin{equation}\label{eq:adler-slN}
\begin{array}{l}
\displaystyle{
\vphantom{\Big(}
\{A_\epsilon(z)_\lambda A_\epsilon(w)\}_\epsilon
=
(\id\otimes A_\epsilon(w+\lambda+\partial))(z-w-\lambda-\partial)^{-1}
(A_\epsilon^*(\lambda-z)\otimes\id)\Omega
} \\
\displaystyle{
\vphantom{\Big(}
-\Omega\,(A_\epsilon(z)\otimes(z-w-\lambda-\partial)^{-1}A_\epsilon(w))
} \\
\displaystyle{
\vphantom{\Big(}
+\frac1N\big(\id\otimes \big(A_\epsilon(w+\lambda+\partial)-A_\epsilon(w)\big)\big)
(\lambda+\partial)^{-1}
\big(\big(A_\epsilon^*(\lambda-z)-A_\epsilon(z)\big)\otimes\id\big)
\,.}
\end{array}
\end{equation}
\item[Case 3]
For $\mf g\simeq\mf{so}_N$ or $\mf{sp}_N$ and $V\simeq\mb F^N$,
$A_\epsilon(\partial)$ satisfies 
$(A^*(\partial))^\dagger=-A(\partial)$,
and the following \emph{twisted Adler identity}:
\begin{equation}\label{eq:adler-soN}
\begin{array}{l}
\displaystyle{
\vphantom{\Big(}
\{A_\epsilon(z)_\lambda A_\epsilon(w)\}_\epsilon
=
\frac12
(\id\otimes A_\epsilon(w+\lambda+\partial))(z-w-\lambda-\partial)^{-1}
(A_\epsilon^*(\lambda-z)\otimes\id)\Omega
} \\
\displaystyle{
\vphantom{\Big(}
-\frac12
\Omega\,(A_\epsilon(z)\otimes(z-w-\lambda-\partial)^{-1}A_\epsilon(w))
} \\
\displaystyle{
\vphantom{\Big(}
-\frac12
(\id\otimes A_\epsilon(w+\lambda+\partial))
\Omega^\dagger(z+w+\partial)^{-1}(A_\epsilon(z)\otimes\id)
} \\
\displaystyle{
\vphantom{\Big(}
+\frac12
(A_\epsilon^*(\lambda-z)\otimes\id)\Omega^\dagger(z+w+\partial)^{-1}(\id\otimes A_\epsilon(w))
\,.}
\end{array}
\end{equation}
\end{description}

\subsection{The generalized Adler identity}\label{sec:4.3}

The following notion is introduced to include all three 
equations \eqref{eq:adler-glN}-\eqref{eq:adler-soN} as special cases.
\begin{definition}\label{def:adler-general}
Let $A(\partial)\in\mc V((\partial^{-1}))\otimes \End V$
be an $\End V$-valued pseudodifferential operator
over the PVA $\mc V$.
We say that $A(\partial)$ is an operator of 
\emph{generalized Adler type} if
\begin{equation}\label{eq:adler-general}
\begin{array}{l}
\displaystyle{
\vphantom{\Big(}
\{A(z)_\lambda A(w)\}
=
\alpha
(\id\otimes A(w+\lambda+\partial))(z-w-\lambda-\partial)^{-1}
(A^*(\lambda-z)\otimes\id)\Omega
} \\
\displaystyle{
\vphantom{\Big(}
-\alpha
\Omega\,\big(A(z)\otimes(z-w-\lambda-\partial)^{-1}A(w)\big)
} \\
\displaystyle{
\vphantom{\Big(}
-\beta
(\id\otimes A(w+\lambda+\partial))
\Omega^\dagger(z+w+\partial)^{-1}(A(z)\otimes\id)
} \\
\displaystyle{
\vphantom{\Big(}
+\beta
(A^*(\lambda-z)\otimes\id)\Omega^\dagger(z+w+\partial)^{-1}(\id\otimes A(w))
} \\
\displaystyle{
\vphantom{\Big(}
+\gamma\big(\id\otimes \big(A(w+\lambda+\partial)-A(w)\big)\big)
(\lambda+\partial)^{-1}
\big(\big(A^*(\lambda-z)-A(z)\big)\otimes\id\big)
\,,}
\end{array}
\end{equation}
for some $\alpha,\beta,\gamma\in\mb F$,
where $\Omega$ is given by \eqref{Omega} (for the vector space $V$).
If $\beta\neq0$,
we assume that $V$ carries a symmetric or skewsymmetric 
non-degenerate bilinear form $\langle\cdot\,|\,\cdot\rangle$
and $\Omega^\dagger$ is given by \eqref{Omega-tau}.
Moreover, in this case we also assume that 
\begin{equation}\label{eq:adjoint}
(A^*(\partial))^\dagger=\eta A(\partial)
\,\,\text{ where }\,\, \eta\in\{\pm1\}
\,.
\end{equation}
\end{definition}
\begin{remark}\label{rem:regular}
In equation \eqref{eq:adler-general}, as well as in the analogous equations above in this Section,
we can expand all terms either using $\iota_z$ or using $\iota_w$ (but not both)
and the result is the same.
Indeed, the coefficient of $\alpha$ is clearly regular in $z-w-\lambda-\partial$,
and, thanks to the assumption \eqref{eq:adjoint},
the coefficient of $\beta$ is regular in $z+w+\partial$.
Note also that the last term of the RHS is regular in $\lambda+\partial$
(so no expansion is needed).
\end{remark}
For example, for the three cases listed in Section \ref{sec:4.2}
equations \eqref{eq:adler-glN}, \eqref{eq:adler-slN} and \eqref{eq:adler-soN}
for the operator $A_\epsilon(\partial)$ 
correspond to the following values of the parameters $\alpha,\beta,\gamma$:
\begin{equation}\label{table}
\begin{tabular}{ll|lll}
\vphantom{\Big(}
$\mf g$ & $V$ & \,\,$\alpha$\,\, & \,\,$\beta$\,\, & \,\,$\gamma$\,\, \\
\hline 
\vphantom{\Big(}
$\mf{gl}_N$ & $\mb F^N$ & 1 & 0 & 0 \\
\vphantom{\Big(}
$\mf{sl}_N$ & $\mb F^N$ & 1 & 0 & $\frac1N$ \\
\vphantom{\Big(}
$\mf{so}_N$ or $\mf{sp}_N$ & $\mb F^N$ & $\frac12$ & $\frac12$ & 0 \\
\end{tabular}
\end{equation}
\begin{theorem}\label{thm:main2}
Let $A(\partial)\in\mc V((\partial^{-1}))\otimes\End V$
be an $\End V$ valued pseudodifferential operator over the Poisson vertex algebra $\mc V$,
of generalized Adler type.
Then:
\begin{enumerate}[(a)]
\item
If $A(\partial)$ is invertible in $\mc V((\partial^{-1}))\otimes\End V$,
then $A^{-1}(\partial)$ satisfies the generalized Adler identity \eqref{eq:adler-general}
with the opposite values of $\alpha$ and $\beta$ (and the same value of $\gamma$).
Furthermore, if $\beta\neq0$, then $((A^{-1})^*(\partial))^\dagger=\eta A^{-1}(\partial)$.
\item
Let $T\in\End V$ and $I,J$ be as in \eqref{IJ}.
If $\beta\neq0$, we assume that $T^\dagger=\delta T$, $\delta\in\{\pm1\}$
and we consider the corresponding inner product $\langle\cdot\,|\,\cdot\rangle^T$ 
on $\im T$, defined by \eqref{eq:productU}.
Then, $JA(\partial)I\in\mc V\otimes\End(\im T)\otimes\End(\im T)$ 
satisfies the generalized Adler identity \eqref{eq:adler-general}
(with the same values of $\alpha,\beta,\gamma$),
and, for $\beta\neq0$, we have $(JA^*(\partial)I)^{\dagger_T}=\eta\delta JA(\partial)I$.
\item
If, moreover, the generalized quasideterminant 
$$
|A(\partial)|_{I,J}
:=
(JA^{-1}(\partial)I)^{-1}
\,,
$$
exists
(i.e. $A(\partial)$ is invertible in $\mc V((\partial^{-1}))\otimes\End V$,
and $JA^{-1}(\partial)I$ is invertible in $\mc V((\partial^{-1}))\otimes\End(\im T)$),
then $|A(\partial)|_{I,J}$ satisfies the generalized Adler identity \eqref{eq:adler-general}
(with the same values of $\alpha,\beta,\gamma$),
and, for $\beta\neq0$, we have $(|A|_{I,J}^*(\partial))^{\dagger_T}=\eta\delta |A(\partial)|_{I,J}$.
\end{enumerate}
\end{theorem}
\begin{proof}
If we apply Lemma \ref{20170404:lem2}(g)-(h), we get
\begin{equation}\label{20170321:eq1}
\begin{array}{l}
\displaystyle{
\vphantom{\Big(}
\{A^{-1}(z)_\lambda A^{-1}(w)\}
=
\big(
\big(\big|_{x_1=\partial}(A^{-1})^*(\lambda-z)\big)\otimes A^{-1}(w+\lambda+x_1+x_2+y_2+u)
\big)
} \\
\displaystyle{
\vphantom{\Big(}
\times\big(\big|_{u=\partial}\{A(z+x_2)_{\lambda+x_1+x_2}A(w+y_2)\}\big)
\big(
\big(\big|_{x_2=\partial}A^{-1}(z)\big)\otimes\big(\big|_{y_2=\partial}A^{-1}(w)\big)
\big)
\,.}
\end{array}
\end{equation}
Here we are using of the notation \eqref{eq:notation}.
We then use the generalized Adler identity \eqref{eq:adler-general}
to rewrite the RHS of \eqref{20170321:eq1} as
\begin{equation}\label{20170321:eq2}
\begin{array}{l}
\displaystyle{
\vphantom{\Big(}
\{A^{-1}(z)_\lambda A^{-1}(w)\}
=
\alpha\,
(z-w-\lambda-y_2)^{-1}
\Omega
\big(
A^{-1}(z)\otimes\big(\big|_{y_2=\partial}A^{-1}(w)\big)
\big)
} \\
\displaystyle{
\vphantom{\Big(}
-\alpha
\big(
\big(\big|_{x_1=\partial}(A^{-1})^*(\lambda-z)\big)\otimes A^{-1}(w+\lambda+x_1)
\big)
\Omega
(z-w-\lambda-x_1)^{-1}
} \\
\displaystyle{
\vphantom{\Big(}
-\beta
\big((A^{-1})^*(\lambda-z)\otimes \id\big)
\Omega^\dagger(z+w+y_2)^{-1}
\big(\id\otimes\big(\big|_{y_2=\partial}A^{-1}(w)\big)\big)
} \\
\displaystyle{
\vphantom{\Big(}
+\beta
\big(\id\otimes A^{-1}(w+\lambda+x_2)\big)
\Omega^\dagger(z+w+x_2)^{-1}
\big(\big(\big|_{x_2=\partial}A^{-1}(z)\big)\otimes\id\big)
} \\
\displaystyle{
\vphantom{\Big(}
+\!\gamma
\big(
\id\!\otimes\! 
\big(
A^{-1}(w)
-
A^{-1}(w\!+\!\lambda\!+\!x)
\big)
\big)
(\lambda\!+\!x)^{-1}
\big(
\big(\big|_{x=\partial}
A^{-1}(z)-(\!A^{-1}\!)^*(\lambda\!-\!z)
\big)\!\otimes\!\id
\big)
.}
\end{array}
\end{equation}
Here we used the identities
$$
\begin{array}{l}
\displaystyle{
\vphantom{\Big(}
A^{-1}(w+\lambda+y)\big(\big|_{y=\partial}A(w+\lambda)\big)
=\id\,\,,\,\,\,\,
A(w+y)\big(\big|_{y=\partial}A^{-1}(w)\big)
=\id
} \\
\displaystyle{
\vphantom{\Big(}
\big(\big|_{x=\partial}(A^{-1})^*(\lambda-z)\big)A^*(\lambda+x-z)
=\id\,\,,\,\,\,\,
A(z+x)\big(\big|_{x=\partial}A^{-1}(z)\big)
=\id
\,,}
\end{array}
$$
which are a consequence of the identities $AA^{-1}=A^{-1}A=\id$ and Lemma \ref{20170404:lem1}.
Equation \eqref{20170321:eq2} is the generalized Adler identity \eqref{eq:adler-general} for $A^{-1}$,
with the opposite values of $\alpha$ and $\beta$.
The last assertion of claim (a) follows from Lemma \ref{20170411:lem1}.

Next, for part (b),
we have
$\{JA(z)I_\lambda JA(w)I\}
=(J\otimes J)\{A(z)_\lambda A(w)\}(I\otimes I)$.
Hence, by \eqref{eq:adler-general} we get
\begin{equation}\label{20170323:eq1}
\begin{array}{l}
\displaystyle{
\vphantom{\Big(}
\{JA(z)I_\lambda JA(w)I\}
} \\
\displaystyle{
\vphantom{\Big(}
=
\alpha
(J\otimes J)(\id_V\otimes A(w+\lambda+\partial))(z-w-\lambda-\partial)^{-1}
(A^*(\lambda-z)\otimes\id_V)\Omega_V(I\otimes I)
} \\
\displaystyle{
\vphantom{\Big(}
-\alpha
(J\otimes J)\Omega_V(A(z)\otimes(z-w-\lambda-\partial)^{-1}A(w))(I\otimes I)
} \\
\displaystyle{
\vphantom{\Big(}
-\beta
(J\otimes J)(\id_V\otimes A(w+\lambda+\partial))
\Omega_V^\dagger(z+w+\partial)^{-1}(A(z)\otimes\id_V)(I\otimes I)
} \\
\displaystyle{
\vphantom{\Big(}
+\beta
(J\otimes J)(A^*(\lambda-z)\otimes\id_V)\Omega_V^\dagger(z+w+\partial)^{-1}(\id_V\otimes A(w))(I\otimes I)
} \\
\displaystyle{
\vphantom{\Big(}
+\gamma
(J\!\otimes\! J)\big(\id_V\!\otimes\! \big(A(w\!+\!\lambda\!+\!\partial)\!-\!A(w)\big)\big)
(\lambda\!+\!\partial)^{-1}
\big(\big(A^*(\lambda\!-\!z)\!-\!A(z)\big)\!\otimes\!\id_V\big)(I\!\otimes\! I)
.}
\end{array}
\end{equation}
Obviously, we have $J\id_V=\id_{\im T}J$ and $\id_VI=I\id_{\im T}$.
We can then use Lemmas \ref{20170322:lem1} and \ref{20170322:lem6} 
to rewrite the RHS of \eqref{20170323:eq1}
as follows:
\begin{equation}\label{20170323:eq2}
\begin{array}{l}
\displaystyle{
\vphantom{\Big(}
\alpha
(\id_{\im T}\otimes JA(w+\lambda+\partial)I)(z-w-\lambda-\partial)^{-1}
(JA^*(\lambda-z)I\otimes\id_{\im T})\Omega_{\im T}
} \\
\displaystyle{
\vphantom{\Big(}
-\alpha
\Omega_{\im T}(JA(z)I\otimes(z-w-\lambda-\partial)^{-1}JA(w)I)
} \\
\displaystyle{
\vphantom{\Big(}
-\beta
(\id_{\im T}\otimes JA(w+\lambda+\partial)I)
\Omega_{\im T}^\dagger(z+w+\partial)^{-1}(JA(z)I\otimes\id_{\im T})
} \\
\displaystyle{
\vphantom{\Big(}
+\beta
(JA^*(\lambda-z)I\otimes\id_{\im T})\Omega_{\im T}^\dagger(z+w+\partial)^{-1}(\id_{\im T}\otimes JA(w)I)
} \\
\displaystyle{
\vphantom{\Big(}
+\gamma
\big(\id_V\!\otimes\! \big(JA(w\!+\!\lambda\!+\!\partial)I\!-\!JA(w)I\big)\big)
(\lambda\!+\!\partial)^{-1}
\big(\big(JA^*(\lambda\!-\!z)I\!-\!JA(z)I\big)\!\otimes\!\id_V\big)
.}
\end{array}
\end{equation}
Hence, $JA(\partial)I\in\mc V((\partial^{-1}))\otimes\End(\im T)$ satisfies the generalized Adler identity
\eqref{eq:adler-general} 
(with respect to the space $\im T$ with inner product $\langle\cdot\,|\,\cdot\rangle^T$).
The last assertion of claim (b) follows from Lemma \ref{20170411:lem2}.

By part (a), $A^{-1}(\partial)$ is an operator of generalized Adler type
with parameters $-\alpha,-\beta,\gamma$.
By part (b) $JA^{-1}(\partial)I$ is of generalized Adler type
with the same parameters $-\alpha,-\beta,\gamma$.
Finally, again by part (a),
$|A(\partial)|_{I,J}$ is of generalized Adler type
with parameters $\alpha,\beta,\gamma$,
proving (c).
\end{proof}
\begin{remark}\label{rem:general}
We could try to generalize the generalized Adler identity \eqref{eq:adler-general} more,
in a way that Theorem \ref{thm:main2} still holds.
This leads to an identity of the following type,
\begin{equation}\label{eq:adler-general2}
\begin{array}{l}
\displaystyle{
\vphantom{\Big(}
\{A(z)_\lambda A(w)\}
=
\big(
\big(\big|_{x=\partial}A^*(\lambda-z)\big)
\otimes
A(w+\lambda+x)
\big)
\alpha(z,w,\lambda+x)
} \\
\displaystyle{
\vphantom{\Big(}
+\beta(z+x,w+y,\lambda+x)
\big(
\big(\big|_{x=\partial}A(z)\big)
\otimes
\big(\big|_{y=\partial}A(w)\big)
\big)
} \\
\displaystyle{
\vphantom{\Big(}
+\big(\id\otimes A(w+\lambda+x)\big)
\gamma(z+x,w,\lambda+x)
\big(\big(\big|_{x=\partial}A(z)\big)\otimes\id\big)
} \\
\displaystyle{
\vphantom{\Big(}
+\big(\big(\big|_{x=\partial}A^*(\lambda-z)\big)\otimes\id\big)
\delta(z,w+y,\lambda+x)
\big(\id\otimes \big(\big|_{y=\partial}A(w)\big)\big)
\,,}
\end{array}
\end{equation}
where $\alpha,\beta,\gamma,\delta$ are functions with values in $\End V\otimes\End V$
satisfying some compatibility condition with $T\in\End V$.
However, despite our efforts, we were not able to find any interesting pair $(\mf g,V)$, 
other than the one listed in table \eqref{table},
for which the operator $A_\epsilon(\partial)\in\mc V(\mf g)((\partial^{-1}))\otimes\End V$
satisfies a generalized Adler identity of the form \eqref{eq:adler-general2}.
\end{remark}

\begin{remark}\label{rem:XYZ}
Let $A(\partial)$ be an operator of $(\alpha,\beta,\gamma,\delta)$-Adler type for the $\lambda$-bracket $\{\cdot\,_\lambda\,\cdot\}$,
in the sense of Remark \ref{rem:general}, 
i.e. we assume that equation \eqref{eq:adler-general2} holds.
Skewsymmetry for the $\lambda$-bracket $\{\cdot\,_\lambda\,\cdot\}$
translates to the condition:
\begin{equation}\label{eq:skew}
\{A(z)_\lambda A(w)=-\{A(w)_{-\lambda-\partial}A(z)\}^{\sigma}
\,\in (\mc V[\lambda])[[z^{-1},w^{-1}]][z,w]\otimes\End V\otimes\End V
\,,
\end{equation}
where $\sigma:\,\End V\otimes\End V\to \End V\otimes\End V$
is the transposition of the two factors.
One can check that \eqref{eq:skew} holds
provided that $\alpha,\beta,\gamma,\delta$ satisfy the following conditions:
\begin{equation}
\begin{split}\label{eq:cond1}
&\alpha(z,w,\lambda)=-\alpha^\sigma(w,z,-\lambda)
\,,
\qquad
\beta(z,w,\lambda)=-\beta^\sigma(w,z,-\lambda)
\,,
\\
&\gamma(z,w,\lambda)=-\delta^\sigma(w,z,-\lambda)
\,.
\end{split}
\end{equation}
Furthermore,
the Jacobi identity for the $\lambda$-bracket $\{\cdot\,_\lambda\,\cdot\}$
translates to the condition:
\begin{equation}\label{eq:jacobi}
\{A(z_1)_\lambda \{A(z_2)_\mu A(z_3)\}\}-\{A(z_2)_\mu \{A(z_1)_\lambda A(z_3)\}\}^{(12)}
=\{\{A(z_1)_\lambda A(z_2)\}_{\lambda+\mu} A(z_3)\}
\,,
\end{equation}
where $(1,2):\,\End V^{\otimes3}\to\End V^{\otimes3}$ is the transposition of the first two factors.
One can check that \eqref{eq:jacobi} holds
provided that $\alpha,\beta,\gamma,\delta$ satisfy the following conditions:
\begin{equation}
\begin{split}\label{eq:cond2}
\Cond(\alpha,\alpha,\alpha)=0\,,
\quad
\Cond(-\alpha,\delta,\delta)=0\,,
\quad
\Cond(\gamma,-\alpha,\gamma)=0\,,
\quad
\Cond(\delta,\gamma,-\alpha)=0\,,
\\
\Cond(\beta,\beta,\beta)=0\,,
\quad
\Cond(-\beta,\gamma,\gamma)=0\,,
\quad
\Cond(\delta,-\beta,\delta)=0\,,
\quad
\Cond(\gamma,\delta,-\beta)=0
\,,
\end{split}
\end{equation}
where
\begin{equation*}
\begin{split}
&\Cond(X,Y,Z)=
X_{12}(z_1,z_2-\lambda-\mu,\lambda)Y_{23}(z_2,z_3,\lambda+\mu)
-Y_{23}(z_2,z_3+\lambda,\mu)Z_{13}(z_1,z_3,\lambda) \\
& -Z_{13}(z_1+\mu,z_3,\lambda+\mu)X_{12}(z_1,z_2,-\mu)
+X_{12}(z_1-\lambda-\mu,z_2,-\mu)Z_{13}(z_1,z_3,\lambda+\mu)\\
&+Z_{13}(z_1,z_3+\mu,\lambda)Y_{23}(z_2,z_3,\mu)
-Y_{23}(z_2+\lambda,z_3,\lambda+\mu)X_{12}(z_1,z_2,\lambda)
\,.
\end{split}
\end{equation*}
Here we are using the standard notation $X_{12}=X\otimes\id\in(\End V)^{\otimes3}$,
and similarly for $X_{23}$ and $X_{13}$.
\end{remark}
\begin{remark}\label{rem:faddeev}
As a special case of Remarks \ref{rem:general} and \ref{rem:XYZ}, let 
$\alpha=-\beta$, and $\gamma=\delta=0$ in equation \eqref{eq:adler-general2}. 
Moreover, assume that $\alpha(z_1,z_2,\lambda)=\alpha(z_1-z_2-\lambda)$
is a function of $z_1-z_2-\lambda$.
Then, condition \eqref{eq:cond1} is equivalent to
\begin{equation}\label{eq:cond1b}
\alpha(z)=-\alpha(-z)^\sigma
\,,
\end{equation}
and condition \eqref{eq:cond2} is equivalent to
\begin{equation}
\begin{split}\label{eq:cond2b}
&\alpha_{12}(z-w+\lambda-\mu)\alpha_{23}(w+\mu)
+\alpha_{13}(z)\alpha_{23}(w)
\\
&+\alpha_{12}(z-w)\alpha_{13}(z)
-\alpha_{23}(w+\mu)\alpha_{13}(z+\lambda)
\\
&-\alpha_{13}(z+\lambda)\alpha_{12}(z-w+\lambda-\mu)
-\alpha_{23}(w)\alpha_{12}(z-w)
=0
\,.
\end{split}
\end{equation}
Equation \eqref{eq:cond1b} is the same as 
\cite[Ch.III, Eq.(1.40)]{faddeev}.
For $\lambda=\mu=0$, 
equation \eqref{eq:cond2b} reduces to \cite[Ch.III, Eq.(1.41)]{faddeev}, 
see also \cite[Eq.(1.4)]{BD82}. 
Moreover, in this case the identity \eqref{eq:adler-general2} reduces to 
the so-called \emph{fundamental Poisson bracket} given by \cite[Ch.III, Eq.(1.20)]{faddeev}.
\end{remark}
\begin{remark}\label{rem:BD}
Assuming that
$\alpha(z_1,z_2,\lambda)=\alpha(z_1-\lambda,z_2)$ is a function of $z_1-\lambda$ and $z_2$,
then, the first equation in \eqref{eq:cond1} can be rewritten as
$\alpha(u_1,u_2)=-\alpha(u_2,u_1)^\sigma$,
which is called unitary condition in \cite{BD82},
while the first equation in \eqref{eq:cond2}, i.e. $\Cond(\alpha,\alpha,\alpha)=0$, 
is equivalent to
\begin{equation*}
\begin{split}
&\alpha_{12}(u_1,u_2-\lambda)\alpha_{23}(u_2-\lambda,u_3)
+\alpha_{13}(u_1,u_3+\mu)\alpha_{23}(u_2,u_3)
\\
&+\alpha_{12}(u_1,u_2)\alpha_{13}(u_1-\mu,u_3)
-\alpha_{23}(u_2,u_3+\lambda)\alpha_{13}(u_1,u_3)
\\
&-\alpha_{13}(u_1,u_3)\alpha_{12}(u_1+\mu,u_2)
-\alpha_{23}(u_2,u_3)\alpha_{12}(u_1,u_2)
=0
\,.
\end{split}
\end{equation*}
which, for $\lambda=\mu=0$, reduces to \cite[Eq.(1.1)]{BD82}.
\end{remark}

\begin{corollary}\label{20170404:cor}
Let $\mf g,V$ and the parameters $\alpha,\beta,\gamma$ be as in table \eqref{table}.
Let $\{f,2x,e\}\subset\mf g$ be an $\mf{sl}_2$-triple,
consider the corresponding $\ad x$-eigenspace decompositions \eqref{eq:grading} and 
\eqref{eq:grading_EndV},
let $s\in\mf g_d$, $T\in(\End V)[D]$, and assume that condition \eqref{20170317:eq2} holds.
In Case 3 (i.e. $\mf g=\mf{so}_N$ or $\mf{sp}_N$ and $V=\mb F^N$),
assume also that $T^\dagger=\delta T$, $\delta\in\{=\pm1\}$,
and consider the corresponding inner product $\langle\cdot\,|\,\cdot\rangle^T$ on $\im T$
defined by \eqref{eq:productU}.
Consider the pencil of $W$-algebras $\mc W_\epsilon(\mf g,f,s)$, $\epsilon\in\mb F$,
and the $\End(\im T)$-valued pseudodifferential operator
$L_\epsilon(\partial)$, defined by \eqref{eq:Lw}, over the PVA $\mc W_\epsilon(\mf g,f,s)$.
Then
$L_\epsilon(\partial)$ is an operator of generalized Adler type
for every $\epsilon\in\mb F$.
\end{corollary}
\begin{proof}
By \eqref{eq:adler-glN}, \eqref{eq:adler-slN} and \eqref{eq:adler-soN},
$A_\epsilon(\partial)$ satisfies \eqref{eq:adler-general}.
Hence, by Theorem \ref{thm:main2}(c)
so does the generalized quasideterminant $(J(A_\epsilon^{-1}(\partial)I)^{-1}$.
Applying the differential algebra homomorphism $\rho$ to both sides of 
the generalized Adler identity \eqref{eq:adler-general} 
for this generalized quasideterminant
(and applying \cite[Cor.3.3(d)]{DSKV13}),
we get the desired result. 
\end{proof}

\subsection{Scalar operators with constant coefficients of generalized Adler type}\label{sec:5.4}

It is natural to ask when a scalar operator with constant coefficients,
$A(\partial)=a(\partial)\id$,
satisfies the generalized Adler identity \eqref{eq:adler-general}.
In this case \eqref{eq:adler-general} reads:
\begin{equation}\label{eq:adler-scalar}
\begin{array}{l}
\displaystyle{
\vphantom{\Bigg(}
\alpha
\frac{a(z-\lambda)a(w+\lambda)-a(z)a(w)}{z-w-\lambda}
\Omega
-\beta
\frac{a(z)a(w+\lambda)-a(z-\lambda)a(w)}{z+w}
\Omega^\dagger
} \\
\displaystyle{
\vphantom{\Bigg(}
+\gamma
\frac{(a(z-\lambda)-a(z))
(a(w+\lambda)-a(w))}\lambda
\id\otimes\id
=
0
\,.}
\end{array}
\end{equation}
In order to make sense of equation \eqref{eq:adler-scalar},
we may assume that $a(\partial)$ lies in $\mb F((\partial^{-1}))$ or in $\mb F((\partial))$.
If $\dim V>1$, the operators $\Omega,\Omega^\dagger$ and $\id\otimes\id\in\End V\otimes\End V$
are linearly independent.
Hence, in this case, equation \eqref{eq:adler-scalar} implies $a(\partial)=a\in\mb F$.
Let us then consider the case when $V=\mb F$,
in which case $\Omega=\Omega^\dagger=\id\otimes\id$,
and equation \eqref{eq:adler-scalar} becomes
\begin{equation}\label{eq:adler-scalar1}
\begin{array}{l}
\displaystyle{
\vphantom{\Bigg(}
\alpha
\frac{a(z-\lambda)a(w+\lambda)-a(z)a(w)}{z-w-\lambda}
-\beta
\frac{a(z)a(w+\lambda)-a(z-\lambda)a(w)}{z+w}
} \\
\displaystyle{
\vphantom{\Bigg(}
+\gamma
\frac{(a(z-\lambda)-a(z))
(a(w+\lambda)-a(w))}\lambda
=
0
\,.}
\end{array}
\end{equation}
If we take the derivative at $\lambda=0$ of both sides of equation \eqref{eq:adler-scalar1},
we get
\begin{equation}\label{eq:adler-scalar2}
\alpha
\frac{a(z)a'(w)-a'(z)a(w)}{z-w}
-\beta
\frac{a(z)a'(w)+a'(z)a(w)}{z+w}
-\gamma a'(z)a'(w)
=
0
\,.
\end{equation}
We can then take the limit for $z\to w$ of both sides of \eqref{eq:adler-scalar2} to get 
\begin{equation}\label{eq:adler-scalar3}
\alpha
\big((a'(w))^2-a(w)a''(w)\big)
-\beta
\frac{a(w)a'(w)}{w}
-\gamma (a'(w))^2
=
0
\,.
\end{equation}
Letting $y(w)=\frac{a'(w)}{a(w)}$,
equation \eqref{eq:adler-scalar3} reduces to the following first order differential equation
for the function $y(w)$:
\begin{equation}\label{eq:adler-scalar4}
\alpha y'
+\beta\frac{y}{w}
+\gamma y^2
=
0
\,,
\end{equation}
which can be easily solved by the method of variation of constants.
The general solutions of equation \eqref{eq:adler-scalar4},
and, up to a multiplicative constant, of equation \eqref{eq:adler-scalar3},
are given in the following table:
\begin{equation}\label{table-2}
\begin{tabular}{l|l|l}
\vphantom{\Big(}
conditions on $\alpha,\beta,\gamma$ & $y(w)$ & $a(w)$ \\
\hline 
\vphantom{\Big(}
--- & $0$ & $1$ \\
$\alpha-\beta-n\gamma=0$ & $\frac{n}{w}$ & $w^n$ \\
$\alpha(n-1)+\beta=0,\,
\alpha\neq0,\,\gamma=0$ & $kw^{n-1}$ & $\exp({kw^n})$ \\
$\beta=\alpha,\,\gamma\neq0,\,\alpha-n\gamma=0$ & $\frac{n}{w(k+\log w)}$ 
& $(k+\log w)^n$
\end{tabular}
\,,
\end{equation}
where $k\in\mb F$.
To conclude, we need to see which of the solutions $a(w)$ listed in Table \eqref{table-2}
are indeed solutions of the generalized Adler identity \eqref{eq:adler-scalar}.
As a result, we get the following complete list of scalar operators of generalized Adler type
(in dimension $1$), up to a multiplicative constant:
\begin{equation}\label{table-3}
\begin{tabular}{l|l}
\vphantom{\Big(}
conditions on $\alpha,\beta,\gamma$ & $a(\partial)$ \\
\hline 
\vphantom{\Big(}
--- & $1$ \\
$\alpha-\beta-\gamma=0$ & $\partial$ \\
$\alpha-\beta+\gamma=0$ & $\partial^{-1}$ \\
$\alpha=-\beta=\gamma$ & $\partial^2$ \\
$\alpha=-\beta=-\gamma$ & $\partial^{-2}$ \\
$\alpha\neq0, \beta=\gamma=0$ & $e^{k\partial}$,\, $k\in\mb F$
\end{tabular}
\end{equation}

\section{Integrable hierarchies for generalized Adler type operators}\label{sec:6}

\begin{theorem}\label{thm:hn}
Let $A(\partial)\in\mc V((\partial^{-1}))\otimes\End V$
be an $\End V$-valued pseudodifferential operator over the Poisson vertex algebra $\mc V$.
Assume that $A(\partial)$ is an operator of generalized Adler type,
and that it is invertible in $\mc V((\partial^{-1}))\otimes\End V$.
For $B(\partial)\in\mc V((\partial^{-1}))\otimes\End V$
a $K$-th root of $A$ (i.e. $A(\partial)=B(\partial)^K$ for $K\in\mb Z\backslash\{0\}$)
define the elements $h_{n,B}\in\mc V$, $n\in\mb Z$, by ($\tr=1\otimes\tr$)
\begin{equation}\label{eq:hn}
h_{n,B}=
\frac{-K}{n}
\Res_z\tr(B^n(z))
\text{ for } n\neq0
\,,\,\,
h_0=0\,.
\end{equation}
Then: 
\begin{enumerate}[(a)]
\item
All the elements $\tint h_{n,B}$ are Hamiltonian functionals in involution:
\begin{equation}\label{eq:invol}
\{\tint h_{m,B},\tint h_{n,C}\}=0
\,\text{ for all } m,n\in\mb Z,\,
B,C \text{ roots of } A
\,.
\end{equation}
\item
The corresponding compatible hierarchy of Hamiltonian equations satisfies
\begin{equation}\label{eq:hierarchy}
\frac{dA(w)}{dt_{n,B}}
=
\{\tint h_{n,B},A(w)\}
=
[\alpha(B^n)_+-\beta((B^{n})^{*\dagger})_+,A](w)
\,,\,\,n\in\mb Z,\,
B \text{ root of } A
\end{equation}
(in the RHS we are taking the symbol of the commutator of matrix pseudodifferential operators),
and the Hamiltonian functionals $\tint h_{n,C}$, $n\in\mb Z_+$, $C$ root of $A$,
are integrals of motion of all these equations.
\end{enumerate}
\end{theorem}
\begin{remark}\label{rem:negative-n}
One can state the same Theorem \ref{thm:hn} without the assumption that $A$ 
(and therefore $B$) is invertible,
at the price of assuming that $K\geq1$,
and of restricting the sequence $h_n$ in \eqref{eq:hn} to $n\in\mb Z_+$.
Moreover, since the proof of \eqref{eq:invol} and \eqref{eq:hierarchy} 
is based on Lemma \ref{lem:hn2},
one needs to restrict equation \eqref{eq:invol} to $m\geq K$ and $n\geq L$,
where $B^K=C^L=A$,
and equation \eqref{eq:hierarchy} to $n\geq K$.
\end{remark}
In the remainder of the section we will give a proof of Theorem \ref{thm:hn}.
This theorem is an extension of \cite[Thm.5.1]{DSKVnew}
to the case of the generalized Adler identity \eqref{eq:adler-general}.
Its proof is based on the following Lemmas \ref{lem:hn1} and Lemma \ref{lem:hn2},
which are essentially the same as Lemmas 2.1 and 5.6 in \cite{DSKVnew} respectively,
but written in terms of endomorphisms instead of matrix elements.
\begin{lemma}\label{lemma:29032017}
Let $A$, $B$ be in $\End V\otimes\End V$. Then
\begin{enumerate}[(a)]
\item $(\tr\otimes 1)(\Omega A)=A'A''\in\End V$;
\item $(\tr\otimes \tr)(\Omega A)=\tr(A'A'')\in\mb F$;
\item $(\tr\otimes 1)(A\Omega^\dagger B)=A''(B'A')^\dagger B'')\in\End V$;
\item $(\tr\otimes \tr)(A\Omega^\dagger B)=\tr(A''(B'A')^\dagger B''))\in\mb F$.
\end{enumerate}
\end{lemma}
\begin{proof}
Parts (a) and (c) are immediate consequences of \eqref{20170404:eq1}
and the cyclic property of the trace.
Parts (b) and (d) are obvious consequences of (a) and (c) respectively.
\end{proof}
\begin{lemma}{\cite[Lem.2.1]{DSKVnew}}\phantomsection\label{lem:hn1}
Given two operators $A(\partial),B(\partial)\in\mc V((\partial^{-1}))\otimes\End V$, we have
\begin{enumerate}[(a)]
\item
$\Res_z A(z)B^*(\lambda-z)=\Res_zA(z+\lambda+\partial)B(z)$;
\item
$\tint \Res_z \tr(A(z+\partial)B(z))=\tint \Res_z\tr(B(z+\partial)A(z))$.
\end{enumerate}
\end{lemma}
\begin{proof}
Part (a) is a consequence of the combinatorial identity
$\res_zz^m(z-\lambda)^n=\res_z(z+\lambda)^mz^n$,
which holds for every $m,n\in\mb Z$.
For part (b) we have
$$
\begin{array}{l}
\displaystyle{
\vphantom{\Big(}
\tint \Res_z \tr(A(z+\partial)B(z))
=
\tint \Res_z \tr(A(z)B^*(-z))
=
\tint \Res_z \tr(B^*(-z)A(z))
} \\
\displaystyle{
\vphantom{\Big(}
=
\tint \Res_z \tr(B(z+\partial)A(z))
\,.}
\end{array}
$$
In the first equality we used (a) (with $\lambda=0$),
in the second property we used the cyclic property of the trace,
and in the third equality we performed integration by parts.
\end{proof}
\begin{lemma}{\cite[Lem.5.6]{DSKVnew}}\label{lem:hn2}
Let $A(\partial)\in\mc V((\partial^{-1}))\otimes\End V$
and let $B(\partial)\in\mc V((\partial^{-1}))\otimes\End V$ be its $K$-th root,
i.e. $B^K(\partial)=A(\partial)$, for $K\in\mb Z\backslash\{0\}$.
Let $h_{n,B}\in\mc V$ be given by \eqref{eq:hn}.
Then, for $a\in\mc V$, $n\in\mb Z$, we have
\begin{equation}\label{eq:hn2}
\begin{split}
& \{{h_{n,B}}_\lambda a\}\big|_{\lambda=0}
=
-\Res_z\tr 
\{A(z+x)_x a\}\big(\big|_{x=\partial}B^{n-K}(z)\big)
\,,
\\
& \tint \{a_\lambda h_{n,B}\}\big|_{\lambda=0}
=
-\int \Res_w\tr 
\{a_\lambda A(w+x)\}\big|_{\lambda=0}\big(\big|_{x=\partial}B^{n-K}(w)\big)
\,.
\end{split}
\end{equation}
\end{lemma}
\begin{proof}
By Lemma \ref{20170404:lem2}(f) and (j), we have, for every $n\in\mb Z$,
\begin{equation}\label{20170405:eq1}
\{B^n(z)_\lambda a\}
=
\frac{n}{|n|}\sum_{\ell}
\big(\big|_{y=\partial}(B^\ell)^*(\lambda-z)\big)
\{B(z+x)_{\lambda+x+y}a\}
\big(\big|_{x=\partial}B^{n-1-\ell}(z)\big)
\,,
\end{equation}
where the sum over $\ell$ is $\sum_{\ell=0}^{n-1}$ if $n\geq1$, and $\sum_{\ell=n}^{-1}$ if $n\leq-1$.
We let $n=K$ in \eqref{20170405:eq1},
replace $z$ by $z+\partial$ and $\lambda$ by $\partial$ acting on $B^{n-K}(z)$,
and take $\res_z$ and $\tr$, to get:
\begin{equation}\label{20170405:eq2}
\begin{array}{l}
\displaystyle{
\vphantom{\Big(}
\res_z\tr
\{A(z+x)_x a\}
\big(\big|_{x=\partial}B^{n-K}(z)\big)
} \\
\displaystyle{
\vphantom{\Big(}
=
\frac{K}{|K|}
\sum_{\ell}
\res_z\tr
\big(\big|_{y=\partial}(B^\ell)^*(-z)\big)
\{B(z+x_1+x_2)_{x_1+x_2+y}a\}
} \\
\displaystyle{
\vphantom{\Big(}
\qquad\qquad\qquad
\times\big(\big|_{x_1=\partial}B^{K-1-\ell}(z+x_2)\big)
\big(\big|_{x_2=\partial}B^{n-K}(z)\big)
} \\
\displaystyle{
\vphantom{\Big(}
=
\frac{K}{|K|}
\sum_{\ell}
\res_z\tr
\big(\big|_{y=\partial}(B^\ell)^*(-z)\big)
\{B(z+x)_{x+y}a\}
\big(\big|_{x_1=\partial}B^{n-1-\ell}(z)\big)
} \\
\displaystyle{
\vphantom{\Big(}
=
\frac{K}{|K|}
\sum_{\ell}
\res_z\tr
\{B(z+x)_{x+y}a\}
\big(\big|_{x=\partial}B^{n-1-\ell}(z)\big)
\big(\big|_{y=\partial}(B^\ell)^*(-z)\big)
} \\
\displaystyle{
\vphantom{\Big(}
=
\frac{K}{|K|}
\sum_{\ell}
\res_z\tr
\{B(z+x+y)_{x+y}a\}
\big(\big|_{x=\partial}B^{n-1-\ell}(z+y)\big)
\big(\big|_{y=\partial}B^\ell(z)\big)
} \\
\displaystyle{
\vphantom{\Big(}
=
\frac{K}{|K|}
\sum_{\ell}
\res_z\tr
\{B(z+x)_{x}a\}
\big(\big|_{x=\partial}B^{n-1}(z)\big)
} \\
\displaystyle{
\vphantom{\Big(}
=
K
\res_z\tr
\{B(z+x)_{x}a\}
\big(\big|_{x=\partial}B^{n-1}(z)\big)
\,.}
\end{array}
\end{equation}
In the first equality of \eqref{20170405:eq2} we used \eqref{20170405:eq1},
in the second equality we used Lemma \ref{20170404:lem1}(a),
in the third equality we used the cyclic property of the trace,
in the fourth equality we used Lemma \ref{lem:hn1}(a),
in the fifth equality we used Lemma \ref{20170404:lem1}(a) again,
and in the last equality we used the obvious identity $\sum_\ell=|K|$.
With the same line of reasoning,
we get, by the definition \eqref{eq:hn} of $h_{n,B}$ and equation \eqref{20170405:eq1},
\begin{equation}\label{20170405:eq3}
\begin{array}{l}
\displaystyle{
\vphantom{\Big(}
\{{h_{n,B}}_\lambda a\}\big|_{\lambda=0}
=
-\frac{K}{n}
\res_z\tr
\{B^n(z)_\lambda a\}
\big|_{\lambda=0}
} \\
\displaystyle{
\vphantom{\Big(}
=
-\frac{K}{|n|}
\sum_{\ell}
\res_z\tr
\big(\big|_{y=\partial}(B^\ell)^*(-z)\big)
\{B(z+x)_{x+y}a\}
\big(\big|_{x=\partial}B^{n-1-\ell}(z)\big)
} \\
\displaystyle{
\vphantom{\Big(}
=
-\frac{K}{|n|}
\sum_{\ell}
\res_z\tr
\{B(z+x)_{x+y}a\}
\big(\big|_{x=\partial}B^{n-1-\ell}(z)\big)
\big(\big|_{y=\partial}(B^\ell)^*(-z)\big)
} \\
\displaystyle{
\vphantom{\Big(}
=
-\frac{K}{|n|}
\sum_{\ell}
\res_z\tr
\{B(z+x+y)_{x+y}a\}
\big(\big|_{x=\partial}B^{n-1-\ell}(z+y)\big)
\big(\big|_{y=\partial}B^\ell(z)\big)
} \\
\displaystyle{
\vphantom{\Big(}
=
-\frac{K}{|n|}
\sum_{\ell}
\res_z\tr
\{B(z+x)_{x}a\}
\big(\big|_{x=\partial}B^{n-1}(z)\big)
} \\
\displaystyle{
\vphantom{\Big(}
=
-K
\res_z\tr
\{B(z+x)_{x}a\}
\big(\big|_{x=\partial}B^{n-1}(z)\big)
\,,}
\end{array}
\end{equation}
where this time $\sum_\ell=|n|$.
Comparing the RHS's of equations \eqref{20170405:eq2} and \eqref{20170405:eq3},
we get the first equation in \eqref{eq:hn2}.

By Lemma \ref{20170404:lem2}(e) and (i), we have, for every $n\in\mb Z$,
\begin{equation}\label{20170405:eq4}
\{a_\lambda B^n(w)\}
=
\frac{n}{|n|}\sum_\ell
B^{n-\ell-1}(w+\lambda+\partial)
\{a_\lambda B(w+x)\}
\big(\big|_{x=\partial}B^\ell(w)\big)
\,.
\end{equation}
where, as before, $\sum_\ell$ stands for 
$\sum_{\ell=0}^{n-1}$ if $n\geq1$, and for $\sum_{\ell=n}^{-1}$ if $n\leq-1$.
We let $n=K$ in \eqref{20170405:eq4}, let $\lambda=0$,
replace $w$ by $w+\partial$ where $\partial$ acts on $B^{n-K}(w)$,
and take $\tint$, $\res_w$ and $\tr$, to get:
\begin{equation}\label{20170405:eq5}
\begin{array}{l}
\displaystyle{
\vphantom{\Big(}
\tint\res_w\tr
\{a_\lambda A(w+x)\}\big|_{\lambda=0}
\big(\big|_{x_2=\partial}B^{n-K}(w)\big)
} \\
\displaystyle{
\vphantom{\Big(}
=
\frac{K}{|K|}\sum_\ell
\tint\res_w\tr
B^{K-\ell-1}(w+\partial)
\{a_\lambda B(w+x_1+x_2)\}\big|_{\lambda=0}
} \\
\displaystyle{
\vphantom{\Big(}
\qquad\qquad\qquad\times
\big(\big|_{x_1=\partial}B^\ell(w+x_2)\big)
\big(\big|_{x_2=\partial}B^{n-K}(w)\big)
} \\
\displaystyle{
\vphantom{\Big(}
=
\frac{K}{|K|}\sum_\ell
\tint\res_w\tr
B^{K-\ell-1}(w+\partial)
\{a_\lambda B(w+x)\}\big|_{\lambda=0}
\big(\big|_{x=\partial}B^{n+\ell-K}(w)\big)
} \\
\displaystyle{
\vphantom{\Big(}
=
\frac{K}{|K|}\sum_\ell
\tint\res_w\tr
\{a_\lambda B(w\!+\!x\!+\!y)\}\big|_{\lambda=0}
\big(\big|_{x=\partial}B^{n+\ell-K}(w\!+\!y)\big)
\big(\big|_{y=\partial}B^{K\!-\!\ell-1}(w)\big)
} \\
\displaystyle{
\vphantom{\Big(}
=
\frac{K}{|K|}\sum_\ell
\tint\res_w\tr
\{a_\lambda B(w+x)\}\big|_{\lambda=0}
\big(\big|_{x=\partial}B^{n-1}(w)\big)
} \\
\displaystyle{
\vphantom{\Big(}
=
K
\tint\res_w\tr
\{a_\lambda B(w+x)\}\big|_{\lambda=0}
\big(\big|_{x=\partial}B^{n-1}(w)\big)
\,.}
\end{array}
\end{equation}
In the first equality of \eqref{20170405:eq5} we used \eqref{20170405:eq4},
in the second equality we used Lemma \ref{20170404:lem1}(a),
in the third equality we used Lemma \ref{lem:hn1}(b),
in the fourth equality we used Lemma \ref{20170404:lem1}(a) again,
and in the last equality we used the obvious identity $\sum_\ell=|K|$.
With the same line of reasoning,
we get, by the definition \eqref{eq:hn} of $h_{n,B}$ and equation \eqref{20170405:eq4},
\begin{equation}\label{20170405:eq6}
\begin{array}{l}
\displaystyle{
\vphantom{\Big(}
\tint\{a_\lambda{h_{n,B}}\}\big|_{\lambda=0}
=
-\frac{K}{n}
\tint\res_w\tr
\{a_\lambda B^n(w)\}
\big|_{\lambda=0}
} \\
\displaystyle{
\vphantom{\Big(}
=
-\frac{K}{|n|}
\sum_\ell
\tint\res_w\tr
B^{n-\ell-1}(w+\partial)
\{a_\lambda B(w+x)\}\big|_{\lambda=0}
\big(\big|_{x=\partial}B^\ell(w)\big)
} \\
\displaystyle{
\vphantom{\Big(}
=
-\frac{K}{|n|}
\sum_\ell
\tint\res_w\tr
\{a_\lambda B(w+x+y)\}\big|_{\lambda=0}
\big(\big|_{x=\partial}B^\ell(w+y)\big)
\big(\big|_{y=\partial}B^{n-\ell-1}(w)\big)
} \\
\displaystyle{
\vphantom{\Big(}
=
-\frac{K}{|n|}
\sum_\ell
\tint\res_w\tr
\{a_\lambda B(w+x)\}\big|_{\lambda=0}
\big(\big|_{x=\partial}B^{n-1}(w)\big)
} \\
\displaystyle{
\vphantom{\Big(}
=
-K
\tint\res_w\tr
\{a_\lambda B(w+x)\}\big|_{\lambda=0}
\big(\big|_{x=\partial}B^{n-1}(w)\big)
\,.}
\end{array}
\end{equation}
Comparing the RHS's of equations \eqref{20170405:eq5} and \eqref{20170405:eq6},
we get the second equation in \eqref{eq:hn2}.
\end{proof}
\begin{proof}[Proof of Theorem \ref{thm:hn}]
Suppose $B$ is a $K$-th root of $A$, $K\in\mb Z\backslash\{0\}$ 
and $C$ is an $H$-th root of $A$, $H\in\mb Z\backslash\{0\}$.
Applying the second equation in \eqref{eq:hn2} first,
and then the first equation in \eqref{eq:hn2}, we get
\begin{equation}\label{eq:hn-pr1}
\begin{split}
& \{\tint h_{m,B},\tint h_{n,C}\}
= 
\int \Res_z \Res_w (\tr\otimes\tr)
\{A(z+x)_x A(w+y)\}
\\
& \,\,\,\,\,\,\,\,\,\,\,\,\,\,\,\,\,\, \times
\Big(
\big(\big|_{x=\partial}B^{m-K}(z)\big)\otimes
\big(\big|_{y=\partial}C^{n-H}(w)\big)
\Big)\,.
\end{split}
\end{equation}
We can now use the generalized Adler identity \eqref{eq:adler-general} 
to rewrite the RHS of \eqref{eq:hn-pr1} as
\begin{equation}\label{eq:hn-pr2}
\begin{array}{l}
\displaystyle{
\vphantom{\Big(}
\alpha
\int \Res_z \Res_w (\tr\otimes\tr)(z-w-x_1-y)^{-1}
} \\
\displaystyle{
\vphantom{\Big(}
\,\,\,\times
\Omega
\Big(
A(w+x_1+x_2+y)
\big(\big|_{x_2=\partial}B^{m-K}(z)\big)
\otimes
\big(\big|_{x_1=\partial}A^*(-z)\big)
\big(\big|_{y=\partial}C^{n-H}(w)\big)
\Big)
} \\
\displaystyle{
\vphantom{\Big(}
-\alpha
\int \Res_z \Res_w (\tr\otimes\tr)(z-w-y_1-y_2)^{-1}
} \\
\displaystyle{
\vphantom{\Big(}
\,\,\,\times
\Omega
\Big(
A(z+x)
\big(\big|_{x=\partial}B^{m-K}(z)\big)
\otimes
\big(\big|_{y_1=\partial}A(w+y_2)\big)
\big(\big|_{y_2=\partial}C^{n-H}(w)\big)
\Big)
} 
\end{array}
\end{equation}

\begin{equation}\label{eq:hn-pr2}
\begin{array}{l}
\displaystyle{
\vphantom{\Big(}
-\beta
\int \Res_z \Res_w (\tr\otimes\tr)(z+w+x_1+x_2+y)^{-1}
} \\
\displaystyle{
\vphantom{\Big(}
\,\,\,\times
\big(\id\!\otimes\! A(w\!\!+\!\!x_1\!\!+\!\!x_2\!\!+\!\!y)\big)
\Omega^\dagger
\Big(\!
\big(\big|_{x_1=\partial}A(z\!\!+\!\!x_2)\big)
\big(\big|_{x_2=\partial}B^{m-K}(z)\big)
\otimes\!
\big(\big|_{y=\partial}C^{n-H}(w)\big)
\!\Big)
} \\
\displaystyle{
\vphantom{\Big(}
+\beta
\int \Res_z \Res_w (\tr\otimes\tr)(z+w+x+y_1+y_2)^{-1}
} \\
\displaystyle{
\vphantom{\Big(}
\,\,\,\times
\big(A^*(-z)\otimes\id\big)
\Omega^\dagger
\Big(
\big(\big|_{x=\partial}B^{m-K}(z)\big)
\otimes 
\big(\big|_{y_1=\partial}A(w+y_2)\big)
\big(\big|_{y_2=\partial}C^{n-H}(w)\big)
\Big)
} \\
\displaystyle{
\vphantom{\Big(}
+\gamma
\int \Res_z \Res_w (\tr\otimes\tr)(x_1+x_2)^{-1}
} \\
\displaystyle{
\vphantom{\Big(}
\,\,\,\times
\big(\big|_{x_1=\partial}(A^*(-z)-A(z+x_2))\big)
\big(\big|_{x_2=\partial}B^{m-K}(z)\big)
} \\
\displaystyle{
\vphantom{\Big(}
\,\,\,\,\,\,\,\,\,\otimes 
\big(A(w+x_1+x_2+y)-A(w+y)\big)
\big(\big|_{y=\partial}C^{n-H}(w)\big)
\,.}
\end{array}
\end{equation}
Note that (cf. Remark \ref{rem:regular}) the coefficient 
of $\alpha$ (resp. $\beta$) in \eqref{eq:adler-general}
is regular in $z-w-\lambda-\partial$ (resp. $z+w+\partial$).
Hence, if we expand each term of \eqref{eq:hn-pr2} in the domain $|z|>|w|$
(or, equivalently, $|z|<|w|$), the result is unchanged.
We can use
Lemma \ref{lemma:29032017}(b),
to rewrite the first term in the RHS of \eqref{eq:hn-pr2} as
\begin{equation}\label{eq:hn-pr3}
\begin{array}{l}
\displaystyle{
\vphantom{\Big(}
\alpha
\int \Res_z \Res_w \tr
\iota_z(z-w-x_1-y)^{-1}
} \\
\displaystyle{
\vphantom{\Big(}
\,\,\,\times
A(w+x_1+x_2+y)
\big(\big|_{x_2=\partial}B^{m-K}(z)\big)
\big(\big|_{x_1=\partial}A^*(-z)\big)
\big(\big|_{y=\partial}C^{n-H}(w)\big)
\,.}
\end{array}
\end{equation}
By \eqref{20170406:eq1} and \eqref{eq:positive}
and the identity $A=B^K$,
we can then rewrite \eqref{eq:hn-pr3} as
$$
\alpha
\int 
\Res_w \tr
A(w+x+y)
\big(\big|_{x=\partial}B^{m}(w+y)\big)_+
\big(\big|_{y=\partial}C^{n-H}(w)\big)
\,,
$$
or, equivalently, as
\begin{equation}\label{eq:hn-pr4}
\alpha
\int 
\Res_w \tr
A(w+\partial)
B^{m}(w+\partial)_+
C^{n-H}(w)
\,.
\end{equation}
Furthermore, by Lemma \ref{lem:hn1}(b) and the identity $A=C^H$, \eqref{eq:hn-pr4} becomes
\begin{equation}\label{eq:hn-pr5}
\alpha
\int 
\Res_w \tr
B^{m}(w+\partial)_+
C^{n}(w)
\,.
\end{equation}
Next, let us consider the second term of \eqref{eq:hn-pr2},
which, by Lemma \ref{20170404:lem1} and the identities $A=B^K=C^H$,
can be rewritten as
\begin{equation}\label{eq:hn-pr6}
-\alpha
\int \Res_z \Res_w (\tr\otimes\tr)\iota_z(z-w-y)^{-1}
\Omega
\Big(
B^{m}(z)
\otimes
\big(\big|_{y=\partial}C^{n}(w)\big)
\Big)
\,.
\end{equation}
By Lemma \ref{lemma:29032017}(b),
we can rewrite \eqref{eq:hn-pr6} as
\begin{equation}\label{eq:hn-pr7}
-\alpha
\int \Res_z \Res_w \tr \iota_z(z-w-y)^{-1}
B^{m}(z)
\big(\big|_{y=\partial}C^{n}(w)\big)
\,,
\end{equation}
and by \eqref{eq:positive}, \eqref{eq:hn-pr7} is equal to
\begin{equation}\label{eq:hn-pr8}
-\alpha
\int \Res_w \tr
B^{m}(w+\partial)_+
C^{n}(w)
\,.
\end{equation}
Combining \eqref{eq:hn-pr5} and \eqref{eq:hn-pr8},
we conclude that the coefficient of $\alpha$ in \eqref{eq:hn-pr2} vanishes.
Next, let us consider the third term in \eqref{eq:hn-pr2},
which, by Lemma \ref{20170404:lem1}(a)
and the identity $A=B^K$, can be rewritten as
\begin{equation}\label{eq:hn-pr9}
\begin{array}{l}
\displaystyle{
\vphantom{\Big(}
-\beta
\int \Res_z \Res_w (\tr\otimes\tr)\iota_z(z+w+x+y)^{-1}
} \\
\displaystyle{
\vphantom{\Big(}
\,\,\,\times
\big(\id\otimes A(w+x+y)\big)
\Omega^\dagger
\Big(
\big(\big|_{x=\partial}B^{m}(z)\big)
\otimes
\big(\big|_{y=\partial}C^{n-H}(w)\big)
\!\Big)
\,.}
\end{array}
\end{equation}
By Lemma \ref{lemma:29032017}(d),
\eqref{eq:hn-pr9} is equal to
\begin{equation}\label{eq:hn-pr10}
-\beta
\int \Res_z \Res_w \tr \iota_z(z+w+x+y)^{-1}
A(w+x+y)
\big(\big|_{x=\partial}B^{m}(z)\big)^\dagger
\big(\big|_{y=\partial}C^{n-H}(w)\big)
\,.
\end{equation}
We then use equation \eqref{eq:positive} and \eqref{20170406:eq1}, 
to rewrite \eqref{eq:hn-pr10} as
$$
-\beta
\int \Res_w \tr 
A(w+x+y)
\big(\big|_{x=\partial}(B^{m})^*(w+y)\big)_+^\dagger
\big(\big|_{y=\partial}C^{n-H}(w)\big)
\,,
$$
or, equivalently, as
\begin{equation}\label{eq:hn-pr10bis}
-\beta
\int \Res_w \tr 
A(w+\partial)
\big((B^{m})^*(w+\partial)\big)_+^\dagger
C^{n-H}(w)
\,.
\end{equation}
Furthermore, we can use Lemma \ref{lem:hn1}(b) and the identity $A=C^H$, 
to rewrite \eqref{eq:hn-pr10bis} as
\begin{equation}\label{eq:hn-pr11}
-\beta
\int \Res_w \tr 
\big((B^{m})^*(w+\partial)\big)_+^\dagger
C^{n}(w)
\,.
\end{equation}
Next, let us consider the fourth term in \eqref{eq:hn-pr2},
which, by Lemma \ref{20170404:lem1}(a)
and the identity $A=C^H$, can be rewritten as
\begin{equation}\label{eq:hn-pr12}
\begin{array}{l}
\displaystyle{
\vphantom{\Big(}
\beta
\int \Res_z \Res_w (\tr\otimes\tr) \iota_z(z+w+x+y)^{-1}
} \\
\displaystyle{
\vphantom{\Big(}
\,\,\,\times
\big(A^*(-z)\otimes\id\big)
\Omega^\dagger
\Big(
\big(\big|_{x=\partial}B^{m-K}(z)\big)
\otimes 
\big(\big|_{y=\partial}C^{n}(w)\big)
\Big)
\,.}
\end{array}
\end{equation}
We can use Lemma \ref{lemma:29032017} to rewrite \eqref{eq:hn-pr12} as
\begin{equation}\label{eq:hn-pr13}
\beta
\int \Res_z \Res_w \tr \iota_z(z+w+x+y)^{-1}
\big(\big(\big|_{x=\partial}B^{m-K}(z)\big)A^*(-z)\big)^\dagger
\big(\big|_{y=\partial}C^{n}(w)\big)
\,.
\end{equation}
Using \eqref{20170406:eq1} and \eqref{eq:positive}, we can rewrite \eqref{eq:hn-pr13} as
\begin{equation}\label{eq:hn-pr14}
\beta
\int \Res_w \tr 
\iota_z(z+w+x+y)^{-1}
\big(\big(\big|_{x=\partial}(B^{m-K})^*(w+y)\big)A^*(w+x+y)\big)_+^\dagger
\big(\big|_{y=\partial}C^{n}(w)\big)
\,.
\end{equation}
By Lemma \ref{20170404:lem1}(b) and the identity $A=B^K$, \eqref{eq:hn-pr14} is equal to
\begin{equation}\label{eq:hn-pr15}
\beta
\int \Res_w \tr 
\big((B^{m})^*(w+\partial)\big)_+^\dagger
C^{n}(w)
\,.
\end{equation}
Combining \eqref{eq:hn-pr11} and \eqref{eq:hn-pr15},
we conclude that the coefficient of $\beta$ in \eqref{eq:hn-pr2} vanishes.
Finally, let us consider the last term of \eqref{eq:hn-pr2},
which, by Lemma \ref{20170404:lem1}(a) and the identities $A=B^K=C^H$, 
can be rewritten as
\begin{equation}\label{eq:hn-pr16}
\begin{array}{l}
\displaystyle{
\vphantom{\Big(}
\gamma
\int 
\Res_w \tr
\Big(
A(w+x+y)\big(\big|_{y=\partial}C^{n-H}(w)\big)
-
C^{n}(w)
\Big)
x^{-1}
} \\
\displaystyle{
\vphantom{\Big(}
\,\,\,\times
\Big|_{x=\partial}
\Res_z \tr
\Big(
B^{m-K}(z)
A^*(-z)
-
B^{m}(z)
\Big)
\,.}
\end{array}
\end{equation}
By Lemma \ref{lem:hn1}(a), we have
\begin{equation}\label{20170406:eq2}
\Res_z
B^{m-K}(z)
A^*(-z)
=
\res_zB^{m-K}(z+\partial)A(z)
=
\res_zB^{m}(z)
\,.
\end{equation}
Hence \eqref{eq:hn-pr16} vanishes. 
In conclusion, also the coefficient of $\gamma$ in \eqref{eq:hn-pr2}
vanishes, proving (a).

We are left to prove part (b).
We have
\begin{equation}\label{eq:proofb-1}
\begin{array}{l}
\displaystyle{
\vphantom{\Big(}
\{\tint h_{n,B},A(w)\}
=
\{{h_{n,B}}_\lambda A(w)\}\big|_{\lambda=0}
} \\
\displaystyle{
\vphantom{\Big(}
=
-\Res_z (\tr\otimes1)
\{A(z+x)_x A(w)\}
\big(\big|_{x=\partial}B^{n-K}(z)\otimes\id\big)
} \\
\displaystyle{
\vphantom{\Big(}
=
-\alpha
\Res_z\! \iota_z(z\!-\!w\!-\!y)^{-1} (\tr\otimes1)
\Omega
\Big(
A(w\!+\!x\!+\!y)
\big(\big|_{x=\partial}B^{n-K}(z)\big)
\!\otimes\!
\big(\big|_{y=\partial}A^*(-\!z)\big)
\Big)
} \\
\displaystyle{
\vphantom{\Big(}
\,\,\,\,
+\alpha
\Res_z \iota_z(z-w-y)^{-1} (\tr\otimes1)
\Omega
\Big(
B^{n}(z)
\otimes
\big(\big|_{y=\partial}A(w)\big)
\Big)
} \\
\displaystyle{
\vphantom{\Big(}
\,\,\,\,
+\beta
\Res_z (\tr\otimes1)
\iota_z(z+w+x)^{-1}
(\id\otimes A(w+x))
\Omega^\dagger
\big(\big|_{x=\partial}B^{n}(z)\otimes\id\big)
} \\
\displaystyle{
\vphantom{\Big(}
\,\,\,\,
-\beta
\Res_z (\tr\otimes1)
\iota_z(z\!+\!w\!+\!x\!+\!y)^{-1}
(A^*(-z)\otimes\id)
\Omega^\dagger
\big(\big|_{x=\partial}B^{n-K}\!(z)\otimes \big|_{y=\partial}A(w)\big)
} \\
\displaystyle{
\vphantom{\Big(}
\,\,\,\,
-\gamma
\Res_z (\tr\otimes1)
(x+y)^{-1}
\big(\big|_{y=\partial}\big(A^*(-z)-A(z+x)\big)\big)
\big(\big|_{x=\partial}B^{n-K}(z)\big)
} \\
\displaystyle{
\vphantom{\Big(}
\,\,\,\,\,\,\,\,\,\,\,\,
\otimes
\big(A(w+x+y)-A(w)\big)
} \\
\displaystyle{
\vphantom{\Big(}
=
-\alpha
\Res_z \iota_z(z-w-y)^{-1} 
A(w+x+y)
\big(\big|_{x=\partial}B^{n-K}(z)\big)
\big(\big|_{y=\partial}A^*(-z)\big)
} \\
\displaystyle{
\vphantom{\Big(}
\,\,\,\,
+\alpha
\Res_z \iota_z(z-w-y)^{-1}
B^{n}(z)
\big(\big|_{y=\partial}A(w)\big)
} \\
\displaystyle{
\vphantom{\Big(}
\,\,\,\,
+\beta
\Res_z 
\iota_z(z+w+x)^{-1}
A(w+x)
\big(\big|_{x=\partial}B^{n}(z)\big)^\dagger
} \\
\displaystyle{
\vphantom{\Big(}
\,\,\,\,
-\beta
\Res_z
\iota_z(z+w+x+y)^{-1}
\big(\big(\big|_{x=\partial}B^{n-K}(z)\big)A^*(-z)\big)^\dagger
\big(\big|_{y=\partial}A(w)\big)
} \\
\displaystyle{
\vphantom{\Big(}
\,\,\,\,
-\gamma
\big(A(w+\partial)-A(w)\big)\partial^{-1}
\tr\Res_z\big(A^*(-z)B^{n-K}(z)-B^{n}(z)\big)
} \\
\displaystyle{
\vphantom{\Big(}
=
-\alpha
A(w+\partial)
\Res_z 
B^{n-K}(z)
\iota_z(z-w-\partial)^{-1} 
A^*(-z)
} \\
\displaystyle{
\vphantom{\Big(}
\,\,\,\,
+\alpha
\Res_z
B^{n}(z)
\iota_z(z-w-\partial)^{-1}
A(w)
} \\
\displaystyle{
\vphantom{\Big(}
\,\,\,\,
+\beta
A(w+\partial)
\Res_z 
\iota_z(z+w+\partial)^{-1}
\big(B^{n}(z)\big)^\dagger
} \\
\displaystyle{
\vphantom{\Big(}
\,\,\,\,
-\beta
\Res_z
\iota_z(z+w+x+y)^{-1}
\big(\big(\big|_{x=\partial}B^{n-K}(z)\big)A^*(-z)\big)^\dagger
\big(\big|_{y=\partial}A(w)\big)
} \\
\displaystyle{
\vphantom{\Big(}
=
-\alpha
A(w+\partial)
B^{n}(w)_+
+\alpha
B^{n}(w+\partial)_+
A(w)
} \\
\displaystyle{
\vphantom{\Big(}
\,\,\,\,
+\beta
A(w+\partial)
\big((B^{n})^*(w)_+\big)^\dagger
-\beta
\big((B^{n})^*(w+\partial)_+\big)^\dagger
A(w)
\,.}
\end{array}
\end{equation}
In the second equality we used 
the first equation in \eqref{eq:hn2},
in the third equality we used the generalized Adler identity \eqref{eq:adler-general}
and some algebraic manipulations based on the identity $A=B^K$,
in the third equality we used Lemma \ref{lemma:29032017}(a) and (c),
in the fourth equality we did some algebraic manipulations and we used equation \eqref{20170406:eq2},
in the fifth equality we used \eqref{eq:positive} and Lemma \ref{20170404:lem1}.
This proves \eqref{eq:hierarchy} and completes the proof of the Theorem.
\end{proof}

\subsection{Integrable Hamiltonian hierarchies associated to $\mc W$-algebras}

\begin{corollary}\label{cor:w1}
Let $(\mf g,V)$ be as in the three Cases 1., 2. or 3. in Section \ref{sec:4.2}
and, for an arbitrary nilpotent element $f\in\mf g$,
consider the classical $\mc W$-algebra $\mc W(\mf g,f)$
with PVA $\lambda$-bracket $\{\cdot\,_\lambda\,\cdot\}^{\mc W}_0$
given by \eqref{20120511:eq3} (with $\epsilon=0$).
Let $T\in(\End V)[D]$ have canonical decomposition $T=IJ$,
and assume, in Case 3., that $T^\dagger=\pm T$.
Consider the operator (cf. \eqref{eq:Lw}):
\begin{equation}\label{eq:L0}
\begin{array}{l}
\displaystyle{
\vphantom{\Big(}
L(\partial)=L(\mf g,f,T)(\partial)
=
\Big(J\big(\partial\id_V+F+\sum_{i\in I_f}w_iU^i\big)^{-1}I\Big)^{-1}
} \\
\displaystyle{
\vphantom{\Big(}
\qquad\qquad\qquad
\,\in\mc W(\mf g,f)((\partial^{-1}))\otimes\End(\im T)
\,.}
\end{array}
\end{equation}
Then,
the elements $h_{n,B}\in\mc W(\mf g,f)$, $n\in\mb Z$, $B$ a root of $L$,
defined by \eqref{eq:hn},
are Hamiltonian densities in involution, i.e. \eqref{eq:invol} holds,
and the corresponding hierarchy of Hamiltonian equation takes the form \eqref{eq:hierarchy},
with $\alpha,\beta$ given by Table \eqref{table}.
\end{corollary}
\begin{proof}
It follows by Corollary \ref{20170404:cor} and Theorem \ref{thm:hn}.
\end{proof}

\begin{remark}\label{rem:modified}
Recall from \cite{DSKV16a} that we have an injective Poisson vertex algebra homomorphism
$\mu:\mc W(\mf g,f)\to\mc V:=\mc V(\mf g_0)\otimes \mc F(\mf g_{\frac12})$, where $\mc V(\mf g_0)$ is the affine Poisson vertex
algebra over the subalgebra $\mf g_0\subset\mf g$, whose $\lambda$-bracket on generators is given by equation \eqref{lambda}, and 
$\mc F(\mf g_{\frac12})$ is the algebra of differential polynomials $S(\mb F[\partial]\mf g_{\frac12})$ endowed with the
$\lambda$-bracket defined on generators by $\{a_\lambda b\}^{ne}=(f|[a,b])$, for every $a,b\in\mf g_{\frac12}$.
The map $\mu$ is called generalized Miura map.

By applying $\mu$ to both sides of \eqref{eq:L0}, and using the fact that it is a differential algebra homomorphism, we get the identity
\begin{equation}\label{eq:Lmod}
\mu(L(\partial))=\Big(J\big(\partial\id_V+F+\sum_{i\in I_0\sqcup I_{\frac12}}u_iU^i\big)^{-1}I\Big)^{-1}
:=L^{mod}(\partial)\in \mc V((\partial^{-1}))\otimes\End(\im T)
\,.
\end{equation}
The above identity yields another definition of the generalized Miura map. By Corollary \ref{cor:w1} the elements
$\bar h_{n,B}=\mu(h_{n,B})\in\mc V$ are Hamiltonian densities in involution with respect
to the Lie algebra structure induced by the PVA $\mc V$, and the corresponding hierarchy of Hamiltonian equations takes the
form \eqref{eq:hierarchy} with $L^{mod}$ in place of $L$ and $\mu(B)$ in place of $B$ (note that $\mu(B)$ is a root of $L^{mod}$
since $\mu$ is a differential algebra homomorphism). 
\end{remark}

\section{Operators of generalized bi-Adler type and corresponding integrable hierarchies}
\label{sec:bi}

\subsection{Generalized bi-Adler identity}
\label{sec:bi1}

Recall that a bi-PVA $\mc V$ is a differential algebra endowed with a pencil of PVA $\lambda$-brackets
$\{\cdot\,_\lambda\,\cdot\}_\epsilon=\{\cdot\,_\lambda\,\cdot\}_0+\epsilon\{\cdot\,_\lambda\,\cdot\}_1$,
$\epsilon\in\mb F$.
\begin{definition}\label{def:bi-adler}
Let $S\in\End V$ and let $A(\partial)\in\mc V((\partial^{-1}))\otimes\End V$,
where $\mc V$ is a bi-PVA.
We say that $A(\partial)$ is of \emph{generalized} $S$-\emph{Adler type} if 
$$
A_\epsilon(\partial)=A(\partial)+\epsilon S
$$ 
is of generalized Adler type
w.r.t. the PVA $\lambda$-bracket $\{\cdot\,_\lambda\,\cdot\}_\epsilon$,
for every $\epsilon\in\mb F$
(and with values of the parameters $\alpha,\beta,\gamma$ independent of $\epsilon$).
Equivalently, $A(\partial)$ satisfies the generalized Adler identity \eqref{eq:adler-general}
w.r.t. the PVA $\lambda$-bracket $\{\cdot\,_\lambda\,\cdot\}_0$, and 
\begin{equation}\label{eq:bi-adler-general}
\begin{array}{l}
\displaystyle{
\vphantom{\Big(}
\{A(z)_\lambda A(w)\}_1
=
\alpha\Omega
\big(S\otimes(z-w-\lambda-\partial)^{-1}(A^*(\lambda-z)-A(w))\big)
} \\
\displaystyle{
\vphantom{\Big(}
-\alpha\Omega
\big((z-w-\lambda)^{-1}(A(z)-A(w+\lambda))\otimes S\big)
} \\
\displaystyle{
\vphantom{\Big(}
+\beta
(S\otimes\id)\Omega^\dagger
\big(\id\otimes(z+w+\partial)^{-1}(A(w)-A^*(z))\big)
} \\
\displaystyle{
\vphantom{\Big(}
-\beta
\big(\id\otimes(z+w)^{-1}(A(w+\lambda)-A(\lambda-z))\big)
\Omega^\dagger(S\otimes\id)
\,.}
\end{array}
\end{equation}
As usual, for $\beta\neq0$ 
we assume that $V$ carries a symmetric or skewsymmetric non-degenerate
bilinear form $\langle\cdot\,|\,\cdot\rangle$,
$\Omega^\dagger$ is given by \eqref{Omega-tau},
and we assume that 
\begin{equation}\label{eq:bi-dagger}
(A^\star(\partial))^\dagger=\eta A(\partial)
\,\,\text{ and }\,\,
S^\dagger=\eta S
\,\,,\,\,\text{ where } \eta\in\{\pm1\}
\,.
\end{equation}
We also say that $A(\partial)$ is of \emph{generalized bi-Adler type}
if it is of generalized $S$-Alder type, with $S=\id$.
\end{definition}
\begin{example}\label{ex:S-adler}
Consider the operator (cf. \eqref{Aepsilon})
$$
A(\partial)=\partial\id+\sum_{i\in I}u_iU^i\in\mc V(\mf g)((\partial^{-1}))\otimes\End\mb F^N
\,,
$$
in the three Cases 1., 2. and 3. of Section \ref{sec:4.2},
i.e. for $\mf g=\mf{gl}_N,\mf{sl}_N,\mf{so}_N$ or $\mf{sp}_N$.
By equations \eqref{eq:adler-glN}, \eqref{eq:adler-slN} and \eqref{eq:adler-soN}, 
$A(\partial)$ is of $S$-Adler type w.r.t. the affine bi-PVA structure on $\mc V(\partial)$
defined by \eqref{lambda}
(where $S=\varphi(s)$ and $s\in\mf g_d$).
The corresponding values of the parameters $\alpha,\beta,\gamma$ are given by Table \eqref{table}.
\end{example}

Recall by \cite[Thm.4.5]{DSKVnew}
that, if $S\in\End V$ has canonical decomposition $S=IJ$,
where $I:\,\im S\hookrightarrow V$ is the inclusion map
and $J=S:\,V\twoheadrightarrow\im S$,
then the following identity holds 
\begin{equation}\label{20170421:eq1b}
|A+\epsilon S|_{I,J}=|A|_{I,J}+\epsilon\id
\,\,,\,\,\,\,\epsilon\in\mb F
\end{equation}
provided that the above generalized quasideterminants exist.
\begin{example}\label{ex:bi-adler}
Let $(\mf g,V)$ be as in Cases 1., 2. or 3. from Section \ref{sec:4.2}.
In the notation of Section \ref{sec:3},
assume that $D=d$, and let $T=S$.
Consider the $\End(\im S)$-valued pseudodifferential operator
over the bi-PVA $\mc W_\epsilon(\mf g,f,s)$, $\epsilon\in\mb F$,
(cf. \eqref{eq:L})
$$
L(\partial)
=
|\rho A(\partial)|_{I,J}
=
\big(J(\partial\id+F+\sum_{i\in I_{\leq\frac12}}u_iU^i)^{-1}I\big)^{-1}
\,.
$$
By Example \ref{ex:S-adler} and equation \eqref{20170421:eq1b},
$L(\partial)$ is an operator of bi-Adler type,
with the same values of the parameters $\alpha,\beta,\gamma$ as in Table \eqref{table}.
\end{example}

\subsection{Integrable bi-Hamiltonian hierarchy for a generalized bi-Adler type operator}
\label{sec:bi3}

\begin{theorem}\label{thm:main-bi}
Let $A(\partial)\in\mc V((\partial^{-1}))\otimes\End V$
be an $\End V$-valued pseudodifferential operator over the bi-PVA $\mc V$,
of generalized bi-Adler type.
Assume also that $A(\partial)$ is invertible in $\mc V((\partial^{-1}))\otimes\End V$.
Let $B(\partial)\in\mc V((\partial^{-1}))\otimes\End V$
be a $K$-th root of $A$ (i.e. $A(\partial)=B(\partial)^K$ for $K\in\mb Z\backslash\{0\}$).
Then, the elements $h_{n,B}\in\mc V$, $n\in\mb Z$, given by \eqref{eq:hn},
satisfy the following generalized Lenard-Magri recurrence relation:
\begin{equation}\label{eq:lenard}
\{\tint h_{n,B},A(w)\}_1
=
\{\tint h_{n-K,B},A(w)\}_0
\,,\,\,n\in\mb Z\,.
\end{equation}
Hence, \eqref{eq:hierarchy}
is a compatible hierarchy of bi-Hamiltonian equations over the bi-PVA
subalgebra $\mc V_1\subset\mc V$
generated by the coefficients of the entries of $A(z)$.
Moreover, all the Hamiltonian functionals $\tint h_{n,C}$, $n\in\mb Z$,
$C$ a root of $A$, are integrals of motion in involution of all the equations of this hierarchy.
\end{theorem}
\begin{proof}
By the first equation in \eqref{eq:hn2}, we have
\begin{equation}\label{eq:proof-bi1}
\begin{array}{l}
\displaystyle{
\vphantom{\Big(}
\{\tint h_{n,B},A(w)\}_1
=
\{{h_{n,B}}_\lambda A(w)\}_1\big|_{\lambda=0}
} \\
\displaystyle{
\vphantom{\Big(}
=
-\Res_z (\tr\otimes1)
\{A(z+x)_x A(w)\}_1
\big(\big|_{x=\partial}B^{n-K}(z)\otimes\id\big)
\,.}
\end{array}
\end{equation}
We then use the generalized bi-Adler identity \eqref{eq:bi-adler-general} (with $S=\id$)
to rewrite the RHS of \eqref{eq:proof-bi1} as
\begin{equation}\label{eq:proof-bi2}
\begin{array}{l}
\displaystyle{
\vphantom{\Big(}
-\alpha
\Res_z (z-w-y)^{-1}
(\tr\otimes1)
\Omega
\big(\big|_{x=\partial}B^{n-K}(z)\big)
\otimes
\big(\big|_{y=\partial}(A^*(-z)-A(w))\big)
} \\%
\displaystyle{
\vphantom{\Big(}
+\alpha
\Res_z (z-w)^{-1}
(\tr\otimes1)
\Omega
\big((A(z+x)-A(w+x))
\big(\big|_{x=\partial}B^{n-K}(z)\big)
\otimes \id\big)
} \\%
\displaystyle{
\vphantom{\Big(}
-\beta
\Res_z (z\!+\!w\!+\!x\!+\!y)^{-1}
(\tr\otimes1)
\Omega^\dagger
\big(
\big(\big|_{x=\partial}B^{n-K}(z)\big)
\otimes
\big(\big|_{y=\partial}(A(w)-A^*(z\!+\!x))\big)
\big)
} \\%
\displaystyle{
\vphantom{\Big(}
+\beta
\Res_z (z\!+\!w\!+\!x)^{-1}
(\tr\otimes1)
\big(\id\otimes(A(w\!+\!x)-A(-z))\big)
\Omega^\dagger
\big(\big(\big|_{x=\partial}B^{n-K}(z)\big)\otimes\id\big)
.}
\end{array}
\end{equation}
Next, we use Lemma \ref{lemma:29032017} to rewrite \eqref{eq:proof-bi2} as
\begin{equation}\label{eq:proof-bi3}
\begin{array}{l}
\displaystyle{
\vphantom{\Big(}
-\alpha
\Res_z (z-w-y)^{-1}
B^{n-K}(z)
\big(\big|_{y=\partial}(A^*(-z)-A(w))\big)
} \\%
\displaystyle{
\vphantom{\Big(}
+\alpha
\Res_z (z-w)^{-1}
(A(z+x)-A(w+x))
\big(\big|_{x=\partial}B^{n-K}(z)\big)
} \\%
\displaystyle{
\vphantom{\Big(}
-\beta
\Res_z (z+w+x+y)^{-1}
\big(\big|_{x=\partial}B^{n-K}(z)\big)^\dagger
\big(\big|_{y=\partial}(A(w)-A^*(z+x))\big)
} \\%
\displaystyle{
\vphantom{\Big(}
+\beta
\Res_z (z+w+x)^{-1}
(A(w+x)-A(-z))
\big(\big|_{x=\partial}B^{n-K}(z)\big)^\dagger
\,.}
\end{array}
\end{equation}
Furthermore, using equation \eqref{eq:positive} we can rewrite \eqref{eq:proof-bi3} as
\begin{equation}\label{eq:proof-bi4}
\begin{array}{l}
\displaystyle{
\vphantom{\Big(}
-\alpha
B^{n}(w)_+
+\alpha
B^{n-K}(w+\partial)_+
A(w)
+\alpha
B^{n}(w)_+
-\alpha
A(w+\partial)
B^{n-K}(w)_+
} \\%
\displaystyle{
\vphantom{\Big(}
-\beta
\big((B^{n-K})^*(w+\partial)_+\big)^\dagger
A(w)
+\beta\eta
\big((B^{n})^*(w)_+\big)^\dagger
} \\%
\displaystyle{
\vphantom{\Big(}
+\beta
A(w+\partial)
\big((B^{n-K})^*(w)_+\big)^\dagger
-\beta\eta
\big((B^{n})^*(w)_+\big)^\dagger
} \\%
\displaystyle{
\vphantom{\Big(}
=
[\alpha (B^{n-K})_+ - \beta(B^{n-K})^{\star\dagger}_+,A](w)
\,,}
\end{array}
\end{equation}
where we used the identities $A=B^K$,
assumption \eqref{eq:bi-dagger} and Lemma \ref{20170411:lem1}.
Comparing \eqref{eq:proof-bi4} and \eqref{eq:hierarchy}, we get equation \eqref{eq:lenard},
completing the proof.
\end{proof}

\subsection{Integrable bi-Hamiltonian hierarchies associated to \texorpdfstring{$\mc W$}{W}-algebras}

\begin{corollary}\label{cor:w1bis}
Let $(\mf g,V)$ be as in the three Cases 1., 2. or 3. in Section \ref{sec:4.2}.
Let $f\in\mf g$ be a nilpotent element
and assume that the depth does not drop from $\End V$ to $\mf g$,
i.e. $D=d$,
where $d$ largest $\ad x$-eigenvalue in $\mf g$ 
and $D$ is the largest $\ad x$-eigenvalue in $\End V$.
For an element $s\in\mf g_d$,
consider the pencil of classical $\mc W$-algebras $\mc W_\epsilon(\mf g,f,s)$
with PVA $\lambda$-bracket $\{\cdot\,_\lambda\,\cdot\}^{\mc W}_\epsilon$, $\epsilon\in\mb F$,
given by \eqref{20120511:eq3}.
Let $S=\varphi(s)\in\End V$ have canonical decomposition $T=IJ$,
and consider the operator $L(\mf g,f,S)(\partial)$ defined in \eqref{eq:L0}.
Then,
the Hamiltonian densities $h_{n,B}\in\mc W(\mf g,f)$, $n\in\mb Z$, $B$ a root of $L$,
defined by \eqref{eq:hn}, are in involution and they satisfy
the generalized Lenard-Magri recurrence relation \eqref{eq:lenard}
(with $A$ replaced by $L$).
Hence, \eqref{eq:hierarchy} (with $\alpha,\beta$ given by Table \eqref{table})
is a compatible hierarchy of bi-Hamiltonian equations over the bi-PVA
subalgebra $\mc V_1\subset\mc W_\epsilon(\mf g,f,s)$, $\epsilon\in\mb F$,
generated by the coefficients of the entries of $L(z)$.
Moreover, all the Hamiltonian functionals $\tint h_{n,C}$, $n\in\mb Z$,
$C$ a root of $A$, are integrals of motion in involution of all the equations of this hierarchy.
\end{corollary}
\begin{proof}
It follows by Corollary \ref{20170404:cor} and Theorem \ref{thm:main-bi}.
\end{proof}

\section{Product of operators of generalized Adler type and corresponding integrable hierarchies}
\label{sec:7}

\begin{theorem}\label{thm:hn-product}
Let $A_1(\partial),\dots,A_s(\partial)\in\mc V((\partial^{-1}))\otimes\End V$
be $\End V$-valued pseudodifferential operators over the PVA $\mc V$.
Assume that $A_1(\partial),\dots,A_s(\partial)$ 
have pairwise vanishing Poisson $\lambda$-brackets:
\begin{equation}\label{20170411:eq1}
\{A_\ell(z)_\lambda A_r(w)\}
=
0
\,\,\text{ for all }\,\, \ell\neq r
\,,
\end{equation}
and assume that they are all operators of generalized Adler type
with values of the parameters $\alpha_\ell,\beta_\ell,\gamma_\ell$, $\ell=1,\dots,s$, of the following two types
\begin{enumerate}[(i)]
\item
$\alpha_1=\dots=\alpha_s\in\mb F$ and $\beta_\ell=\gamma_\ell=0$ for all $\ell$, 
with $s\in\mb Z_{\geq2}$;
\item
$s=2$ and
$\alpha_1=\alpha_2$, $\beta_1=\beta_2$ and $\gamma_1=-\gamma_2$.
\end{enumerate}
Assume also that $A_1(\partial),\dots,A_s(\partial)$
are invertible in $\mc V((\partial^{-1}))\otimes\End V$.
Let 
\begin{equation}\label{eq:prod}
L(\partial)=A_1(\partial)\dots A_s(\partial)
\,.
\end{equation}
For $B(\partial)\in\mc V((\partial^{-1}))\otimes\End V$
a $K$-th root of $L(\partial)$ (i.e. $L(\partial)=B(\partial)^K$ for $K\in\mb Z\backslash\{0\}$),
define the elements $h_{n,B}\in\mc V$, $n\in\mb Z$, by \eqref{eq:hn}.
Then: 
\begin{enumerate}[(a)]
\item
All the elements $\tint h_{n,B}$ are Hamiltonian functionals in involution:
\begin{equation}\label{eq:invol-prod}
\{\tint h_{m,B},\tint h_{n,C}\}=0
\,\text{ for all } m,n\in\mb Z,\,
B,C \text{ roots of } L
\,.
\end{equation}
\item
The corresponding compatible hierarchy of Hamiltonian equations satisfies
\begin{equation}\label{eq:hierarchy-prod}
\frac{dL(w)}{dt_{n,B}}
=
\{\tint h_{n,B},L(w)\}
=
\big[\alpha_1(B^n)_+-\beta_1(A_2B^{n-K}A_1)_+^{\star,\dagger}
+\gamma_1f_1^n,L\big](\partial)
\,,
\end{equation}
where $n\in\mb Z$, $B$ is a root of $A$, and
$$
f_1^n:=\partial^{-1}\res_z\tr\big((A_2B^{n-K}A_1)(z)-B^n(z)\big)\,\in\mc V/\ker\partial
\,.
$$
The Hamiltonian functionals $\tint h_{n,C}$, $n\in\mb Z_+$, $C$ root of $L$,
are integrals of motion of all the equations \eqref{eq:hierarchy-prod}.
\end{enumerate}
\end{theorem}
\begin{proof}
Suppose $B$ is a $K$-th root of $A$, $K\in\mb Z\backslash\{0\}$ 
and $C$ is an $H$-th root of $A$, $H\in\mb Z\backslash\{0\}$.
Applying the second equation in \eqref{eq:hn2} first,
and then the first equation in \eqref{eq:hn2}, we get
\begin{equation}\label{eq:hn-pr1-prod}
\begin{split}
& 
\{\tint h_{m,B},\tint h_{n,C}\}
= 
\int \Res_z \Res_w (\tr\otimes\tr)
\{L(z+x)_x L(w+y)\}
\\
& \,\,\,\,\,\,\,\,\,\,\,\,\,\,\,\,\,\, \times
\Big(
\big(\big|_{x=\partial}B^{m-K}(z)\big)\otimes
\big(\big|_{y=\partial}C^{n-H}(w)\big)
\Big)\,.
\end{split}
\end{equation}
By Lemma \ref{20170404:lem2}(c)-(d), we have
\begin{equation}\label{20170411:eq2}
\begin{array}{l}
\displaystyle{
\vphantom{\Big(}
\{L(z)_\lambda L(w)\}
} \\
\displaystyle{
\vphantom{\Big(}
=
\sum_{\ell,r=1}^{s}
\big(\big|_{x_1=\partial}(A_1\dots A_{\ell-1})^*(\lambda-z)\big)
\otimes 
(A_1\dots A_{r-1})(w+\lambda+x_1+x_2+y_1+v)
} \\
\displaystyle{
\vphantom{\Big(}
\times
\big(|_{v=\partial}
\{A_\ell(z+x_2)_{\lambda+x_1+x_2}A_r(w+y_1)\}
\big)
} \\
\displaystyle{
\vphantom{\Big(}
\times
\big(\big|_{x_2=\partial}(A_{\ell+1}\dots A_s)(z)\big)
\otimes
\big(\big|_{y_1=\partial}(A_{r+1}\dots A_s)(w)\big)
\,.}
\end{array}
\end{equation}
Using \eqref{20170411:eq2} and Lemma \ref{20170404:lem1}(a),
we can rewrite the RHS of \eqref{eq:hn-pr1-prod} as
\begin{equation}\label{20170411:eq3}
\begin{array}{l}
\displaystyle{
\vphantom{\Big(}
\sum_{\ell,r=1}^{s}
\int \Res_z \Res_w (\tr\otimes\tr)
} \\
\displaystyle{
\vphantom{\Big(}
\big(\big|_{x_1=\partial}(A_1\dots A_{\ell-1})^*(-z)\big)
\otimes 
(A_1\dots A_{r-1})(w+x_1+x_2+y+v)
} \\
\displaystyle{
\vphantom{\Big(}
\times
\big(|_{v=\partial}
\{A_\ell(z+x_2)_{x_1+x_2}A_r(w+y)\}
\big)
} \\
\displaystyle{
\vphantom{\Big(}
\times
\big(\big|_{x_2=\partial}(A_{\ell+1}\dots A_sB^{m-K})(z)\big)
\otimes
\big(\big|_{y=\partial}(A_{r+1}\dots A_sC^{n-H})(w)\big)
\,.}
\end{array}
\end{equation}
We next use the cyclic property of the trace ($\tr\otimes1$), Lemma \ref{lem:hn1}(a)
and then Lemma \ref{20170404:lem1}(b), to rewrite \eqref{20170411:eq3} as
\begin{equation}\label{20170411:eq4}
\begin{array}{l}
\displaystyle{
\vphantom{\Big(}
\sum_{\ell,r=1}^{s}
\int \Res_z \Res_w \tr
\Big(
(A_1\dots A_{r-1})(w+x+y+v)
} \\
\displaystyle{
\vphantom{\Big(}
\times
(1\otimes\tr)
\big(|_{v=\partial}
\{A_\ell(z+x)_{x}A_r(w+y)\}
\big)
} \\
\displaystyle{
\vphantom{\Big(}
\times
\big(\big|_{x=\partial}(A_{\ell+1}\dots A_sB^{m-K}A_1\dots A_{\ell-1})(z)\big)
\otimes
\big(\big|_{y=\partial}(A_{r+1}\dots A_sC^{n-H})(w)\big)
\Big)
\,.}
\end{array}
\end{equation}
We then use Lemma \ref{lem:hn1}(b)
and then Lemma \ref{20170404:lem1}(a), to deduce, from \eqref{20170411:eq4}
\begin{equation}\label{20170411:eq5}
\begin{array}{l}
\displaystyle{
\vphantom{\Big(}
\{\tint h_{m,B},\tint h_{n,C}\}
=
\sum_{\ell,r=1}^{s}
\int \Res_z \Res_w (\tr\otimes\tr)
\{A_\ell(z+x)_{x}A_r(w+y)\}
} \\
\displaystyle{
\vphantom{\Big(}
\times
\big(\big|_{x=\partial}(A_{\ell+1}\dots A_sB^{m-K}A_1\dots A_{\ell-1})(z)\big)
} \\
\displaystyle{
\vphantom{\Big(}
\otimes
\big(\big|_{y=\partial}(A_{r+1}\dots A_sC^{n-H}A_1\dots A_{r-1})(w)\big)
\,.}
\end{array}
\end{equation}
By \eqref{20170411:eq1},
only the summands with $\ell=r$ survive in \eqref{20170411:eq5}.
We then use the generalized Adler identity \eqref{eq:adler-general} for $A_\ell(\partial)$
(with $\alpha$ and $\beta$ independent of $\ell$, and with $\gamma=0$), 
to get
\begin{equation}\label{20170412:eq1}
\begin{array}{l}
\displaystyle{
\vphantom{\Big(}
\{\tint h_{m,B},\tint h_{n,C}\}
=
\sum_{\ell=1}^{s}
\int \Res_z \Res_w (\tr\otimes\tr)
} \\
\displaystyle{
\vphantom{\Big(}
\Bigg(
\alpha_\ell
\iota_z(z-w-x_1-y)^{-1}
\Omega\,\,
A_\ell(w+x+x_1+y)
\otimes
\big(\big|_{x_1=\partial}A_\ell^*(-z)\big)
} \\
\displaystyle{
\vphantom{\Big(}
\times
\big(\big|_{x=\partial}(A_{\ell+1}\!\dots\! A_sB^{m-K}\!\!A_1\!\dots\! A_{\ell-1})(z)\big)
\!\otimes\!
\big(\big|_{y=\partial}(A_{\ell+1}\!\dots\! A_sC^{n-H}\!\!A_1\!\dots\! A_{\ell-1})(w)\big)
} \\
\displaystyle{
\vphantom{\Big(}
-
\alpha_\ell
\iota_z(z-w-y-y_1)^{-1}
\Omega
A_\ell(z+x)\otimes\big(\big|_{y_1=\partial}A_\ell(w+y)\big)
} \\
\displaystyle{
\vphantom{\Big(}
\times
\big(\big|_{x=\partial}(A_{\ell+1}\!\dots\! A_sB^{m-K}\!\!A_1\!\dots\! A_{\ell-1})(z)\big)
\!\otimes\!
\big(\big|_{y=\partial}(A_{\ell+1}\!\dots\! A_sC^{n-H}\!\!A_1\!\dots\! A_{\ell-1})(w)\big)
} \\
\displaystyle{
\vphantom{\Big(}
-
\beta_\ell
\iota_z(z+w+x+x_1+y)^{-1}
(\id\otimes A_\ell(w+x+x_1+y))
\Omega^\dagger
\big(\big|_{x_1=\partial}A_\ell(z+x)\otimes\id\big)
} \\
\displaystyle{
\vphantom{\Big(}
\times
\big(\big|_{x=\partial}(A_{\ell+1}\!\dots\! A_sB^{m-K}\!\!A_1\!\dots\! A_{\ell-1})(z)\big)
\!\otimes\!
\big(\big|_{y=\partial}(A_{\ell+1}\!\dots\! A_sC^{n-H}\!\!A_1\!\dots\! A_{\ell-1})(w)\big)
} \\
\displaystyle{
\vphantom{\Big(}
+
\beta_\ell
\iota_z(z+w+x+y+y_1)^{-1}
(A_\ell^*(-z)\otimes\id)
\Omega^\dagger
\big(\id\otimes \big|_{y_1=\partial}A_\ell(w+y)\big)
} \\
\displaystyle{
\vphantom{\Big(}
\times\!
\big(\big|_{x=\partial}(A_{\ell+1}\!\dots\! A_sB^{m-K}\!\!A_1\!\dots\! A_{\ell-1})(z)\big)
\!\otimes\!
\big(\big|_{y=\partial}(A_{\ell+1}\!\dots\! A_sC^{n-H}\!\!A_1\!\dots\! A_{\ell-1})(w)\big)
} \\
\displaystyle{
\vphantom{\Big(}
+
\gamma_\ell
(x+x_1)^{-1}
\big(\big|_{x_1=\partial}\big(A_\ell^*(-z)-A_\ell(z+x)\big)\big)
\otimes 
\big(A_\ell(w\!+\!x\!+\!x_1\!+\!y)-A_\ell(w+y)\big)
} \\
\displaystyle{
\vphantom{\Big(}
\times
\big(\big|_{x=\partial}(A_{\ell+1}\!\!\dots\!\! A_sB^{m-K}\!\!A_1\!\!\dots\!\! A_{\ell-1})(z)\big)
\otimes
\big(\big|_{y=\partial}(A_{\ell+1}\!\!\dots\!\! A_sC^{n-H}\!\!A_1\!\!\dots\!\! A_{\ell-1})(w)\big)
\!\Bigg)
.}
\end{array}
\end{equation}
The first term in the RHS of \eqref{20170412:eq1} can be rewritten as follows:
\begin{equation}\label{20170412:eq2}
\begin{array}{l}
\displaystyle{
\vphantom{\Big(}
\sum_{\ell=1}^{s}
\alpha_\ell
\int \Res_z \Res_w (\tr\otimes\tr)
\iota_z(z-w-x_1-y)^{-1}
\Omega
} \\
\displaystyle{
\vphantom{\Big(}
\quad\times
A_\ell(w+x+x_1+y)
\otimes
\big(\big|_{x_1=\partial}A_\ell^*(-z)\big)
} \\
\displaystyle{
\vphantom{\Big(}
\quad\times\!
\big(\big|_{x=\partial}\!(A_{\ell+1}\!\dots\! A_sB^{m-K}\!\!A_1\!\dots\! A_{\ell-1})(\!z\!)\big)
\!\otimes\!
\big(\big|_{y=\partial}\!(A_{\ell+1}\!\dots\! A_sC^{n-H}\!\!A_1\!\dots\! A_{\ell-1})(\!w\!)\big)
} \\
\displaystyle{
\vphantom{\Big(}
=
\sum_{\ell=1}^{s}
\alpha_\ell
\int \Res_z \Res_w \tr
\iota_z(z-w-x_1-y)^{-1}
A_\ell(w+x+x_1+y)
} \\
\displaystyle{
\vphantom{\Big(}
\qquad\times
\big(\big|_{x=\partial}(A_{\ell+1}\dots A_sB^{m-K}A_1\dots A_{\ell-1})(z)\big)
\big(\big|_{x_1=\partial}A_\ell^*(-z)\big)
} \\
\displaystyle{
\vphantom{\Big(}
\qquad\times
\big(\big|_{y=\partial}(A_{\ell+1}\dots A_sC^{n-H}A_1\dots A_{\ell-1})(w)\big)
} \\
\displaystyle{
\vphantom{\Big(}
=
\sum_{\ell=1}^{s}
\alpha_\ell
\int \Res_w \tr
A_\ell(w+x+x_1+y)
} \\
\displaystyle{
\vphantom{\Big(}
\qquad\times
\Big(
\big(\big|_{x=\partial}(A_{\ell+1}\dots A_sB^{m-K}A_1\dots A_{\ell-1})(w+x_1+y)\big)
} \\
\displaystyle{
\vphantom{\Big(}
\qquad\times
\big(\big|_{x_1=\partial}A_\ell(w+y)\big)
\Big)_+
\big(\big|_{y=\partial}(A_{\ell+1}\dots A_sC^{n-H}A_1\dots A_{\ell-1})(w)\big)
} \\
\displaystyle{
\vphantom{\Big(}
=
\sum_{\ell=1}^{s}
\alpha_\ell
\int \Res_w \tr
A_\ell(w+\partial)
(A_{\ell+1}\dots A_sB^{m-K}A_1\dots A_{\ell-1}A_\ell)(w+\partial)_+
} \\
\displaystyle{
\vphantom{\Big(}
\qquad\times
(A_{\ell+1}\dots A_sC^{n-H}A_1\dots A_{\ell-1})(w)
} \\
\displaystyle{
\vphantom{\Big(}
=
\sum_{\ell=1}^{s}
\alpha_\ell
\int \Res_w \tr
(A_{\ell+1}\dots A_sB^{m-K}A_1\dots A_{\ell-1}A_\ell)(w+\partial)_+
} \\
\displaystyle{
\vphantom{\Big(}
\qquad\times
(A_{\ell+1}\dots A_sC^{n-H}A_1\dots A_{\ell-1}A_\ell)(w)
\,,}
\end{array}
\end{equation}
where we used Lemma \ref{lemma:29032017}(b) for the first equality,
equations \eqref{20170406:eq1} and \eqref{eq:positive} for the second equality,
Lemma \ref{20170404:lem1}(a) for the third equality,
and Lemma \ref{lem:hn1}(b) for the fourth equality.
The second term in the RHS of \eqref{20170412:eq1} is
\begin{equation}\label{20170412:eq3}
\begin{array}{l}
\displaystyle{
\vphantom{\Big(}
-
\sum_{\ell=1}^{s}
\alpha_\ell
\!\int\! \Res_z \Res_w (\tr\otimes\tr)
\iota_z(z\!-\!w\!-\!y\!-\!y_1)^{-1}
\Omega
A_\ell(z+x)\otimes\big(\big|_{y_1=\partial}A_\ell(w+y)\big)
} \\
\displaystyle{
\vphantom{\Big(}
\quad\times\!
\big(\big|_{x=\partial}\!(A_{\ell+1}\!\dots\! A_sB^{m-K}\!\!A_1\!\dots\! A_{\ell\!-\!1})(\!z\!)\big)
\!\otimes\!
\big(\big|_{y=\partial}\!(A_{\ell+1}\!\dots\! A_sC^{n-H}\!\!A_1\!\dots\! A_{\ell\!-\!1})(\!w\!)\big)
} \\
\displaystyle{
\vphantom{\Big(}
=-
\sum_{\ell=1}^{s}
\alpha_\ell
\int \Res_z \Res_w \tr
\iota_z(z-w-y-y_1)^{-1}
} \\
\displaystyle{
\vphantom{\Big(}
\qquad\times
A_\ell(z+x)
\big(\big|_{x=\partial}(A_{\ell+1}\dots A_sB^{m-K}A_1\dots A_{\ell-1})(z)\big)
} \\
\displaystyle{
\vphantom{\Big(}
\qquad\times
\big(\big|_{y_1=\partial}A_\ell(w+y)\big)
\big(\big|_{y=\partial}(A_{\ell+1}\dots A_sC^{n-H}A_1\dots A_{\ell-1})(w)\big)
} \\
\displaystyle{
\vphantom{\Big(}
=-
\sum_{\ell=1}^{s}
\alpha_\ell
\int \Res_z \Res_w \tr
(A_\ell A_{\ell+1}\dots A_sB^{m-K}A_1\dots A_{\ell-1})(w+\partial)_+
} \\
\displaystyle{
\vphantom{\Big(}
\qquad\times
(A_\ell A_{\ell+1}\dots A_sC^{n-H}A_1\dots A_{\ell-1})(w)
\,,}
\end{array}
\end{equation}
where we used Lemma \ref{lemma:29032017}(b) for the first equality,
Lemma \ref{20170404:lem1}(a) for the second equality,
and equation \eqref{eq:positive} for the last equality.
If $\alpha_1=\dots=\alpha_\ell$,
combining \eqref{20170412:eq2} and \eqref{20170412:eq3},
we get a telescopic sum where only the $\ell=s$ term in \eqref{20170412:eq2}
and the $\ell=1$ term in \eqref{20170412:eq3} survive.
As a result we obtain
\begin{equation}\label{20170412:eq4}
\begin{array}{l}
\displaystyle{
\vphantom{\Big(}
\alpha_s
\int \Res_w \tr
(B^{m-K}A_1\dots A_s)(w+\partial)_+
(C^{n-H}A_1\dots A_s)(w)\big)
} \\
\displaystyle{
\vphantom{\Big(}
-\alpha_1
\int \Res_z \Res_w \tr
(A_1\dots A_sB^{m-K})(w+\partial)_+
(A_1\dots A_sC^{n-H})(w)
} \\
\displaystyle{
\vphantom{\Big(}
=
\alpha_1
\int \Res_w \tr
\Big(
B^{m}(w+\partial)_+
(C^{n}(w)\big)
-
B^{m}(w+\partial)_+
C^{n}(w)
\Big)
=0
\,,}
\end{array}
\end{equation}
where we used the identities $A_1\dots A_s=L=B^K=C^H$.
This proves claim (a) in case (i), i.e. when $\alpha_1=\dots=\alpha_s$ and $\beta_\ell=\gamma_\ell=0$
for all $\ell$.
To prove claim (a) in case (ii),
we let $s=2$ and we consider the last three terms in the RHS of \eqref{20170412:eq1}.
For the third term, we have
\begin{equation}\label{20170412:eq5}
\begin{array}{l}
\displaystyle{
\vphantom{\Big(}
-\beta_1
\int \Res_z \Res_w (\tr\otimes\tr)
\iota_z(z+w+x+y)^{-1}
(\id\otimes A_1(w+x+y))
\Omega^\dagger
} \\
\displaystyle{
\vphantom{\Big(}
\qquad\times
\big(\big|_{x=\partial}(A_1A_2B^{m-K})(z)\big)
\otimes
\big(\big|_{y=\partial}(A_2C^{n-H})(w)\big)
} \\
\displaystyle{
\vphantom{\Big(}
-\beta_2
\int \Res_z \Res_w (\tr\otimes\tr)
\iota_z(z+w+x+y)^{-1}
(\id\otimes A_2(w+x+y))
\Omega^\dagger
} \\
\displaystyle{
\vphantom{\Big(}
\qquad\times
\big(\big|_{x=\partial}(A_2B^{m-K}A_1)(z)\big)
\otimes
\big(\big|_{y=\partial}(C^{n-H}A_1)(w)\big)
} \\
\displaystyle{
\vphantom{\Big(}
=
-\beta_1
\int \Res_z \Res_w \tr
\iota_z(z\!+\!w\!+\!x\!+\!y)^{-1}
A_1(w+x+y)
\big(\big|_{x=\partial}(A_1A_2B^{m-K})(z)\big)^\dagger
} \\
\displaystyle{
\vphantom{\Big(}
\quad\times\big(\big|_{y=\partial}(A_2C^{n-H})(w)\big)
-\beta_2
\int \Res_z \Res_w \tr
\iota_z(z\!+\!w\!+\!x\!+\!y)^{-1}
A_2(w+x+y)
} \\
\displaystyle{
\vphantom{\Big(}
\quad\times
\big(\big|_{x=\partial}(A_2B^{m-K}A_1)(z)\big)^\dagger
\big(\big|_{y=\partial}(C^{n-H}A_1)(w)\big)
} \\
\displaystyle{
\vphantom{\Big(}
=
-\beta_1\!
\int\! \Res_w\! \tr
A_1(w\!+\!x\!+\!y)
\big(\big|_{x=\partial}\!(A_1A_2B^{m-K})^*(w\!+\!y)_+\big)^\dagger
\big(\big|_{y=\partial}\!(A_2C^{n-H})(w)\big)
} \\
\displaystyle{
\vphantom{\Big(}
-\beta_2
\int \Res_w \tr
A_2(w\!+\!x\!+\!y)
\big(\big|_{x=\partial}(A_2B^{m-K}A_1)^*(w\!+\!y)_+\big)^\dagger
\big(\big|_{y=\partial}(C^{n-H}A_1)(w)\big)
} \\
\displaystyle{
\vphantom{\Big(}
=
-\beta_1
\int \Res_w \tr
A_1(w+\partial)
\big((A_1A_2B^{m-K})^*(w+\partial)_+\big)^\dagger
(A_2C^{n-H})(w)
} \\
\displaystyle{
\vphantom{\Big(}
-\beta_2
\int \Res_w \tr
A_2(w+\partial)
\big((A_2B^{m-K}A_1)^*(w+\partial)_+\big)^\dagger
(C^{n-H}A_1)(w)
} 
\end{array}
\end{equation}

\begin{equation}\label{eq:hn-pr2}
\begin{array}{l}
\displaystyle{
\vphantom{\Big(}
=
-\beta_1
\int \Res_w \tr
\big((A_1A_2B^{m-K})^*(w+\partial)_+\big)^\dagger
(A_2C^{n-H}A_1)(w)
} \\
\displaystyle{
\vphantom{\Big(}
-\beta_2
\int \Res_w \tr
\big((A_2B^{m-K}A_1)^*(w+\partial)_+\big)^\dagger
(C^{n-H}A_1A_2)(w)
} \\
\displaystyle{
\vphantom{\Big(}
=
-\beta_1
\int \Res_w \tr
\big((B^{m})^*(w+\partial)_+\big)^\dagger
(A_2C^{n-H}A_1)(w)
} \\
\displaystyle{
\vphantom{\Big(}
-\beta_2
\int \Res_w \tr
\big((A_2B^{m-K}A_1)^*(w+\partial)_+\big)^\dagger
C^{n}(w)
\,.}
\end{array}
\end{equation}
Here we used first Lemma \ref{20170404:lem1}, 
then we used Lemma \ref{lemma:29032017}(d) for the first equality,
equations \eqref{20170406:eq1} and \eqref{eq:positive} for the second equality,
Lemma \ref{lem:hn1}(b) for the fourth equality,
and the identities $A_1A_2=B^K=C^H$ for the last equality.
Similarly, for the fourth term in the RHS of \eqref{20170412:eq1}, we have
\begin{equation}\label{20170412:eq6}
\begin{array}{l}
\displaystyle{
\vphantom{\Big(}
\beta_1
\int \Res_z \Res_w (\tr\otimes\tr)
\iota_z(z+w+x+y)^{-1}
(A_1^*(-z)\otimes\id)
\Omega^\dagger
} \\
\displaystyle{
\vphantom{\Big(}
\qquad\times
\big(\big|_{x=\partial}(A_2B^{m-K})(z)\big)
\otimes
\big(\big|_{y=\partial}(A_1A_2C^{n-H})(w)\big)
} \\
\displaystyle{
\vphantom{\Big(}
+\beta_2
\int \Res_z \Res_w (\tr\otimes\tr)
\iota_z(z+w+x+y)^{-1}
(A_2^*(-z)\otimes\id)
\Omega^\dagger
} \\
\displaystyle{
\vphantom{\Big(}
\qquad\times
\big(\big|_{x=\partial}(B^{m-K}A_1)(z)\big)
\otimes
\big(\big|_{y=\partial}(A_2C^{n-H}A_1)(w)\big)
} \\
\displaystyle{
\vphantom{\Big(}
=
\beta_1
\int \Res_z \Res_w \tr
\iota_z(z+w+x+y)^{-1}
\Big(\big(\big|_{x=\partial}(A_2B^{m-K})(z)\big)
A_1^*(-z)\Big)^\dagger
} \\
\displaystyle{
\vphantom{\Big(}
\quad\times
\big(\big|_{y=\partial}(A_1A_2C^{n-H})(w)\big)
+\beta_2
\int \Res_z \Res_w \tr
\iota_z(z+w+x+y)^{-1}
} \\
\displaystyle{
\vphantom{\Big(}
\quad\times
\Big(\big(\big|_{x=\partial}(B^{m-K}A_1)(z)\big)
A_2^*(-z)\Big)^\dagger
\big(\big|_{y=\partial}(A_2C^{n-H}A_1)(w)\big)
} \\
\displaystyle{
\vphantom{\Big(}
=
\!\beta_1
\!\int\! \Res_w \!\tr\!
\Big(\big(\big|_{x=\partial}(A_2B^{m-K})^*(w\!+\!y)\big)
A_1^*(w\!+\!x\!+\!y)\Big)_+^\dagger
\big(\big|_{y=\partial}(A_1A_2C^{n-H})(w)\big)
} \\
\displaystyle{
\vphantom{\Big(}
+\beta_2
\!\int\! \Res_w \!\tr\!
\Big(\big(\big|_{x=\partial}(B^{m-K}A_1)^*(w\!+\!y)\big)
A_2^*(w\!+\!x\!+\!y)\Big)_+^\dagger
\big(\big|_{y=\partial}(A_2C^{n-H}A_1)(w)\big)
} \\
\displaystyle{
\vphantom{\Big(}
=
\beta_1
\int \Res_w \tr
\big((A_2B^{m-K}A_1)^*(w+\partial)_+\big)^\dagger
(A_1A_2C^{n-H})(w)
} \\
\displaystyle{
\vphantom{\Big(}
+\beta_2
\int \Res_w \tr
\big((B^{m-K}A_1A_2)^*(w+\partial)_+\big)^\dagger
(A_2C^{n-H}A_1)(w)
} \\
\displaystyle{
\vphantom{\Big(}
=
\beta_1
\int \Res_w \tr
\big((A_2B^{m-K}A_1)^*(w+\partial)_+\big)^\dagger
C^{n}(w)
} \\
\displaystyle{
\vphantom{\Big(}
+\beta_2
\int \Res_w \tr
\big((B^{m})^*(w+\partial)_+\big)^\dagger
(A_2C^{n-H}A_1)(w)
\,,}
\end{array}
\end{equation}
where we used first Lemma \ref{20170404:lem1}, 
then we used Lemma \ref{lemma:29032017}(d) for the first equality,
equations \eqref{20170406:eq1} and \eqref{eq:positive} for the second equality,
Lemma \ref{20170404:lem1}(b) for the third equality,
and the identities $A_1A_2=B^K=C^H$ for the last equality.
If $\beta_1=\beta_2$, as assumed in case (ii),
we can combine \eqref{20170412:eq5} and \eqref{20170412:eq6} to get zero.
Finally, let us consider the last term in the RHS of \eqref{20170412:eq1}.
It is equal to
\begin{equation}\label{20170412:eq7}
\begin{array}{l}
\displaystyle{
\vphantom{\Big(}
\gamma_1
\int \Res_z \Res_w
x^{-1}
\tr\Big(\Big|_{x=\partial}
\big(A_1^*(-z)-A_1(z+\partial)\big)
(A_2B^{m-K})(z)
\Big)
} \\
\displaystyle{
\vphantom{\Big(}
\qquad\times
\tr\Big(
\big(A_1(w+x+\partial)-A_1(w+\partial)\big)
(A_2C^{n-H})(w)
\Big)
} \\
\displaystyle{
\vphantom{\Big(}
+
\gamma_2
\int \Res_z \Res_w 
x^{-1}
\tr\Big(\Big|_{x=\partial}
\big(A_2^*(-z)-A_2(z+\partial)\big)
(B^{m-K}A_1)(z)
\Big)
} \\
\displaystyle{
\vphantom{\Big(}
\qquad\times
\tr\Big(
\big(A_2(w+x+\partial)-A_2(w+\partial)\big)
(C^{n-H}A_1)(w)
\Big)
\,.}
\end{array}
\end{equation}
By Lemma \ref{lem:hn1}(a) and (b), the cyclic property of the trace, and the identity $A_1A_2=B^K$,
we have
$$
\begin{array}{l}
\displaystyle{
\vphantom{\Big(}
\Res_z\tr\Big(
\big(A_1^*(-z)-A_1(z+\partial)\big)
(A_2B^{m-K})(z)
\Big)
} \\
\displaystyle{
\vphantom{\Big(}
=
\Res_z\tr\Big(
(A_2B^{m-K}A_1-A_1A_2B^{m-K})(z)
\Big)
\,.}
\end{array}
$$
Moreover, integrating by parts, we can also replace in \eqref{20170412:eq7}
$$
\begin{array}{l}
\displaystyle{
\vphantom{\Big(}
\Res_w\tr\Big(
\big(A_1(w+x+\partial)-A_1(w+\partial)\big)
(A_2C^{n-H})(w)\Big)
} \\
\displaystyle{
\vphantom{\Big(}
=
\Res_w\tr\Big(
\big(A_1^*(-w)-A_1(w+\partial)\big)
(A_2C^{n-H})(w)\Big)
} \\
\displaystyle{
\vphantom{\Big(}
=
\Res_w\tr\Big(
(A_2C^{n-H}A_1-C^n)(w)
\Big)
\,.}
\end{array}
$$
Hence, \eqref{20170412:eq7} becomes
\begin{equation}\label{20170412:eq8}
\begin{array}{l}
\displaystyle{
\vphantom{\Big(}
\gamma_1
\int
\Res_w
\tr(A_2C^{n-H}A_1-C^n)(w)
\partial^{-1}
\Res_z 
\tr
(A_2B^{m-K}A_1-B^m)(z)
} \\
\displaystyle{
\vphantom{\Big(}
+\!
\gamma_2\!
\int 
\Res_w 
\tr(C^{n}\!-\!A_2C^{n-H}\!A_1)(w)
\partial^{-1}
\Res_z 
\tr(B^m\!-\!A_2B^{m-K}\!A_1)(z)
\,,}
\end{array}
\end{equation}
which is zero since, by assumption, $\gamma_1=-\gamma_2$.
This completes the proof of (a).

Next, let us prove part (b).
By the first equation in \eqref{eq:hn2} and Lemma \ref{20170404:lem2}(d), 
we have
\begin{equation}\label{eq:prb1}
\begin{array}{l}
\displaystyle{
\vphantom{\Big(}
\frac{dA_\ell(w)}{dt_{n,B}}
=
\{\tint h_{n,B},A_\ell(w)\}
=
\{{h_{n,B}}_\lambda A_\ell(w)\}\big|_{\lambda=0}
} \\
\displaystyle{
\vphantom{\Big(}
=
-\Res_z (\tr\otimes1)
\{L(z+x)_x A_\ell(w)\}
\big(\big|_{x=\partial}B^{n-K}(z)\otimes\id\big)
} \\
\displaystyle{
\vphantom{\Big(}
=
-
\sum_{r=1}^{s}
\Res_z (\tr\otimes1)
\big(\big|_{x_2=\partial}(A_1\dots A_{r-1})^*(-z)\otimes1\big)
\{A_r(z\!+\!x\!+\!x_1)_{x\!+\!x_1\!+\!x_2}A_\ell(w)\}
} \\
\displaystyle{
\vphantom{\Big(}
\quad\times
\big(\big|_{x_1=\partial}(A_{r+1}\dots A_s)(z+x)\otimes1\big)
\big(\big|_{x=\partial}B^{n-K}(z)\otimes\id\big)
\,.}
\end{array}
\end{equation}
We can use the cyclic property of the trace, Lemma \ref{lem:hn1}(b) and Lemma \ref{20170404:lem1}(a),
to rewrite the RHS of \eqref{eq:prb2} as
\begin{equation}\label{eq:prb2}
\begin{array}{l}
\displaystyle{
\vphantom{\Big(}
-
\sum_{r=1}^{s}
\Res_z (\tr\otimes1)
\{A_r(z\!+\!x)_{x}A_\ell(w)\}
\big(\big|_{x=\partial}(A_{r\!+\!1}\dots A_sB^{n-K}A_1\dots A_{r\!-\!1})(z)\otimes1\big)
\,.}
\end{array}
\end{equation}
By assumption \eqref{20170411:eq1}, only the $r=\ell$ term in \eqref{eq:prb2} is non zero,
and we can use the generalized Adler identity \eqref{eq:adler-general}
to rewrite \eqref{eq:prb2} as
\begin{equation}\label{eq:prb3}
\begin{array}{l}
\displaystyle{
\vphantom{\Big(}
-
\alpha_\ell
\Res_z 
\iota_z(z-w-x_1)^{-1}
(\tr\otimes1)
\Omega
} \\
\displaystyle{
\vphantom{\Big(}
\quad\times
A_\ell(w+x+x_1)
\big(\big|_{x=\partial}(A_{\ell+1}\dots A_sB^{n-K}A_1\dots A_{\ell-1})(z)\big)
\otimes
\big(\big|_{x_1=\partial}A_\ell^*(-z)\big)
} \\
\displaystyle{
\vphantom{\Big(}
+
\alpha_\ell
\Res_z 
\iota_z(z-w-y)^{-1}
(\tr\otimes1)
\Omega
} \\
\displaystyle{
\vphantom{\Big(}
\quad\times
(A_\ell A_{\ell+1}\dots A_sB^{n-K}A_1\dots A_{\ell-1})(z)
\otimes
\big(\big|_{y=\partial}A_\ell(w)\big)
} \\
\displaystyle{
\vphantom{\Big(}
+
\beta_\ell
\Res_z 
\iota_z(z+w+x)^{-1}
(\tr\otimes1)
\big(\id\otimes A_\ell(w+x)\big)
\Omega^\dagger
} \\
\displaystyle{
\vphantom{\Big(}
\quad\times
\big(\big|_{x=\partial}(A_\ell A_{\ell+1}\dots A_sB^{n-K}A_1\dots A_{\ell-1})(z)\otimes1\big)
} \\
\displaystyle{
\vphantom{\Big(}
-
\beta_\ell
\Res_z 
\iota_z(z+w+x+y)^{-1}
(\tr\otimes1)
\big(A_\ell^*(-z)\otimes\id\big)
\Omega^\dagger
} \\
\displaystyle{
\vphantom{\Big(}
\quad\times
\big(
\big|_{x=\partial}(A_{\ell+1}\dots A_sB^{n-K}A_1\dots A_{\ell-1})(z)
\otimes 
\big|_{y=\partial}A_\ell(w)
\big)
} 
\end{array}
\end{equation}

\begin{equation}\label{eq:hn-pr2}
\begin{array}{l}
\displaystyle{
\vphantom{\Big(}
-
\gamma_\ell
\big(A_\ell(w+\partial)-A_\ell(w)\big)\partial^{-1}
\Res_z
\tr\Big(
\big(A_\ell^*(-z)-A_\ell(z\!+\!x)\big)
} \\
\displaystyle{
\vphantom{\Big(}
\quad\times
\big(\big|_{x=\partial}(A_{\ell+1}\!\dots\! A_sB^{n-K}A_1\!\dots\! A_{\ell-1})(z)\big)
\Big)\,.}
\end{array}
\end{equation}
Next, we use Lemmas \ref{lemma:29032017}(a)-(c), \ref{lem:hn1}(b) and \ref{20170404:lem1},
and equations \eqref{20170406:eq1} and \eqref{eq:positive},
to rewrite \eqref{eq:prb3} as
\begin{equation}\label{eq:prb4}
\begin{array}{l}
\displaystyle{
\vphantom{\Big(}
-
\alpha_\ell
A_\ell(w+\partial)
(A_{\ell+1}\dots A_sB^{n-K}A_1\dots A_{\ell-1}A_\ell)(w)_+
} \\
\displaystyle{
\vphantom{\Big(}
+
\alpha_\ell
(A_\ell A_{\ell+1}\dots A_sB^{n-K}A_1\dots A_{\ell-1})(w+\partial)_+A_\ell(w)
} \\
\displaystyle{
\vphantom{\Big(}
+
\beta_\ell
A_\ell(w+\partial)
\big((A_\ell A_{\ell+1}\dots A_sB^{n-K}A_1\dots A_{\ell-1})^*(w)_+\big)^\dagger
} \\
\displaystyle{
\vphantom{\Big(}
-
\beta_\ell
\big((A_{\ell+1}\dots A_sB^{n-K}A_1\dots A_{\ell-1}A_\ell)^*(w+\partial)_+\big)^\dagger
A_\ell(w)
} \\
\displaystyle{
\vphantom{\Big(}
-
\gamma_\ell
\big(A_\ell(w+\partial)-A_\ell(w)\big)\partial^{-1}
\Res_z
\tr\Big(
(A_{\ell+1}\!\dots\! A_sB^{n-K}A_1\!\dots\! A_{\ell-1}A_\ell)(z)
} \\
\displaystyle{
\vphantom{\Big(}
\qquad
-(A_\ell A_{\ell+1}\!\dots\! A_sB^{n-K}A_1\!\dots\! A_{\ell-1})(z)
\Big)\,.}
\end{array}
\end{equation}
Rewriting \eqref{eq:prb4} in operator form, we get the following evolution equation for 
$A_\ell(\partial)\in\mc V((\partial^{-1}))\otimes\End V$
\begin{equation}\label{eq:prb5}
\begin{array}{l}
\displaystyle{
\vphantom{\Big(}
\frac{dA_\ell(\partial)}{dt_{n,B}}
=
-\alpha_\ell
A_\ell(\partial) B_{\ell}^n(\partial)_+
+\alpha_\ell
B_{\ell-1}^n(\partial)_+ A_\ell(\partial)
} \\
\displaystyle{
\vphantom{\Big(}
+\beta_\ell
A_\ell(\partial) ((B_{\ell-1}^n)^*(\partial)_+)^{\dagger}
-\beta_\ell
((B_{\ell}^n)^*(\partial)_+)^{\dagger} A_\ell(\partial)
-\gamma_\ell [A_\ell(\partial),f_{\ell}^n]
\,,}
\end{array}
\end{equation}
where
\begin{equation}\label{eq:prb6}
B_{\ell}^n(\partial)
:=
(A_{\ell+1}\dots A_sB^{n-K}A_1\dots A_\ell)(\partial)
\,\in\mc V((\partial^{-1}))\otimes\End V
\,,
\end{equation}
and
\begin{equation}\label{eq:prb7}
f_{\ell}^n
:=
\partial^{-1}\res_z\tr \big( B_\ell^n(z)-B_{\ell-1}^n(z) \big)
\,\in\mc V/\ker\partial
\,.
\end{equation}
Note that $f_\ell^n$ is a well defined element of $\mc V/\ker\partial$
since, by Lemma \ref{lem:hn1}(b),
$$
\tint \res_z\tr B_\ell^n(z)=\tint \res_z\tr B_{\ell-1}^n(z)
\,.
$$

We now recall the definition \eqref{eq:prod} of $L(\partial)$
and compute $\frac{dL(w)}{dt_{n,B}}$ using the Leibniz rule:
\begin{equation}\label{eq:prb8}
\begin{array}{l}
\displaystyle{
\frac{dL(\partial)}{dt_{n,B}}
=
\sum_{\ell=1}^s A_1(\partial)\dots A_{\ell-1}(\partial)
\frac{dA_\ell(\partial)}{dt_{n,B}}
A_{\ell+1}(\partial)\dots A_s(\partial)
} \\
\displaystyle{
\vphantom{\Big(}
=
-
\sum_{\ell=1}^s \alpha_\ell
A_1(\partial)\dots A_{\ell}(\partial)
B_{\ell}^n(\partial)_+
A_{\ell+1}(\partial)\dots A_s(\partial)
} \\
\displaystyle{
\vphantom{\Big(}
+
\sum_{\ell=1}^s \alpha_\ell
A_1(\partial)\dots A_{\ell-1}(\partial)
B_{\ell-1}^n(\partial)_+
A_{\ell}(\partial)\dots A_s(\partial)
} \\
\displaystyle{
\vphantom{\Big(}
+
\sum_{\ell=1}^s \beta_\ell
A_1(\partial)\dots A_{\ell}(\partial)
((B_{\ell-1}^n)^*(\partial)_+)^{\dagger}
A_{\ell+1}(\partial)\dots A_s(\partial)
} \\
\displaystyle{
\vphantom{\Big(}
-
\sum_{\ell=1}^s \beta_\ell
A_1(\partial)\dots A_{\ell-1}(\partial)
((B_{\ell}^n)^*(\partial)_+)^{\dagger}
A_{\ell}(\partial)\dots A_s(\partial)
} \\
\displaystyle{
\vphantom{\Big(}
-
\sum_{\ell=1}^s \gamma_\ell
A_1(\partial)\dots A_{\ell-1}(\partial)
[A_\ell(\partial),f_{\ell}^n]
A_{\ell+1}(\partial)\dots A_s(\partial)
\,.}
\end{array}
\end{equation}
For $\alpha_1=\dots=\alpha_s$, the first two terms of the RHS of \eqref{eq:prb8}
form a telescopic sum where only the $\ell=s$ term of the first sum and the $\ell=1$ term of the second sum
survive. As a result we get
\begin{equation}\label{eq:prb9}
-
\alpha_1
L(\partial)
B^n(\partial)_+
+
\alpha_1
B^n(\partial)_+
L(\partial)
\,.
\end{equation}
Here we used the fact that $A_1\dots A_s=L$ and $B_0^n=B_s^n=B^n$.
This proves equation \eqref{eq:hierarchy-prod} in case (i).

We are left to prove equation \eqref{eq:hierarchy-prod} in case (ii) (i.e. when $s=2$, $\alpha_1=\alpha_2$, 
$\beta_1=\beta_2$ and $\gamma_1=-\gamma_2$).
In this case, the $\beta$-terms in the RHS of \eqref{eq:prb8} are
\begin{equation}\label{eq:prb10}
\begin{array}{l}
\displaystyle{
\beta_1
A_1(\partial)
((B^n)^*(\partial)_+)^{\dagger}
A_{2}(\partial)
+
\beta_2
L(\partial)
((B_{1}^n)^*(\partial)_+)^{\dagger}
} \\
\displaystyle{
\vphantom{\Big(}
-
\beta_1
((B_{1}^n)^*(\partial)_+)^{\dagger}
L(\partial)
-
\beta_2
A_1(\partial)
((B^n)^*(\partial)_+)^{\dagger}
A_2(\partial)
=
-\beta_1 [((B_{1}^n)_+^{\star\dagger},L]
\,.}
\end{array}
\end{equation}
where we used the identities $A_1A_2=L$, $B_0^n=B_2^n=B^n$.
Finally,
let us compute the $\gamma$-terms in the RHS of \eqref{eq:prb8}.
Note that, for $s=2$, we have, by \eqref{eq:prb7},
$$
f_1^n=-f_2^n=\partial^{-1}\res_z\tr(B_1^n(z)-B^n(z))\,.
$$
Hence, the $\gamma$-term in the RHS of \eqref{eq:prb8} is
\begin{equation}\label{eq:prb11}
\begin{array}{l}
\displaystyle{
-
\gamma_1
[A_1(\partial),f_{1}^n] A_2(\partial)
-
\gamma_2
A_1(\partial)
[A_2(\partial),f_{2}^n]
=
-\gamma_1[L(\partial),f_{1}^n]
\,.}
\end{array}
\end{equation}
This completes the proof of \eqref{eq:hierarchy-prod} and of the Theorem.
\end{proof}

\section{Integrable hierarchies for \texorpdfstring{$\mc W$}{W}-algebras associated to classical Lie algebras}\label{sec:8}

As a consequence of Theorem \ref{thm:hn-product}
and the results of Sections \ref{sec:3}-\ref{sec:6},
we get the following list of hierarchies of Hamiltonian equations
associated to $\mc W$-algebras for classical Lie algebras:

\medskip

{\bf 1.}
Let $s\in\mb Z_{\geq1}$, $N_1,\dots,N_s\in\mb Z_{\geq1}$,
let $f_1\in\mf{gl}_{N_1},\dots,f_s\in\mf{gl}_{N_s}$ be nilpotent elements
(of depths $D_1,\dots,D_s$ respectively).
Let $T_1\in(\End\mb F^{N_1})[D_1],\dots,T_s\in(\End\mb F^{N_s})[D_s]$
satisfy condition \eqref{20170317:eq2}
and $\dim(\im T_1)=\dots=\dim(\im T_s)$.
Fix linear isomorphisms $U:=\im T_1\simeq\dots\simeq\im T_s$.
Consider the PVA 
$$
\mc W=\mc W(\mf{gl}_{N_1},f_1)\otimes\dots\otimes\mc W(\mf{gl}_{N_s},f_s)
\,,
$$
and the operator
$$
L(\partial)=L(\mf{gl}_{N_1},f_1,T_1)(\partial)\dots L(\mf{gl}_{N_s},f_s,T_s)(\partial)
\,\in\mc W((\partial^{-1}))\otimes\End U
\,,
$$
where $L(\mf g,f,T)$ is defined by \eqref{eq:L0}.
It defines an integrable hierarchy of Hamiltonian equations over $\mc W$
of Lax form:
\begin{equation}\label{eq:list1}
\frac{dL(w)}{dt_{n,B}}
=
[(B^n)_+,L](w)
\,,
\end{equation}
where  $n\in\mb Z$ and $B\in\mc W((\partial^{-1}))\otimes\End U$ is a root of $L$.
Moreover, the Hamiltonian functionals $\tint h_{n,B}$, $n\in\mb Z$ and $B$ a root of $L$,
defined by \eqref{eq:hn},
are integrals of motions in involution for all the equations of the hierarchy \eqref{eq:list1}.
\begin{remark}\label{rem:exp}
Recalling Table \eqref{table-3},
we could consider a larger class of operators:
$$
L(\partial)=e^{k_1\partial}L(\mf{gl}_{N_1},f_1,T_1)(\partial)e^{k_2\partial}
\dots e^{k_s\partial}L(\mf{gl}_{N_s},f_s,T_s)(\partial)e^{k_{s+1}\partial}
\,,
$$
where $k_1,\dots,k_{s+1}\in\mb F$,
and we assume that $\dim(\im T_i)=1$ for all $i$.
It should still be of Adler type, so it should define a Hamiltonian hierarchy of Lax equations.
But to make sense of it we need to enlarge the algebra of operators that we consider,
since, if we expand the exponentials, $L(\partial)$ has infinitely many positive powers of $\partial$
and we might get diverging series.
\end{remark}

\medskip

{\bf 2.}
Let $N\in\mb Z_{\geq1}$,
let $f\in\mf{sl}_{N}$ be a nilpotent element
(of depth $D$),
and let $T\in(\End\mb F^{N})[D]$ satisfy condition \eqref{20170317:eq2}.
Consider the PVA 
$\mc W=\mc W(\mf{sl}_{N},f)$,
and the operator
$$
L(\partial)=L(\mf{sl}_{N},f,T)(\partial)
\,\in\mc W((\partial^{-1}))\otimes\End(\im T)
\,,
$$
defined by \eqref{eq:L0}.
It defines an integrable hierarchy of Hamiltonian equations over $\mc W$
of the Lax form \eqref{eq:list1},
and the Hamiltonian functionals $\tint h_{n,B}$, $n\in\mb Z$ and $B$ root of $L$,
defined by \eqref{eq:hn},
are integrals of motions in involution for all the equations of the hierarchy \eqref{eq:list1}.
Note that this hierarchy is bi-Hamiltonian in both cases 1. and 2.,
cf. \cite{DSKVnew,DSKV16b}.

\medskip

{\bf 3.}
Let $\mf g\simeq\mf{so}_N$ or $\mf{sp}_N$, with $N\in\mb Z_{\geq1}$,
and let $f\in\mf g$ be a nilpotent element
(of depth $d$ in $\mf g$ and $D$ in $\End \mb F^N$).
Let $T\in(\End\mb F^{N})[D]$
satisfy condition \eqref{20170317:eq2}
and be such that $T^\dagger=\pm T$.
Consider the PVA 
$\mc W=\mc W(\mf g,f)$,
and the operator
$$
L(\partial)
=
L(\mf g,f,T)(\partial)
\,\in\mc W((\partial^{-1}))\otimes\End(\im T)
\,.
$$
It defines an integrable hierarchy of Hamiltonian equations over $\mc W$
of Lax form:
\begin{equation}\label{eq:list2}
\frac{dL(w)}{dt_{n,B}}
=
\frac12[(B^n)_+ - (B^n)^{\star\dagger}_+,L](w)
\,,
\end{equation}
where  $n\in\mb Z$ and $B\in\mc W((\partial^{-1}))\otimes\End(\im T)$ is a root of $L$.
Moreover, the Hamiltonian functionals $\tint h_{n,B}$, $n\in\mb Z$ and $B$ root of $L$,
defined by \eqref{eq:hn},
are integrals of motions in involution for all the equations of the hierarchy \eqref{eq:list2}.
This hierarchy is bi-Hamiltonian provided that $d=D$,
which is always the case for $\mf g=\mf{sp}_N$.

\medskip

{\bf 4.}
Let $\mf g,N,f,T$ and $\mc W$ be as in 3.,
assume that $\dim(\im T)=1$,
and consider the operator
$$
L(\partial)
=
L(\mf g,f,T)(\partial)\partial^{\pm1}
\,\in\mc W((\partial^{-1}))
\,.
$$
It defines an integrable hierarchy of Hamiltonian equations over $\mc W$
of Lax form:
\begin{equation}\label{eq:list3}
\frac{dL(w)}{dt_{n,B}}
=
\frac12[(B^n)_+ - (\partial^{\pm1}B^{n}\partial^{\mp1})^{\star\dagger}_+,L](w)
\,,
\end{equation}
where  $n\in\mb Z$ and $B\in\mc W((\partial^{-1}))$ is a $K$-th root of $L$.
Moreover, the Hamiltonian functionals $\tint h_{n,B}$, $n\in\mb Z$ and $B$ root of $L$,
defined by \eqref{eq:hn},
are integrals of motions in involution for all the equations of the hierarchy \eqref{eq:list3}.
Similarly, we can consider the operator 
$$
L(\partial)
=
\partial^{\pm1}L(\mf g,f,T)(\partial)
\,\in\mc W((\partial^{-1}))
\,,
$$
and, in this case, we need to replace $\partial^{\pm1}B^{n-K}L(\mf g,f,T)$ by $L(\mf g,f,T)B^{n-K}\partial^{\pm1}$
in \eqref{eq:list3}.

\medskip

{\bf 5.}
For $\ell=1,2$, let $\mf g_\ell\simeq\mf{so}_{N_\ell}$ or $\mf{sp}_{N_\ell}$, with $N_\ell\in\mb Z_{\geq1}$,
and let $f_\ell\in\mf g_\ell$ be a nilpotent element
(of depth $D_\ell$ in $\End \mb F^{N_\ell}$).
Let $T_\ell\in(\End\mb F^{N_\ell})[D_\ell]$
satisfy condition \eqref{20170317:eq2}
and be such that $T_\ell^\dagger=\pm T_\ell$.
Assume moreover that the there is an isometry $U:=\im T_1\simeq\im T_2$.
Consider the PVA 
$$
\mc W=\mc W(\mf g_1,f_1)\otimes\mc W(\mf g_2,f_2)
\,,
$$
and the operator
$$
L(\partial)
=
L(\mf g_1,f_1,T_1)(\partial)
L(\mf g_2,f_2,T_2)(\partial)
\,\in\mc W((\partial^{-1}))\otimes\End U
\,.
$$
It defines an integrable hierarchy of Hamiltonian equations over $\mc W$
of Lax form:
\begin{equation}\label{eq:list4}
\frac{dL(w)}{dt_{n,B}}
=
\frac12[(B^n)_+ - (L(\mf g_2,f_2,T_2)B^{n-K}L(\mf g_1,f_1,T_1))^{\star\dagger}_+,L](w)
\,,
\end{equation}
where  $n\in\mb Z$ and $B\in\mc W((\partial^{-1}))\otimes\End U$ is a $K$-th root of $L$.
Moreover, the Hamiltonian functionals $\tint h_{n,B}$, $n\in\mb Z$, and $B$ root of $L$,
defined by \eqref{eq:hn},
are integrals of motions in involution for all the equations of the hierarchy \eqref{eq:list4}.

\begin{remark}
By Remark \ref{rem:modified}, we get integrable hierarchies of Lax type equations for the PVA
$\mc V=\mc V(\mf g_0)\otimes\mc F(\mf g_{\frac12})$ by replacing the pseudodifferential
operator $L$ in \eqref{eq:list1}, \eqref{eq:list2}, \eqref{eq:list3} and \eqref{eq:list4} by the corresponding modified Adler operator
$L^{mod}$ defined in \eqref{eq:Lmod}.
\end{remark}

\section{Examples of generalized Adler type operators}\label{sec:9}

\subsection{Example 1: Lax operator $L_\epsilon(z)$ for principal nilpotent $f$ in $\mf{gl}_N$}\label{sec:example1.1}

This computation already appeared in \cite{DSKV16b}.
Let $\mf g=\mf{gl}_N$ in its standard representation. 
As usual, we denote by $E_{ij}\in\mf{gl}_N$ the elementary matrix 
with $1$ in position $(ij)$ and $0$ elsewhere
and we denote by $e_{ij}\in\mf{g}$ the same matrix when viewed as an element 
of the differential algebra $\mc V(\mf g)$.
Consider the principal nilpotent element 
$f=\sum_{k=1}^{N-1}e_{k+1,k}\in\mf g$.
In this case the depths of the grading \eqref{eq:grading} and \eqref{eq:grading_EndV} coincide: $d=D=N-1$.
Moreover, $\mf{g}_{N-1}=\mb F e_{1N}$. Hence, we may choose $T=S=E_{1N}$.

As a subspace $U\subset\mf{g}$ complementary to $[f,\mf{g}]$ and compatible  with the grading \eqref{eq:grading} we choose
(see \cite{DSKV16b}) $U=\Span\{u^i\mid 0\leq i\leq N-1\}$, where
\begin{equation}\label{eq:U_gl}
u^i=e_{1,N-i}
\,,
\qquad
0\leq i\leq N-1
\,.
\end{equation}
The dual basis $\{u_i\mid 0\leq i\leq N-1\}$ of $\mf g^f$ is given by
$$
u_i=\sum_{k=1}^{i+1}e_{N+k-i-1,k}
\,,
\qquad
0\leq i\leq N-1
\,.
$$
Let $w:\mc V(\mf g^f)\to\mc W(\mf g,f)$ be the differential algebra isomorphism given in Theorem \ref{thm:structure-W}, and let $w_i=w(u_i)$, for $0\leq i\leq N-1$.
By Theorem \ref{thm:main} we have \cite[Sec.7.1]{DSKV16b}
\begin{equation}\label{eq:glN_quasidet}
\begin{array}{l}
\displaystyle{
\vphantom{\Big(}
L_\epsilon(\partial)
=|\id_N\partial+F+\sum_{i=0}^{N-1}w_{i} E_{1N-i}+\epsilon S|_{1N}
} \\
\displaystyle{
\vphantom{\Big(}
=
-(-\partial)^N
+\sum_{i=0}^{N-1}w_i(-\partial)^i
+\epsilon
\in\mc W(\mf{g},f)[\partial]
\,.}
\end{array}
\end{equation}
Thus in this case the Adler type operator $L_\epsilon(\partial)$
is the generic scalar differential operator of order $N$ \cite{GD76,DS85,AGD}.

\begin{remark}\label{rem:OK}
The generators $w_i,\,i=0,\dots,N-1$, can be computed explicitly,
as elements of $\mc V(\mf g_{\leq\frac12})$,
comparing equations \eqref{eq:L} and \eqref{eq:glN_quasidet}
for the matrix $L_0(\partial)$.
We thus recover the formula for the $w_i$'s obtained in \cite[Sec.7.1]{DSKV16b}.
The generators of the $W$-algebra in the remaining examples
can be computed in a similar way, but we will not provide the details.
\end{remark}

\begin{remark}\label{rem:modified_gln}
By equation \eqref{eq:Lmod}, the Miura map $\mu:\mc W(\mf g,f)\to\mc V(\mf g_0)$ is given by
(see also \cite{DS85})
$$
\mu(L_0(\partial))=(-1)^{N+1}(\partial+e_{11})\dots(\partial+e_{NN})
\,\in\mc V(\mf g_0)[\partial]\,.
$$
\end{remark}

\subsection{Example 2: Lax operator $L_\epsilon(z)$ for principal nilpotent $f$ in $\mf{sl}_N$}\label{sec:example1.2}

This computation is similar to the one performed in Section \ref{sec:example1.1}.
Let $\mf g=\mf{sl}_N$ in its standard representation. Given the elementary matrix $E_{ij}\in\mf{gl}_N$ we denote by $E_{ij}^\sharp=E_{ij}-\frac{\delta_{ij}}{N}\id_N$
its projection on $\mf{sl}_N$
and we denote by $e^\sharp_{ij}\in\mf{g}$ the same matrix when viewed as an element of the differential algebra $\mc V(\mf g)$.
Consider the principal nilpotent element 
$f=\sum_{k=1}^{N-1}e_{k+1,k}\in\mf{sl}_N$.
As in Section \ref{sec:example1.1}, we have
$d=D=N-1$, $\mf{g}_{N-1}=\mb F e_{1N}$, and we choose $T=S=E_{1N}$.

As a subspace $U\subset\mf{g}$ complementary fo $[f,\mf{g}]$ and compatible  with the grading \eqref{eq:grading} we choose
$U=\Span\{u^i\mid 0\leq i\leq N-2\}$, where
\begin{equation}\label{eq:U_sl}
u^i=e_{1,N-i}
\,,
\qquad
0\leq i\leq N-2
\,.
\end{equation}
The dual basis $\{u_i\mid 0\leq i\leq N-2\}$ of $\mf g^f$ is given by
$$
u_i=\sum_{k=1}^{i+1}e_{N+k-i-1,k}
\,,
\qquad
0\leq i\leq N-2
\,.
$$
Let $w:\mc V(\mf g^f)\to\mc W(\mf g,f)$ be the differential algebra isomorphism given in Theorem \ref{thm:structure-W}, and let $w_i=w(u_i)$, for $0\leq i\leq N-2$.
As in the case of $\mf{gl}_N$, we have
\begin{equation}\label{eq:slN_quasidet}
\begin{array}{l}
\displaystyle{
\vphantom{\Big(}
L_\epsilon(\partial)
=|\id_N\partial+f+\sum_{i=0}^{N-2}w_{i} E_{1N-i}+\epsilon S|_{1N}
} \\
\displaystyle{
\vphantom{\Big(}
=
-(-\partial)^N
+\sum_{i=0}^{N-2}w_i(-\partial)^i
+\epsilon
\in\mc W(\mf{g},f)[\partial]
\,.}
\end{array}
\end{equation}

\begin{remark}\label{rem:modified_sln}
By equation \eqref{eq:Lmod}, the Miura map $\mu:\mc W(\mf g,f)\to\mc V(\mf g_0)$ is given by
(see also \cite{DS85})
$$
\mu(L_0(\partial))=(-1)^{N+1}(\partial+e^\sharp_{11})\dots(\partial+e^\sharp_{NN})
\in\mc V(\mf g_0)[\partial]\,.
$$
\end{remark}

\subsection{Example 3: Lax operator $L_\epsilon(z)$ for principal nilpotent $f$ in $\mf{g}$ of type $B$ or $C$}\label{sec:example1.3}

For $N\geq2 $ let $V=\mb F^N$ be an $N$-dimensional vector space with basis $\{v_i\}_{i=1}^N$, and let
\begin{equation}\label{20170424:eq1}
\epsilon_i=(-1)^i\,,
\qquad
1\leq i\leq N
\,.
\end{equation}
We define a non-degenerate bilinear form on $V$ as follows:
\begin{equation}\label{20170411:form}
\langle v_i|v_j\rangle=-\epsilon_{i}\delta_{i,j'}
\,,
\qquad i,j=1,\dots,N\,,
\end{equation}
where $i'=N+1-i$.
It follows from \eqref{20170411:form} that
\begin{equation}\label{20170411:parity}
\langle u|v\rangle
=
(-1)^{N+1}
\langle v|u\rangle
\,.
\qquad u,v\in V\,,
\end{equation}
Let $A^\dagger$ denote the adjoint of $A\in\End V$ with respect to \eqref{20170411:form}.
Explicitly, in terms of elementary matrices, it is given by:
$$
(E_{ij})^\dagger=\epsilon_i\epsilon_jE_{j'i'}\,.
$$

Let $\mf g=\{A\in\End V\mid A^\dagger=-A\}$.
Then,
by equation \eqref{20170411:parity},
$\mf g\simeq\mf{so}_N$ if $N$ is odd, and $\mf g\simeq\mf{sp}_N$ if $N$ is even. 
For $i,j=1,\dots,N$, we define
\begin{equation}\label{eq:F}
F_{ij}=E_{ij}-\epsilon_i\epsilon_jE_{j'i'}\,\,\big(=-F_{ij}^\dagger\big)
\,.
\end{equation}
The following commutation relations hold ($i,j,h,k=1,\dots,N$):
\begin{equation}\label{comm:BC}
[F_{ij},F_{hk}]=\delta_{jh}F_{ik}-\delta_{ki}F_{hj}-\epsilon_i\epsilon_j\delta_{i'h}F_{j'k}
+\epsilon_i\epsilon_j\delta_{kj'}F_{hi'}
\,.
\end{equation}
By \eqref{eq:F} the following elements form a basis of $\mf g$
$$
\big\{\frac{1}{1+\delta_{ij'}}f_{ij}:=\frac{1}{1+\delta_{ij'}}(e_{ij}-\epsilon_i\epsilon_je_{j'i'})
\,\big|\,
(i,j)\in I\big\}
\,,
$$
where
$$
I=\Bigg\{
\begin{array}{l}
\displaystyle{
\vphantom{\Big(}
\big\{(i,j)\,\big|\, 1\leq i\leq N, 1\leq j\leq i'\big\}
\,\,\text{ if }\,\, N \text{ is even}
} \\
\displaystyle{
\vphantom{\Big(}
\big\{(i,j)\,\big|\, 1\leq i\leq N, 1\leq j< i'\big\}
\,\,\text{ if }\,\, N \text{ is odd}
\,.}
\end{array}
$$
Its dual basis, with respect to the trace form \eqref{20170317:eq1}, is 
$$
\big\{\frac12f_{ji}\mid (i,j)\in I\big\}\,.
$$ 

Let $f=\sum_{k=1}^{N-1}e_{k+1,k}\in\End V$. Then, $f\in\mf g$ and it is a principal nilpotent element. 
Indeed we can write
\begin{equation}\label{20170411:f}
f=\Bigg\{
\begin{array}{l}
\displaystyle{
\vphantom{\Big(}
f_{21}+f_{32}+\dots+f_{n,n-1}+\frac{1}{2}f_{n+1,n}
\,\,\text{ if }\,\, N \text{ is even}
} \\
\displaystyle{
\vphantom{\Big(}
f_{21}+f_{32}+\dots+f_{n,n-1}+f_{n+1,n}
\,\,\text{ if }\,\, N \text{ is odd}
\,.}
\end{array}
\,,
\end{equation}
where we denote by $n=[\frac{N}{2}]$ the integer part of $\frac{N}{2}$.
We can include $f\in\mf g$ in the following $\mf{sl}_2$-triple $\{e,h=2x,f\}\subset\mf g$,
where:
\begin{equation}\label{20170411:sl2tripleBC}
x=\sum_{k=1}^n\frac{N\!+\!1\!-\!2k}2 f_{kk}
\,,
\quad
e=
\Bigg\{\begin{array}{l}
\displaystyle{
\vphantom{\Big(}
\sum_{k=1}^{n-1}k(N\!-\!k)f_{k,k+1}+\frac{n^2}{2} f_{n,n+1}
\,\,\text{ if }\,\, N \text{ is even}
} \\
\displaystyle{
\vphantom{\Big(}
\sum_{k=1}^{n-1}
k(N\!-\!k)f_{k,k+1}+n(n\!+\!1) f_{n,n+1}
\,\,\text{ if } N \text{ is odd}
}
\end{array}
\end{equation}
It is immediate to check using the expression of $x\in\mf g$ 
and the commutation relations \eqref{comm:BC} that
$$
d=2n-1
\,\,,\,\,\,\,
\mf g_d=\mb Ff_{1,2n}
\,\,,\,\,\,\,
D=N-1
\,\,,\,\,\,\,
(\End V)[N-1]=\mb FE_{1N}
\,.
$$
So, $D=d$ only for even $N$.
We thus choose 
$s=\frac12 f_{1,2n}$ and $T=E_{1N}$, 
and $S=T$ for even $N$.

As a subspace $U\subset\mf{g}$ complementary fo $[f,\mf{g}]$ and compatible  with the grading \eqref{eq:grading} we choose
$U=\Span\{u^i\mid 1\leq i\leq n\}$, where
\begin{equation}\label{eq:U_BC}
u^i=\frac12f_{1,2(n+1-i)}
\,,
\qquad
1\leq i\leq n
\,.
\end{equation}
The dual basis $\{u_i\mid 1\leq i\leq n\}$ of $\mf g^f$ is given by
$$
u_i=\sum_{k=1}^{N-1-2(n-i)}e_{k+1+2(n-i),k}
\,,
\qquad
1\leq i\leq n
\,.
$$
Let $w:\mc V(\mf g^f)\to\mc W(\mf g,f)$ be the differential algebra isomorphism given in Theorem \ref{thm:structure-W}, and let us
denote by $w_i=w(u_i)$, for $1\leq i\leq n$.
By Theorem \ref{thm:main} we have that (cf. \cite[Prop.4.2]{DSKVnew})
\begin{equation}\label{eq:BC_quasidet}
L_\epsilon(\partial)
=
|\id_N\partial+F+\frac12\sum_{i=1}^{n}w_{i} F_{1,2(n+1-i)}+\epsilon S|_{1N}
\,.
\end{equation}
By an explicit computation we get
\begin{equation}
\begin{split}\label{eq:BC_L}
L_\epsilon(\partial)
&=-(-\partial)^N+\frac12\sum_{k=1}^n\left(
(-\partial)^{N-2k}\circ w_{n+1-k}+w_{n+1-k}(-\partial)^{N-2k}
\right)
\\
&-\frac14\sum_{k=1}^{n-1}\sum_{h=1}^{n-k}w_{n+1-h}(-\partial)^{N-2(h+k)}\circ w_{n+1-k}+\epsilon(-\partial)^{N-2n}
\in\mc W(\mf g,f)[\partial]
\,.
\end{split}
\end{equation}
Note that $L_\epsilon(\partial)=(-1)^NL_\epsilon(\partial)^*$,
in other words, for $N$ even we get a generic selfadjoint operator,
and for $N$ odd a generic skewadjoint operator \cite{DS85}.

We can rewrite the operator $L_{\epsilon}(\partial)$ (for $\epsilon=0$) given by equations \eqref{eq:BC_L} as
$$
L_0(z)=-(-z)^N+\sum_{k=0}^{N-2}\widetilde w_kz^k
\in\mc W(\mf g,f)[z]
\,.
$$
The elements $\widetilde w_{N-2k}\in\mc W(\mf g,f)$, for $k=1,\dots,n$, provide a different set 
of differential generators for
the differential algebra $\mc W(\mf g,f)$ which coincide 
(up to a rescaling of the variables $f_{ij}$ )
with the set of generators constructed in \cite{MR15}.

\begin{remark}\label{rem:modified_BC}
By equation \eqref{eq:Lmod}, the Miura map $\mu:\mc W(\mf g,f)\to\mc V(\mf g_0)$ is given by
(see also \cite{DS85})
$$
\mu(L_0(\partial))=
(-1)^{N+1}
(\partial+\frac12f_{11})\dots(\partial+\frac12 f_{n,n})\partial^{N-2n}(\partial-\frac12f_{nn})\dots(\partial-\frac12f_{11})
\,.
$$
\end{remark}

\subsection{Example 4: Lax operator $L_\epsilon(z)$ for principal nilpotent in $\mf{g}$ of type $D$}\label{sec:example1.4}

Let $N=2n\geq2 $ and let $V=\mb F^N$ be an $N$-dimensional vector space with basis $\{v_i\}_{i=1}^N$.
We introduce the the following notation:
for $i\in\{1,\dots,N\}$, we let
\begin{equation}\label{20170421:eq1}
i'=\left\{
\begin{array}{ll}
N-i\,,&1\leq i\leq N-1\,,
\\
N\,,&i=N\,,
\end{array}
\right.
\text{and}
\quad
\epsilon_i=
\left\{
\begin{array}{ll}
(-1)^i\,,&1\leq i\leq N-1\,,
\\
(-1)^{n+1}\,,&i=N\,.
\end{array}
\right.
\end{equation}
Then, we define a symmetric non-degenerate bilinear form on $V$ as follows:
\begin{equation}\label{20170411:form_D}
\langle v_i|v_j\rangle=
-\epsilon_i\delta_{j,i'}
\,.
\end{equation}
The adjoint with respect to \eqref{20170411:form_D}
is then
$$
(E_{ij})^\dagger=\epsilon_i\epsilon_jE_{j'i'}
\,.
$$
Let $\mf g=\{A\in\End V\mid A^\dagger=-A\}\simeq\mf{so}_N$.
For $1\leq i,j\leq N-1$ we let $F_{ij}$ be defined as in \eqref{eq:F}, and, as usual,
we also denote by $f_{ij}$ the same elements when viewed as elements 
of $\mf g\subset\mc V(\mf g)$.
It is immediate to check
that the same commutation relations as in \eqref{comm:BC} hold.
A basis of $\mf g$ is
$\{f_{ij}\mid (ij)\in I\}$ where
$$
I=\{(ij)\,\mid1\leq i\leq N-2,1\leq j<i'\}\cup\{(Ni)\mid 1\leq i\leq N-1\}\,. 
$$
Its dual basis, with respect to the trace form \eqref{20170317:eq1}, is $\{\frac12f_{ji}\mid (i,j)\in I\}$. 

Let $f=\sum_{k=1}^{n-1}f_{k+1,k}\in\mf g$. It is a nilpotent element associated to the partition $p=(N-1,1)$.
Hence, it is a principal nilpotent element.
We can include $f\in\mf g$ in the following $\mf{sl}_2$-triple $\{e,h=2x,f\}\subset\mf g$,
where:
\begin{equation}\label{20170411:sl2tripleD}
h=\sum_{k=1}^{n-1}(N-2k)f_{kk}
\,,
\qquad
e=\sum_{k=1}^{n-1}k(N-1-k)f_{k,k+1}
\,.
\end{equation}
It is immediate to check that:
$$
d=N-3
\,\,,\,\,\,\,
\mf g_d=\mb FF_{1,N-2}
\,\,,\,\,\,\,
D=N-2
\,\,,\,\,\,\,
(\End V)[N-2]=\mb FE_{1N-1}
\,.
$$
We thus choose $T=E_{1,N-1}$, and $s=\frac{1}{2}f_{1,N-2}$.

As a subspace $U\subset\mf{g}$ complementary fo $[f,\mf{g}]$ and compatible  with the grading \eqref{eq:grading} we choose
$U=\Span\{u^i\mid 0\leq i\leq n-1\}$, where
\begin{equation}\label{eq:U_D}
u^i=\frac12f_{1,N-2i}
\,,
\qquad
0\leq i\leq n-1
\,.
\end{equation}
The dual basis $\{u_i\mid 0\leq i\leq n-1\}$ of $\mf g^f$ is given by
$$
u_0=f_{N,N-1}\,,
\quad
u_i=\sum_{k=1}^{i}f_{N+k-2i-1,k}
\,,
\qquad
1\leq i\leq n
\,.
$$
Let $w:\mc V(\mf g^f)\to\mc W(\mf g,f)$ be the differential algebra isomorphism given in Theorem \ref{thm:structure-W}, and let $w_i=w(u_i)$, for $0\leq i\leq n-1$.
By Theorem \ref{thm:main} and \cite[Eq.(2.14)]{DSKV16b} (see also \cite[Prop.4.2]{DSKVnew}),
we have
\begin{equation}\label{eq:D_quasidet}
\begin{split}
&L_\epsilon(\partial)
=|\id_N\partial+F+\frac12\sum_{i=0}^{n-1}w_{i} F_{1N-2i}+\epsilon S|_{1N-1}
=
\\
&-
\left(\begin{array}{ccccccccc}
\partial&\frac12w_{n-1}&0&\frac12w_{n-2}&0&\dots&0&\frac12(w_{1}+\epsilon)&\frac12w_0
\end{array}\right)
\circ
\\
&\left(\begin{array}{cccccc}
1&\partial&0&\dots&0 &0\\
0&1&\ddots&\ddots&\vdots &\vdots\\
\vdots&\ddots&\ddots&\ddots&0&0\\
\vdots&&\ddots&\ddots&\partial &0\\
0&\dots&\dots&0&1&0
\\
0&\dots&\dots&0&0&\partial
\end{array}\right)^{-1}
\circ
\left(\begin{array}{c}
\frac12 (w_1+\epsilon)\\ 0\\ \frac12 w_2\\ 0 \\ \vdots \\ \frac12 w_{n-1} \\ \partial\\ -\frac12 w_0
\end{array}\right)
\,.
\end{split}
\end{equation}
We can compute the inverse matrix in the RHS of \eqref{eq:D_quasidet} by expanding in geometric series in its upper left block to get
\begin{equation}
\begin{split}\label{eq:D_L}
L_\epsilon(\partial)
&=\partial^{N-1}-\frac12\sum_{k=1}^{n-1}\left(
\partial^{N-1-2k}\circ w_{n-k}+w_{n-k}\partial^{N-1-2k}
\right)
\\
&+\frac14\sum_{k=1}^{n-2}\sum_{h=1}^{n-k-1}w_{n-h}\partial^{N-1-2(h+k)}\circ w_{n-k}
\\
&+\frac14 w_0\partial^{-1}\circ w_0
-\epsilon\partial\in\mc W(\mf g,f)((\partial^{-1}))
\,.
\end{split}
\end{equation}
Note that $L_{\epsilon}(\partial)=-L_\epsilon(\partial)^*$,
so it is a weakly non-local skewadjoint pseudodifferential operator.

We can rewrite the operator $L_\epsilon(\partial)$ for $\epsilon=0$ given by equations \eqref{eq:D_L} as
$$
L_0(z)=z^{N-1}+\sum_{k=0}^{N-3}\widetilde w_kz^k+(-1)^ny_0(z+\partial)^{-1}y_0
\in\mc W(\mf g,f)((z^{-1}))
\,.
$$
The elements $y_0=\frac{\sqrt{(-1)^n}}2w_0$ and $\widetilde w_{2k-1}\in\mc W(\mf g,f)$, 
for $k=1,\dots,n-1$, provide a different set of differential generators for
the differential algebra $\mc W(\mf g,f)$ which coincide 
(up to a rescaling of the variables $f_{ij}$)
with the set of generators constructed in \cite{MR15}.

\begin{remark}\label{rem:modified_D}
By equation \eqref{eq:Lmod}, the Miura map $\mu:\mc W(\mf g,f)\to\mc V(\mf g_0)$ is given by
(see also \cite{DS85})
$$
\begin{array}{l}
\displaystyle{
\vphantom{\Big(}
\mu(L_0(\partial))=
\big(\partial+\frac12f_{11}\big)
\dots
\big(\partial+\frac12 f_{n-1,n-1}\big)
} \\
\displaystyle{
\vphantom{\Big(}
\times
\big(\partial-\frac14 f_{Nn}\partial^{-1}\circ f_{Nn}\big)
\big(\partial-\frac12f_{n-1,n-1}\big)
\dots
\big(\partial-\frac12f_{11}\big)
\,.}
\end{array}
$$
\end{remark}

\subsection{Example 5: Lax operator $L_\epsilon(z)$ for minimal nilpotent $f$ in $\mf{gl}_N$}\label{sec:example1.5}

This computation already appeared in \cite{DSKV16b} and we briefly review it here.
Let $\mf g=\mf{gl}_N$ in its standard representation. 
The minimal nilpotent element $f=e_{N1}$ in $\mf g=\mf{gl}_N$
is associated to the partition $p=(2,1,\dots,1)$,
and it is embedded in the $\mf{sl}_2$-triple $\{e,h=2x,f\}\subset\mf g$, where
$$
e=e_{1N}\,,
\qquad
x=\frac{e_{11}-e_{NN}}{2}
\,.
$$
In this case the gradings \eqref{eq:grading} and \eqref{eq:grading_EndV} coincide
and their non-zero components are
\begin{align*}
&\mf g_{-1}=\mb Ff\,,
\qquad
\mf g_{-\frac12}=\Span\{e_{Nk},e_{k1}\mid 2\leq k\leq N-1\}\,,
\\
&\mf g_{0}=\Span\{e_{11},e_{NN},e_{hk}\mid 2\leq h,k\leq N-1\}\,,
\\
&
\mf g_{\frac12}=\Span\{e_{1k},e_{kN}\mid 2\leq k\leq N-1\}\,,
\qquad
\mf g_{1}=\mb Fe
\,.
\end{align*}
Hence, we may choose $T=S=E_{1N}$.

As a subspace $U\subset\mf{g}$ complementary fo $[f,\mf{g}]$ and compatible  with the grading \eqref{eq:grading} we choose
(see \cite{DSKV16b}) $U=\Span\{u^{ji}=e_{ji}\mid (ij)\in I_f\}$, where
$$
I_f=\{(11),(N1)\}\cup\{(Nk),(k1)\mid 2\leq k\leq N-1\}\cup\{(hk)\mid 2\leq h,k\leq N-1\}
\,.
$$
The dual basis $\{u_{ij}\mid (ij)\in I_f\}$ of $\mf g^f$ is given by
$$
u_{11}=e_{11}+e_{NN}\,,
\qquad
u_{ij}=e_{ij}\,,\qquad(ij)\in I_f\setminus\{(11)\}
\,.
$$
Let $w:\mc V(\mf g^f)\to\mc W(\mf g,f)$ be the differential algebra isomorphism 
given in Theorem \ref{thm:structure-W}, and let $w_{ij}=w(u_{ij})$ for $(ij)\in I_f$.
By Theorem \ref{thm:main} we have \cite[Sec.7.4]{DSKV16b}
\begin{equation}
\begin{split}
\label{eq:L1-minimal}
&
L_\epsilon(\partial)
=|\id_N\partial+F+\sum_{(ij)\in I_f}w_{ij}E_{ji}+\epsilon S|_{1N}
\\
&
=-\partial^2-w_{11}\partial+w_{N1}
- w_{+1}(\id_{N-2}\partial+W_{++})^{-1}\circ w_{N+}
+\epsilon
\,,
\end{split}
\end{equation}
where
\begin{equation}\label{eq:minimal-notation}
\begin{split}
& w_{+1}
=
\left(\begin{array}{lll} w_{21} & \dots & w_{N-1,1} \end{array}\right) 
\,,\\
& W_{++}
=
\left(\begin{array}{ccc} 
w_{22} & \dots & w_{N-1,2} \\
\vdots & \ddots & \vdots \\
w_{2N-1} & \dots & w_{N-1,N-1} 
\end{array}\right)
\,\,,\,\,\,\,
w_{N+}
=
\left(\begin{array}{l} w_{N2} \\ \vdots \\ w_{N,N-1} \end{array}\right)
\,.
\end{split}
\end{equation}

\subsection{Example 6: Lax operator $L_\epsilon(z)$ for minimal nilpotent $f$ in $\mf{sl}_N$}\label{sec:example1.6}

The minimal nilpotent element $f=e_{N1}$ in $\mf g=\mf{sl}_N$
is also associated to the partition $p=(2,1,\dots,1)$. Similarly to the computation in Section \ref{sec:example1.2}, we can recover
the results for the Lie algebra $\mf{sl}_N$ from the analogous results for $\mf{gl}_N$ provided in Section \ref{sec:example1.5}.
In this case, $L_\epsilon(\partial)$ is the scalar pseudodifferential operator in \eqref{eq:L1-minimal} where
$$
-w_{11}=w_{22}+w_{33}+\dots+w_{N-1N-1}
\,.
$$
We omit the details of this computation.

\subsection{Example 7: Lax operator $L_\epsilon(z)$ for minimal nilpotent $f$ in  $\mf{sp}_N$}\label{sec:example1.7}
Let $N=2n\geq2$ and $V=\mb F^N$ be an $N$-dimensional vector space endowed with the non-degenerate
skew-symmetric bilinear form defined in \eqref{20170411:form}. Let $A^\dagger$ denote the adjoint of $A\in\End V$ with respect to
\eqref{20170411:form} (recall that, in terms of elementary matrices, we have
$(E_{ij})^\dagger=(-1)^{i+j}E_{j'i'}$, where $i'=N+1-i$).
Then, $\mf g=\{A\in\End V\mid A=-A^\dagger\}\simeq\mf{sp}_N$. 
For $i,j=1,\dots,N$, let $F_{ij}$ be defined as in \eqref{eq:F}. Recall from Section \ref{sec:example1.3} that we have the following basis of $\mf g$
$$
\{\frac1{1+\delta_{ij'}}f_{ij}:=\frac1{1+\delta_{ij'}}(e_{ij}-\epsilon_i\epsilon_je_{j'i'})\mid (i,j)\in I\}
\,,
$$
where
$$
I=\{(i,j)\mid 1\leq i\leq N,1\leq j\leq i'\}
\,.
$$
and that its dual basis with respect to the trace form \eqref{20170317:eq1} is
$$
\{\frac{1}{2}f_{ji}\mid (ij)\in I\}
\,.
$$
Let $f=\frac12f_{N1}\in\mf g$. It is associated to the partition $p=(2,1,1,\dots,1)$, 
hence it is a minimal nilpotent element.
We can include $f$ in the following $\mf{sl}_2$-triple $\{e,h=2x,f\}\subset\mf g$,
where
$$
e=\frac12 f_{1N}\,,
\qquad
x=\frac12 f_{11}
\,.
$$
Using the commutation relations \eqref{comm:BC} one checks that 
the non-zero components of the grading \eqref{eq:grading} are
\begin{align*}
&\mf g_{-1}=\mb Ff\,,
\qquad
\mf g_{-\frac12}=\Span\{f_{k1}\mid 2\leq k\leq N-1\}\,,
\\
&\mf g_{0}=\Span\{f_{11},f_{hk}\mid 2\leq h\leq N-1,2\leq k\leq h'\}\,,
\\
&
\mf g_{\frac12}=\Span\{f_{1k}\mid 2\leq k\leq N-1\}\,,
\qquad
\mf g_{1}=\mb Fe
\,.
\end{align*}
Moreover, we have that $D=d=1$, and $(\End V)[1]=\mb FE_{1N}$.
Hence, we may choose $T=S=E_{1N}$.

As a subspace $U\subset\mf{g}$ complementary fo $[f,\mf{g}]$ and compatible  with the grading \eqref{eq:grading} we choose
$U=\Span\{u^{ji}=\frac1{2}f_{ji}\mid (ij)\in I_f\}$, where
$$
I_f=\{(k1)\mid 2\leq k\leq N\}\cup\{(hk)\mid 2\leq h\leq N-1,1\leq k\leq h'\}
\,.
$$
The dual basis $\{u_{ij}\mid (ij)\in I_f\}$ of $\mf g^f$ is given by $u_{ij}=\frac{1}{1+\delta_{ij'}}f_{ij}$, for $(ij)\in I_f$.

Let $w:\mc V(\mf g^f)\to\mc W(\mf g,f)$ be the differential algebra isomorphism given in Theorem \ref{thm:structure-W}, and let us
denote by $w_{ij}=w(u_{ij})$, for $(ij)\in I_f$.
By Theorem \ref{thm:main}, and performing a similar computation as in Section \ref{sec:example1.7}, we have
\begin{equation}
\begin{split}
\label{eq:L1-minimal_C}
&
L_\epsilon(\partial)
=|\id_N\partial+F+\frac12\sum_{(ij)\in I_f}w_{ij}F_{ji}+\epsilon S|_{1N}
\\
&
=-\partial^2+w_{N1}
- \frac14 w_{+1}(\id_{N-2}\partial+\frac12W_{++})^{-1}\circ \tilde w_{+1}
+\epsilon
\,,
\end{split}
\end{equation}
where
\begin{equation}\label{eq:minimal_C-notation}
\begin{split}
& w_{+1}
=
\left(\begin{array}{llll} w_{21} & w_{31}& \dots & w_{N-1,1} \end{array}\right) 
\,,\\
& W_{++}
=
\sum_{h=2}^{N-1}\sum_{k=2}^{h'}w_{kh}\bar F_{hk}
\,\,,\,\,\,\,
\tilde w_{+1}
=
\left(\begin{array}{c} -w_{N-1,1} \\ \vdots \\ -w_{31}\\ w_{21} \end{array}\right)
\,.
\end{split}
\end{equation}
In \eqref{eq:minimal_C-notation} $\bar F_{hk}$ denotes the matrix $F_{hk}$ where we have removed the first and last row, and the first
and last column.

\subsection{Example 8: Lax operator $L_\epsilon(z)$ for minimal nilpotent $f$ in $\mf{so}_N$}
\label{sec:example1.8}

For $N\geq2 $, let $V=\mb F^N$ be an $N$-dimensional vector space with basis $\{v_i\}_{i=1}^N$, and let us denote by $n=[\frac N2]$ the integer part of $\frac N2$.
We introduce the following notation for $i=1,\dots,N$:
\begin{equation}\label{20170424:eq2}
\epsilon_i=\left\{
\begin{array}{ll}
(-1)^i\,,& i=1,\dots,n\,,
\\
(-1)^{i'}\,,& i=n+1,\dots,N
\,,
\end{array}
\right.
\end{equation}
where $i'=N+1-i$.
We define a non-degenerate symmetric bilinear form on $V$ as follows:
\begin{equation}\label{20170424:form}
\langle v_i|v_j\rangle=-\epsilon_i\delta_{j,i'}
\,,
\qquad i,j=1,\dots,N\,.
\end{equation}
Let $A^\dagger$ denote the adjoint of $A\in\End V$ with respect to \eqref{20170424:form}.
As in Sections \ref{sec:example1.3}, \ref{sec:example1.4} and \ref{sec:example1.7},
in terms of elementary matrices, it is given by 
\begin{equation}\label{eq:victor}
(E_{ij})^\dagger=\epsilon_i\epsilon_jE_{j'i'}
\,.
\end{equation}

Let $\mf g=\{A\in\End V\mid A^\dagger=-A\}\simeq\mf{so}_N$.
For $1\leq i,j\leq N-1$ we let $F_{ij}$ be defined as in \eqref{eq:F}, and, as usual,
we also denote by $f_{ij}$ the same elements when viewed as elements 
of $\mf g\subset\mc V(\mf g)$.
It is immediate to check
that the same commutation relations as in \eqref{comm:BC} hold.
Recall that, by \eqref{eq:F}, a basis of $\mf g$ is
$\{f_{ij}\mid(i,j)\in I\}$
where
$$
I=\{(i,j)\,\mid\, 1\leq i\leq N, 1\leq j< i'\}
\,.$$
Its dual basis, with respect to the trace form \eqref{20170424:eq2}, is 
$$
\big\{\frac12f_{ji}\mid (i,j)\in I\big\}\,.
$$ 

Let $f=f_{N-1,1}\in\mf g$. It is associated to the partition $p=(2,2,1,\dots,1)$ thus it is 
a minimal nilpotent element.
We can include $f\in\mf g$ in the following $\mf{sl}_2$-triple $\{e,h=2x,f\}\subset\mf g$:
$$
e=f_{1,N-1}
\,,
\qquad
x=\frac12 f_{11}+\frac12 f_{22}
\,.
$$
Using the commutation relations \eqref{comm:BC} 
one checks that the non-zero components of the grading \eqref{eq:grading} are
\begin{align*}
&\mf g_{-1}=\mb Ff\,,
\qquad
\mf g_{-\frac12}=\Span\{f_{hk}\mid 3\leq h\leq N-2\,,k=1,2\}\,,
\\
&\mf g_{0}=\Span\{f_{11},f_{12},f_{21},f_{22},f_{hk}\mid 3\leq h\leq N-2,3\leq k<h'\}\,,
\\
&
\mf g_{\frac12}=\Span\{f_{hk}\mid h=1,2\,, 3\leq k\leq N-2\}\,,
\qquad
\mf g_{1}=\mb Fe
\,.
\end{align*}
Moreover, we have that $D=d=1$, and $(\End V)[1]=\Span\{E_{h,k}\mid 1\leq h\leq 2\,,N-1\leq k\leq N\}$.
Hence, we may choose $T=S=E_{1N-1}+E_{2,N}$, and we let $S=IJ$ be its canonical
decomposition \eqref{IJ}:
\begin{equation}\label{20140424:IJ}
I=\begin{pmatrix}
\id_2\\0_{N-2}\end{pmatrix}
\qquad
\text{and}
\qquad
J=\begin{pmatrix}
0_{N-2}&\id_2
\end{pmatrix}
\,.
\end{equation}
As a subspace $U\subset\mf{g}$ complementary fo $[f,\mf{g}]$ and compatible  with the grading \eqref{eq:grading} we choose
$U=\Span\{u^{ji}=\frac1{2}f_{ji}\mid (ij)\in I_f\}$, where
$$
I_f=\{(11),(12),(21),(N-1,1)\}\cup\{(hk)\mid 3\leq h\leq N-2,1\leq k\leq h'\}
\,.
$$
The dual basis $\{u_{ij}\mid (ij)\in I_f\}$ of $\mf g^f$ is given by
$$
u_{11}=f_{11}-f_{22}\,,
\qquad
u_{ij}=f_{ij}\,,
\quad(ij)\in I_f\setminus\{(11)\}
\,.
$$
Let $w:\mc V(\mf g^f)\to\mc W(\mf g,f)$ be the differential algebra isomorphism given in Theorem \ref{thm:structure-W}, and let us
denote by $w_{ij}=w(u_{ij})$, for $(ij)\in I_f$.
By Theorem \ref{thm:main} we have
\begin{equation}
\begin{split}
\label{eq:L1-minimal_BD}
&
L_\epsilon(\partial)
=|\id_N\partial+F+\frac12\sum_{(ij)\in I_f}w_{ij}F_{ji}+\epsilon S|_{IJ}
\\
&
=
\left|
\left(\begin{array}{ccccc}
\partial+\frac12w_{11} & \frac12 w_{21}&\frac12w_{+1}& \frac12w_{N-1,1}+\epsilon &0 \\
\frac12 w_{12}&\partial&\frac12w_{+2}&0&\frac{1}{2}w_{N-1,1}+\epsilon
\\
0_{N-4\times1}& 0_{N-4\times1}&\id_{N-4}\partial+\frac12W_{++}& \frac12\tilde w_{+2}& \frac12\tilde w_{+1} \\
1& 0_{1\times N-4}&0_{1\times N-4}&\partial&\frac12 w_{21}
\\
0&1&0&\frac12 w_{12}&\partial-\frac12w_{11}
\end{array}\right)
\right|_{IJ}
\\
&
=
\begin{pmatrix}
-(\partial+\frac12w_{11})\partial+\frac12w_{N-1,1}-\frac14w_{21}w_{12}
&
-\frac12\partial\circ w_{21}-\frac12w_{21}\partial
\\
-\frac12 w_{12}\partial-\frac12\partial\circ w_{12}
&
-\partial(\partial-\frac12w_{11})+\frac12w_{N-1,1}-\frac14w_{12}w_{21}
\end{pmatrix}
\\
&-\frac14
\begin{pmatrix}
w_{+1}\\w_{+2}
\end{pmatrix}
\left(\id_{N-4}\partial+\frac12W_{++}\right)^{-1}\circ
\begin{pmatrix}
\tilde w_{+2}&\tilde w_{+1}\end{pmatrix}
+\epsilon\id_2
\in\Mat_{2\times 2}\mc W(\mf g,f)((\partial^{-1}))\,,
\end{split}
\end{equation}
where ($k=1,2$)
\begin{equation}\label{eq:minimal_D-notation}
\begin{split}
& w_{+k}
=
\left(\begin{array}{llll} w_{3k} & w_{4k}& \dots & w_{N-2,k} \end{array}\right) 
\,,\\
& W_{++}
=
\sum_{h=3}^{N-2}\sum_{k=3}^{h'-1}w_{kh}\bar F_{hk}
\,\,,\,\,\,\,
\tilde w_{+k}
=(-1)^{k+1}
\left(\begin{array}{c} \epsilon_{N-2}w_{N-2,k} \\ \vdots \\ \epsilon_{4}w_{4k}\\ \epsilon_3w_{3k} \end{array}\right)
\,.
\end{split}
\end{equation}
In \eqref{eq:minimal_D-notation} $\bar F_{hk}$ denotes the matrix $F_{hk}$ 
where we have removed the first two and the last two rows and columns.

\subsection{Example 9: Lax operator $L_\epsilon(z)$ for 
a distinguished nilpotent in $\mf{so}_{4n}$}\label{sec:example1.9}

Let $V=\mb F^{4n}$ be a $4n$-dimensional vector space with basis $\{v_i\}_{i=1}^{4n}$,
and let $q=2n+1$. For $i\in\{1,\dots,4n\}$
we introduce the following notation
\begin{equation}\label{20170425:eq1}
i'=\left\{
\begin{array}{ll}
q+1-i\,,&1\leq i\leq q\,,
\\
4n+1+q-i\,,&q+1\leq i\leq 4n\,,
\end{array}
\right.
\quad
\epsilon_i=
\left\{
\begin{array}{ll}
(-1)^i\,,&1\leq i\leq q-1\,,
\\
(-1)^{i+1}\,,&q+1\leq i\leq 4n\,.
\end{array}
\right.
\end{equation}
We use the bilinear form on $V$ given by \eqref{20170424:form}.
The adjoint with respect to this form
is given by \eqref{eq:victor}.
Let $\mf g=\{A\in\End V\mid A^\dagger=-A\}\simeq\mf{so}_{4n}$.
For $1\leq i,j\leq 4n-1$ we let $F_{ij}$ be defined as in \eqref{eq:F}, and, as usual,
we also denote by $f_{ij}$ the same elements when viewed as elements 
of $\mf g\subset\mc V(\mf g)$).
It is immediate to check
that the same commutation relations as in \eqref{comm:BC} hold.
A basis of $\mf g$ is
$\{f_{ij}\mid (ij)\in I\}$ where
\begin{align*}
I=&\{(ij)\mid1\leq i\leq q-1,1\leq j<i'\}\cup
\{(ij)\mid q+1\leq i\leq 4n\,,1\leq j\leq q\}
\\
&\cup\{(ij)\mid q+1\leq i\leq q\,,q+1\leq j<i'\}
\,. 
\end{align*}
Its dual basis, with respect to the trace form \eqref{20170317:eq1}, is $\{\frac12f_{ji}\mid (i,j)\in I\}$. 

Let
$$
f=\sum_{k=1}^{n}f_{k+1,k}+\sum_{k=1}^{n-1}f_{q+k+1,q+k}\in\mf g
\,.
$$
It is a nilpotent element associated to the partition $p=(q,q-2)$ of $4n$. 
This element $f$ is distinguished in the sense that the reductive part 
of its centralizer is trivial.
We can include $f\in\mf g$ in the $\mf{sl}_2$-triple $\{e,h=2x,f\}\subset\mf g$, where
\begin{equation}
\begin{split}\label{20170411:sl2tripleD4n}
&x=\sum_{k=1}^n\frac{q+1-2k}2 f_{k,k}
+\sum_{k=1}^{n-1}\frac{q-1-2k}2 f_{q+k,q+k}\,,
\quad
\\
&e=
\sum_{k=1}^{n}
k(q-k)f_{k,k+1}
+\sum_{k=1}^{n-1}
k(q-2-k)f_{q+k,q+k+1}
\,.
\end{split}
\end{equation}
Using the commutation relations \eqref{comm:BC} it is immediate to check that:
$$
d=q-2
\,\,,\,\,\,\,
\mf g_d=\mb FF_{1,q-1}\oplus\mb FF_{1,4n}
\,\,,\,\,\,\,
D=q-1
\,\,,\,\,\,\,
(\End V)[D]=\mb FE_{1q}
\,.
$$
We thus choose $T=E_{1q}$, and $s=\frac{a}{2}f_{1,q-1}+bf_{1,4n}$, for arbitrary
$a,b\in\mb F$.

As a subspace $U\subset\mf{g}$ complementary fo $[f,\mf{g}]$ and compatible  with the grading \eqref{eq:grading} we choose
$U=\Span\{u^{+,i},1\leq i\leq n\,,u^{0,i},1\leq i\leq q-2,u^{-,i},1\leq i\leq n-1\}$, where
\begin{equation}\label{eq:U_D4n}
u^{+,i}=\frac12f_{1,q+1-2i}
\,,
\quad
u^{0,i}=\frac12f_{1,q+i}
\,,\quad
u^{-,i}=\frac12f_{q+1,4n+1-2i}
\,.
\end{equation}
The dual basis of $\mf g^f$ is given by
\begin{align*}
&u_{+,i}=\sum_{k=1}^if_{q+k-2i,k}
\,,
\quad
1\leq i\leq n
\,,
\\
&u_{0,i}=\sum_{k=1}^{q-1-i}f_{q-1+k+i,k}
\,,
\quad
1\leq k\leq q-2
\,,
\\
&u_{-,i}=\sum_{k=1}^if_{4n+k-2i,q+k}
\,,
\quad
1\leq i\leq n-1
\,.
\end{align*}
Let $w:\mc V(\mf g^f)\to\mc W(\mf g,f)$ be the differential algebra isomorphism given in Theorem \ref{thm:structure-W}, and let 
$$
w_{+,i}=w(u_{+,i})
\,,
\qquad
w_{0,i}=w(u_{0,i})
\,,
\qquad
w_{-,i}=w(u_{-,i})
\,.
$$
By Theorem \ref{thm:main} and \cite[Eq.(2.14)]{DSKV16b} (see also \cite[Prop.4.2]{DSKVnew}),
we have
\begin{equation}\label{eq:distinguished_quasidet}
L_\epsilon(\partial)
=|\id_N\partial+F+\frac12\sum_{i=1}^{n}w_{+,i} U^{+,i}
+\frac12\sum_{i=1}^{q-2}w_{0,i} U^{0,i}
+\frac12\sum_{i=1}^{n-1}w_{-,i} U^{-,i}+\epsilon S|_{1q}
\,.
\end{equation}
The quasideterminant \eqref{eq:distinguished_quasidet} is computed using the usual formula, see \cite[Prop.4.2]{DSKVnew}, and
the result is
\begin{equation}\label{L_distinguished}
L_\epsilon(\partial)=W_1(\partial)-W_{2}(\partial)W_{4}(\partial)^{-1}W_3(\partial)
\,,
\end{equation}
where
\begin{align*}
W_{1}(\partial)
=&\partial^{2n+1}-\frac12\sum_{k=1}^n\left(
\partial^{2n+1-2k}\circ w_{+,n+1-k}+w_{+,n+1-k}\partial^{2n+1-2k}
\right)
\\
&+\frac14\sum_{k=1}^{n-1}\sum_{h=1}^{n-k}w_{+,n+1-h}\partial^{2n+1-2(h+k)}\circ w_{+,n+1-k}
\\
&
+\frac14\sum_{k=1}^{2n-2}\sum_{h=1}^{2n-1-k}(-1)^hw_{0,h}\partial^{2n-1-(h+k)}\circ w_{0,k}
-\epsilon a\partial
\,,
\\
W_2(\partial)=&\frac14
\sum_{j=0}^{n-1}\sum_{i=1}^{2n-1-2j}w_{0,i}(-\partial)^{2n-1-i-2j}\circ w_{-,n-j}
-\epsilon b\,,
\\
W_3(\partial)=&
W_{2}(\partial)^*
\,,
\\
W_{4}(\partial)=&
\partial^{2n-1}-\frac12\sum_{k=1}^{n-1}\left(
\partial^{2n-1-2k}\circ w_{-,n-k}+w_{-,n-k}\partial^{2n-1-2k}
\right)
\\
&+\frac14\sum_{k=1}^{n-2}\sum_{h=1}^{n-1-k}w_{-,n-h}\partial^{2n-1-2(h+k)}\circ w_{-,n-k}
\,.
\end{align*}
In the above formulas we set $w_{-,n}=-2$. Note that $L_\epsilon(\partial)^*=-L_\epsilon(\partial)$.
As noted in Remark \ref{rem:OK},
by comparing equation \eqref{L_distinguished} with the quasideterminant \eqref{eq:L} one gets an explicit formula for the generators
of the $\mc W$-algebra $\mc W(\mf g,f)$.

\section{Examples of integrable hierarchies for generalized Adler type operators}
\label{sec:10}

\subsection{$L(\mf{sl}_2)=L(\mf{sp}_2)$}

Let $f\in\mf {sl}_2\simeq\mf {sp}_2$ be a principal (= non-zero) nilpotent element.
Recall from Section \ref{sec:2} that, as a differential algebra,
$\mc W(\mf{sl}_2,f)=\mc W(\mf{sp}_2,f)=\mb C[w^{(n)}\mid n\in\mb Z_+]$. By equation
\eqref{eq:slN_quasidet} (or \eqref{eq:BC_L}), we have
$$
L(\mf{sl}_2,f)=L(\mf{sp}_2,f)=-\partial^2+w\in\mc W(\mf{sl}_2,f)[\partial]
\,.
$$
Consider the operator $L(\partial)=-L(\mf{sl}_2,f)=-L(\mf{sp}_2,f)=\partial^2-w$.
Its square root is
\begin{equation}\label{sl2_root}
B(\partial)=\partial-\frac w2\partial^{-1}+\frac{w'}4\partial^{-2}-\frac18(w^2+w'')\partial^{-3}+o(\partial^{-3})
\in\mc W(\mf{sl}_2,f)((\partial^{-1}))
\,.
\end{equation}
By a straightforward computation we get
\begin{equation}
\begin{split}\label{sl2_rootb}
B(\partial)^3&=\partial^3-\frac32w\partial-\frac34w'
+\frac18(3w^2-w'')\partial^{-1}+o(\partial^{-1})
\,.
\end{split}
\end{equation}
Using equations \eqref{sl2_root}, \eqref{sl2_rootb} and \eqref{eq:hn} we get
\begin{equation}\label{eq:integrals_kdv}
\tint h_{1,B}=\tint w\,,
\qquad
\tint h_{2,B}=0,
\qquad
\tint h_{3,B}=-\frac14\tint w^2\,,
\end{equation}
and the corresponding non-zero Hamiltonian equations \eqref{eq:list1} are
$$
\frac{dw}{dt_{1,B}}=w'\,,
\qquad
\frac{dw}{dt_{3,B}}=\frac14(w'''-6ww')
\,.
$$
The second equation is the Korteweg-de Vries equation. 

\begin{remark}\label{rem:modified_kdv}
Let $\mf {sl}_2=\Span\{e=e_{21},h=2x=e_{11}-e_{22},f=e_{21}\}$ in its standard representation,
then in the $\ad x$-eigenspace decomposition \eqref{eq:grading} we have $\mf g_0=\mb F x$.
By Remark \ref{rem:modified_sln} the Miura map $\mu:\mc W(\mf{sl}_2,f)\to\mc V(\mf g_0)=\mb F[x^{(n)}\mid n\in\mb Z_+]$ is given by
$$
\mu(L(\partial))=\partial^2-\mu(w)=(\partial+x)(\partial-x)
\,.
$$
Hence we get $\mu(w)=x^2+x'$, which is the famous Miura transformation \cite{Miu68}. Applying the Miura map to the
integrals of motion given by equation \eqref{eq:integrals_kdv} we get
$$
\tint \bar h_{1,B}=\tint x^2+x'\,,
\qquad
\tint \bar h_{3,B}=-\frac14\tint x^4-xx''
\,.
$$
The corresponding Hamiltonian equations are
$$
\frac{dx}{dt_{1,B}}=x'\,,
\qquad
\frac{dx}{dt_{3,B}}=\frac14(x'''-6x^2x')
\,.
$$
The second equation is the modified Korteweg-de Vries equation. 
\end{remark}

\subsection{$L(\mf{sp}_2)\partial^{\pm1}$}
Consider the operator 
$L(\partial)=-L(\mf{sp}_2)\partial=\partial^3-w\partial\in\mc W(\mf{sp}_2,f)$.
Its cube root is
\begin{equation}
\begin{split}\label{sp2_root1}
&B(\partial)=\partial-\frac w3\partial^{-1}+\frac{w'}3\partial^{-2}-\frac19(w^2+2w'')\partial^{-3}
+\frac19(4ww'+w''')\partial^{-4}\\
&-\frac1{81}\left(5w^3+45 (w')^2+45ww''+3w^{(4)}\right)\partial^{-5}+o(\partial^{-5})
\in\mc W(\mf{sp}_2,f)((\partial^{-1}))
\,.
\end{split}
\end{equation}
Then, we have
\begin{align}
\begin{split}\label{sp2_root2}
&B(\partial)^2=\partial^2-\frac23w+\frac{w'}3\partial^{-1}-\frac19(w^2+w'')\partial^{-2}
+\frac{ww'}{3}\partial^{-3}
\\
&-\frac1{81}\left(4w^3+27(w')^2+24ww''-3w^{(4)}\right)\partial^{-4}+o(\partial^{-4})\,,
\\
&B(\partial)^4=\partial^4-\frac43w\partial^2-\frac23w'\partial+\frac29(w^2-w'')
+\frac19(w'''-2ww')\partial^{-1}
\\
&+\frac1{81}\left(4w^3+9(w')^2-3w^{(4)}\right)\partial^{-2}+o(\partial^{-2})
\,,
\\
&B(\partial)^5=\partial^5-\frac53w\partial^{3}-\frac53w'\partial^{2}+\frac59(w^2-2w'')\partial
\\
&+\frac1{81}\left(5w^3-15ww''+3w^{(4)}\right)\partial^{-1}+o(\partial^{-1})
\,,
\end{split}
\end{align}
and, by equations \eqref{sp2_root1}, \eqref{sp2_root2} and \eqref{eq:hn}, we get
\begin{equation}\label{eq:integrals_sk}
\tint h_{1,B}=\tint w\,,
\quad
\tint h_{2,B}=\tint h_{3,B}=\tint h_{4,B}=0,
\quad
\tint h_{5,B}=-\frac1{27}\tint w^3-3ww''\,.
\end{equation}
Note that $B^*$ is a third root of $-\partial L\partial^{-1}$.
It follows from \cite[Prop.1.1(a)]{AGD}
that $B=-\partial^{-1}B^*\partial$. Thus we have
\begin{equation}\label{20170511:eq1}
B^n=(-1)^n\partial^{-1}(B^n)^*\partial\,,
\quad n\in\mb Z\,.
\end{equation}
Using equation \eqref{20170511:eq1}, the Hamiltonian equations \eqref{eq:list3} become
\begin{equation}\label{eq:list3bis}
\frac{dL(z)}{dt_{n,B}}=\frac{1-(-1)^n}2[(B^n)_+,L](z)
\,,
\quad n\in\mb Z\,.
\end{equation}
Using \eqref{sp2_root1} and \eqref{sp2_root2} we get that the first 
non-zero Hamiltonian equations \eqref{eq:list3bis} are
$$
\frac{dw}{dt_{1,B}}=w'\,,
\qquad
\frac{dw}{dt_{5,B}}=-\frac19(w^{(5)}-5w'w''-5ww'''+5w^2w')
\,.
$$
The second equation is the Sawada-Kotera equation \cite{SK74}.
\begin{remark}\label{rem:modified_sk}
Since $\mf {sp}_2\simeq\mf{sl}_2$, the Miura map
$\mu:\mc W(\mf{sp}_2,f)\to\mc V(\mf g_0)=\mb F[x^{(n)}\mid n\in\mb Z_+]$ is given by
$\mu(w)=x^2+x'$ (see Remark \ref{rem:modified_sk}).
Applying the Miura map to the
integrals of motion given by equation \eqref{eq:integrals_sk} we get
$$
\tint \bar h_{1,B}=\tint x^2+x'\,,
\qquad
\tint \bar h_{5,B}=-\frac1{27}\tint x^6+15x^2(x')^2-5(x')^3+3(x'')^2
\,.
$$
The corresponding integrable Hamiltonian equations are
\begin{equation}\label{eq:modified_sk}
\frac{dx}{dt_{1,B}}=x'\,,
\quad
\frac{dx}{dt_{5,B}}=-\frac19(x^{(5)}+5x'x'''-5x^2x'''+5(x'')^2-20xx'x''-5(x')^3+5x^4x')
\,.
\end{equation}
The second equation is the modified Sawada-Kotera equation.
\end{remark}
Consider now the operator
\begin{equation}\label{sp2_root3}
L(\partial)=-L(\mf{sp}_2)\partial^{-1}=\partial-w\partial^{-1}\in\mc W(\mf{sp}_2,f)((\partial^{-1}))
\,.
\end{equation}
We have
\begin{equation}
\begin{split}\label{sp2_root4}
&
L^2(\partial)=\partial^2-w-\partial\circ w\partial^{-1}+w\partial^{-1}\circ w\partial^{-1}
\,,
\\
&
L^3(\partial)=\partial^3-w\partial-\partial\circ w-\partial^2\circ w\partial^{-1}
+w^2\partial^{-1}+w\partial^{-1}\circ w
\\
&+\partial\circ w\partial^{-1}\circ w\partial^{-1}
-w\partial^{-1}\circ w\partial^{-1}\circ w\partial^{-1}
\,.
\end{split}
\end{equation}
Using equations \eqref{sp2_root3}, \eqref{sp2_root4} and \eqref{eq:hn} we get
$$
\tint h_{1,L}=\tint w\,,
\qquad
\tint h_{2,L}=0\,,
\qquad
\tint h_{3,L}=-\tint w^2\,.
$$
Note that, also in this case, equation \eqref{eq:list3} reduces to \eqref{eq:list3bis} (where
we should replace $B$ by $L$). Hence, using \eqref{sp2_root3} and \eqref{sp2_root4}
the first non-zero Hamiltonian equations are
$$
\frac{dw}{dt_{1,L}}=w'\,,
\qquad
\frac{dw}{dt_{3,L}}=w'''-6ww'
\,.
$$
So we get again the Korteweg-de Vries equation. 

\subsection{$L(\mf{s0}_3)$}
Let $f\in\mf {so}_3$ be the principal nilpotent element.
Recall from Section \ref{sec:2} that, as a differential algebra,
$\mc W(\mf{so}_3,f)=\mb C[w^{(n)}\mid n\in\mb Z_+]$. By equation
\eqref{eq:BC_L}, we have
$$
L(\mf{so}_3,f)=\partial^3-w\partial-\frac{w'}2\in\mc W(\mf{so}_3,f)[\partial]
\,.
$$
Consider the cube root of $L(\mf{so}_3,f)$:
\begin{equation}
\begin{split}\label{so3_root1}
&B(\partial)=\partial-\frac w3\partial^{-1}+\frac{w'}6\partial^{-2}-\frac1{18}(2w^2+w'')\partial^{-3}
+\frac{ww'}3\partial^{-4}\\
&-\left(\frac5{81}w^3+\frac{11}{36}(w')^2+\frac13ww''-\frac1{54}w^{(4)}
\right)\partial^{-5}+o(\partial^{-5})
\in\mc W(\mf{so}_3,f)((\partial^{-1}))
\,.
\end{split}
\end{equation}
Then, we have
\begin{align}
\begin{split}\label{so3_root2}
&B(\partial)^2=\partial^2-\frac23w-\frac1{18}(2w^2-w'')\partial^{-2}
-\frac1{18}(w'''-4ww')\partial^{-3}
\\
&-\left(\frac4{81}w^3+\frac{5}{36}(w')^2+\frac{7}{54}ww''-\frac{1}{27}w^{(4)}
\right)\partial^{-4}
+o(\partial^{-4})\,,
\\
&B(\partial)^4=\partial^4-\frac23w\partial^2-\frac23\partial^2\circ w+\frac19(2w^2+w'')
\\
&+\left(\frac{4}{81}w^3-\frac{1}{18}(w')^2-\frac{1}{9}ww''+\frac{1}{54}w^{(4)}
\right)\partial^{-2}+o(\partial^{-2})
\,,
\\
&B(\partial)^5=\partial^5-\frac56w\partial^{3}-\frac56\partial^3\circ w
+\frac5{18}(w^2+w'')\partial+\frac{5}{18}\partial\circ(w^2+w'')
\\
&+\left(\frac5{81}w^3-\frac{5}{36}(w')^2-\frac5{27}ww''+\frac{1}{27}w^{(4)}\right)\partial^{-1}+o(\partial^{-1})
\,,
\end{split}
\end{align}
and, by equations \eqref{so3_root1}, \eqref{so3_root2} and \eqref{eq:hn}, we get
\begin{equation}\label{eq:integrals_kk}
\tint h_{1,B}=\tint w\,,
\quad
\tint h_{2,B}=\tint h_{3,B}=\tint h_{4,B}=0,
\quad
\tint h_{5,B}=-\int \frac{w^3}{27}-\frac{ww''}{36}\,.
\end{equation}
Note that $B^*$ is a third root of $-L$. Hence, it follows from \cite[Prop.1.1(a)]{AGD}
that $B^*=-B$, and so we have $(B^n)^*=(-1)^nB^n$.
Using this fact and equations \eqref{so3_root1} and \eqref{so3_root2} we get that the first 
non-zero Hamiltonian equations \eqref{eq:list2} are:
$$
\frac{dw}{dt_{1,B}}=w'\,,
\qquad
\frac{dw}{dt_{5,B}}=-\frac1{9}(w^{(5)}-\frac{25}2w'w''-5ww'''+5w^2w')
\,.
$$
The second equation is the Kaup-Kupershmidt equation \cite{Kau80}.
\begin{remark}\label{rem:modified_kk}
Let $\mf{so}_3=\Span\{e=2(e_{12}+e_{23}),h=2x=2(e_{11}-e_{33}),f=e_{21}+e_{32}\}$
in its standard representation.
By Remark \ref{rem:modified_BC} the Miura map $\mu:\mc W(\mf{so}_3,f)\to\mc V(\mf g_0)=\mb F[x^{(n)}\mid n\in\mb Z_+]$ is given by
$$
\mu(L(\partial))=\partial^3-\mu(w)\partial-\frac12\mu(w)'
=(\partial+\frac12 x)\partial(\partial-\frac12x)
\,.
$$
Hence, we get $\mu(w)=\frac14x^2+x'$.
Applying the Miura map to the
integrals of motion given by equation \eqref{eq:integrals_kk} we get
$$
\tint \bar h_{1,B}=\tint \frac14x^2+x'\,,
\qquad
\tint \bar h_{5,B}=-\frac1{1728}\tint x^6+60x^2(x')^2+40(x')^3+48(x'')^2
\,.
$$
The corresponding Hamiltonian equations are
\begin{equation}\label{eq:modified_kk}
\frac{dx}{dt_{1,B}}\!=x'\,,\,\,
\frac{dx}{dt_{5,B}}\!=-\frac19(x^{(5)}-\frac52x'x'''-\frac54x^2x'''-\frac52(x'')^2-5xx'x''-\frac54(x')^3+\frac5{16}x^4x')
\,.
\end{equation}
Note that by rescaling the variable $x$ by a factor $-\frac12$ the above equations are the same as equations \eqref{eq:modified_sk}.
\end{remark}
\begin{remark}
In \cite{DS85}, Drinfeld-Sokolov hierarchies are constructed for any affine Kac-Moody algebra $\mf g$ and a node of its Dynkin diagram.
The Sawada-Kotera equation corresponds to the pair $(A_2^{(2)},c_0)$, while the Kaup-Kupershmidt equation
corresponds to the pair $(A_2^{(2)},c_1)$. On the other hand, modified Drinfeld-Sokolov hierarchies do not depend on the choice of
a node in the Dynkin diagram of $\mf g$. This is the reason why the modified equations \eqref{eq:modified_sk} and \eqref{eq:modified_kk} coincide up to
a rescaling.
\end{remark}

\subsection{$L(\mf{s0}_3)\partial^{\pm1}$}

Consider the operator
$$
L=L(\mf{so}_3,f)\partial^{-1}
=\partial^2-w-\frac{w'}{2}\partial^{-1}\in\mc W(\mf{so}_3,f)((\partial^{-1}))
\,,
$$
and its square root
\begin{equation}\label{so3_root3}
B(\partial)=\partial-\frac w2\partial^{-1}-\frac{w^2}8\partial^{-3}+o(\partial^{-3})
\in\mc W(\mf{so}_3,f)((\partial^{-1}))
\,.
\end{equation}
By a straightforward computation we get
\begin{equation}
\begin{split}\label{so3_root4}
B(\partial)^3&=\partial^3-\frac32w\partial-\frac32w'
+\frac18(3w^2-4w'')\partial^{-1}+o(\partial^{-1})
\,.
\end{split}
\end{equation}
Using equations \eqref{so3_root3}, \eqref{so3_root4} and \eqref{eq:hn} we get
$$
\tint h_{1,B}=\tint w\,,
\qquad
\tint h_{2,B}=0,
\qquad
\tint h_{3,B}=-\frac14\tint w^2\,,
$$
and the corresponding non-zero Hamiltonian equations \eqref{eq:list3} are
$$
\frac{dw}{dt_{1,B}}=w'\,,
\qquad
\frac{dw}{dt_{3,B}}=w'''-\frac{3}{2}ww'
\,.
$$
Also in this case we get the Korteweg-de Vries equation. 

Finally, consider the operator
$$
L=L(\mf{so}_3,f)\partial
=\partial^4-w\partial^2-\frac{w'}{2}\partial\in\mc W(\mf{so}_3,f)[\partial^{-1}]
\,,
$$
and its fourth root
\begin{equation}
\begin{split}\label{so3_root5}
&B(\partial)=\partial-\frac w4\partial^{-1}+\frac{w'}{4}\partial^{-2}-\frac1{32}(3w^2+4w'')\partial^{-3}+o(\partial^{-3})
\in\mc W(\mf{so}_3,f)((\partial^{-1}))
\,.
\end{split}
\end{equation}
Then, we have
\begin{align}
\begin{split}\label{so3_root6}
&B(\partial)^2=\partial^2-\frac w2+\frac{w'}4\partial^{-1}-\frac{w^2}{8}\partial^{-2}
+o(\partial^{-2})\,,
\\
&B(\partial)^3=\partial^3-\frac34w\partial-\frac{1}{32}(3w^2-4w'')\partial^{-1}+o(\partial^{-1})
\,.
\end{split}
\end{align}
Using equations \eqref{so3_root5}, \eqref{so3_root6} and \eqref{eq:hn} we get
$$
\tint h_{1,B}=\tint w\,,
\qquad
\tint h_{2,B}=0,
\qquad
\tint h_{3,B}=\frac18\tint w^2\,,
$$
and the corresponding non-zero Hamiltonian equations \eqref{eq:list3} are
$$
\frac{dw}{dt_{1,B}}=w'\,,
\qquad
\frac{dw}{dt_{3,B}}=-\frac12\left(w'''-\frac{3}{2}ww'\right)
\,.
$$
Thus we obtain again the Korteweg-de Vries equation. 

\subsection{$L(\mf{sl}_N,f_{min})$}
Let $f\in\mf{sl}_N$ be a minimal nilpotent element. 
The generalized Adler type operator for
$\mc W(\mf{sl}_N,f)$ has been computed in Section \ref{sec:example1.6}. It is
\begin{equation}\label{eq:minimal_sln}
L(\mf{sl}_N,f)
=-\partial^2-w_{11}\partial+w_{N1}
- w_{+1}(\id_{N-2}\partial+W_{++})^{-1}\circ w_{N+}
\in\mc W(\mf{sl}_N,f)((\partial^{-1}))
\,,
\end{equation}
where
\begin{align*}
\begin{split}
&
-w_{11}=w_{22}+w_{33}+\dots+w_{N-1,N-1}\,,
\quad
w_{+1}
=
\left(\begin{array}{lll} w_{21} & \dots & w_{N-1,1} \end{array}\right) 
\,,\\
& W_{++}
=
\left(\begin{array}{ccc} 
w_{22} & \dots & w_{N-1,2} \\
\vdots & \ddots & \vdots \\
w_{2N-1} & \dots & w_{N-1,N-1} 
\end{array}\right)
\,\,,\,\,\,\,
w_{N+}
=
\left(\begin{array}{l} w_{N2} \\ \vdots \\ w_{N,N-1} \end{array}\right)
\,.
\end{split}
\end{align*}
Consider the operator $L=-L(\mf{sl}_N,f)$, and its square root
\begin{equation}\label{sln_minimal_root1}
B(\partial)=\partial+\frac{w_{11}}2
-\left(\frac{w_{N1}}{2}+\frac{w_{11}^2}{8}+\frac{w_{11}'}{4}\right)\partial^{-1}
+o(\partial^{-1})
\in\mc W(\mf{sl}_N,f)((\partial^{-1}))
\,.
\end{equation}
From equations \eqref{eq:minimal_sln}, \eqref{sln_minimal_root1} and \eqref{eq:hn}, we get
$$
\tint h_{1,B}=\int w_{N1}+\frac{w_{11}^2}{4}
\qquad\text{and}\qquad
\tint h_{2,B}=-\tint w_{+1}w_{N+}
\,.
$$
Using the following commutation relations for pseudodifferential operators 
\begin{align*}
&[\partial,(\partial+a)^{-1}]=(\partial+a)^{-1}\circ a-a(\partial+a)^{-1}
\,,
\\
&[\partial^2,(\partial+a)^{-1}]=
(a^2-a')(\partial+a)^{-1}-(\partial+a)^{-1}\circ (a^2+a')
\,,
\end{align*}
which hold for any differential polynomial $a$, and equations \eqref{eq:minimal_sln},
\eqref{sln_minimal_root1} it is immediate
to compute the corresponding Hamiltonian equations \eqref{eq:list1}. As a result we get:
\begin{align*}
&
\frac{d w_{N1}}{dt_{1,B}}=w_{N1}'+\frac12 w_{11}w_{11}'+\frac12 w_{11}''\,,
&&
\frac{d W_{++}}{dt_{1,B}}=0\,,
\\
&\frac{d w_{+1}}{dt_{1,B}}=w_{+1}'-w_{+1}\left(W_{++}-\frac{w_{11}}2\right)
\,,
&&
\frac{d w_{N+}}{dt_{1,B}}=w_{N+}'+\left(W_{++}-\frac{w_{11}}2\right)w_{N+}
\,,
\end{align*}
and
\begin{align*}
\frac{d w_{N1}}{dt_{2,B}}
&
=-2(w_{+1}w_{N+})'
\,,
\qquad
\frac{d W_{++}}{dt_{2,B}}=0\,,
\\
\frac{d w_{+1}}{dt_{2,B}}
&=
w_{+1}''-2w_{+1}'\left(W_{++}-\frac12 w_{11}\right)-w_{+1}W_{++}'
\\
&+w_{+1}W_{++}(W_{++}-w_{11})-w_{+1}w_{N1}
\,,
\\
\frac{d w_{N+}}{dt_{2,B}}
&
=
-w_{N+}''-2\left(W_{++}-\frac12 w_{11}\right)w_{N+}'-(W_{++}'-w_{11}')w_{N+}
\\
&-(W_{++}-w_{11})W_{++}w_{N+}+w_{N1}w_{N+}
\,.
\end{align*}
The above equations agree with the results in \cite[Sec. 6, arxiv version]{DSKV14a}.
\begin{remark}
As explained in \cite[Cr.5.5]{DSKV16b},
the variables $w_{ij}$, for $2\leq i,j\leq N-1$, do not evolve in time. By applying a
Dirac reduction procedure, see \cite{DSKV14b}, we get the following Dirac reduced 
Hamiltonian equations:
$$
\frac{d w_{N1}}{dt_{1,B}}=w_{N1}'\,,
\qquad
\frac{d w_{+1}}{dt_{1,B}}=w_{+1}'
\,,
\qquad
\frac{d w_{N+}}{dt_{1,B}}=w_{N+}'
\,,
$$
and
$$
\frac{d w_{N1}}{dt_{2,B}}
=-2(w_{+1}w_{N+})'
\,,
\quad
\frac{d w_{+1}}{dt_{2,B}}=
w_{+1}''-w_{+1}w_{N1}
\,,
\quad
\frac{d w_{N+}}{dt_{2,B}}
=
-w_{N+}''
+w_{N1}w_{N+}
\,,
$$
which is the multicomponent Yajima-Oikawa equation (it is the Yajima-Oikawa
equation \cite{YO76} for $N=3$, see also \cite{DSKV15-cor} and the references therein).
\end{remark}

\subsection{$L(\mf{sp}_N,f_{min})$}
Let $n\geq1$, $N=2n$ and $f\in\mf{sp}_N$ be a minimal nilpotent element. The generalized Adler operator for
$\mc W(\mf{sp}_N,f)$ has been computed in Section \ref{sec:example1.7}. It is
\begin{equation}\label{eq:minimal_spN}
L(\mf{sp}_N,f)
=-\partial^2+w_{N1}
- \frac 14w_{+1}(\id_{N-2}\partial+\frac12W_{++})^{-1}\circ \tilde w_{+1}
\in\mc W(\mf{sp}_N,f)((\partial^{-1}))
\,,
\end{equation}
where $w_{+1}$, $W_{++}$ and $\tilde w_{+1}$ are defined in equation \eqref{eq:minimal_C-notation}.
Consider
$$
L=-L(\mf{sp}_N,f)
=\partial^2-w_{N1}-\frac18(w_{+1}W_{++}\tilde w_{+1}+2w_{+1}\tilde w_{+1}')\partial^{-2}+o(\partial^{-2})
$$
and its square root
\begin{equation}
\begin{split}\label{spN_minimal_root1}
&B(\partial)=\partial-\frac{w_{N1}}2\partial^{-1}+\frac{w_{N1}'}{4}\partial^{-2}
\\
&-\frac{1}{16}(2w_{N1}^2+w_{+1}W_{++}\tilde w_{+1}+2w_{+1}\tilde w_{+1}'+2w_{N1}'')\partial^{-3}
+o(\partial^{-3})
\,.
\end{split}
\end{equation}
Then, 
\begin{equation}
\begin{split}\label{spN_minimal_root2}
&B(\partial)^3=\partial^3-\frac32w_{N1}\partial-\frac34w_{N1}'
\\
&+\frac{1}{16}(6w_{N1}^2-3w_{+1}W_{++}\tilde w_{+1}-6w_{+1}\tilde w_{+1}'-2w_{N1}'')\partial^{-1}
+o(\partial^{-1})
\,.
\end{split}
\end{equation}
From equations \eqref{spN_minimal_root1}, \eqref{eq:minimal_spN}, \eqref{spN_minimal_root2} and \eqref{eq:hn}, we get
$$
\tint h_{1,B}=\tint w_{N1}\,,
\quad
\tint h_{2,B}=0\,,
\quad
\tint h_{3,B}=-\frac18\int 2w_{N1}^2-w_{+1}W_{++}\tilde w _{+1}-2w_{+1}\tilde w_{+1}'
\,.
$$
The above integrals of motion agree with the ones obtained for the generalized 
Drinfeld- Sokolov hierarchy constructed in \cite{DSKV14a}.
Hence,
the corresponding Hamiltonian equations \eqref{eq:list2} have already appeared there.

\subsection{$L(\mf{sp}_N,f_{min})\partial^{\pm1}$}

Consider the operator
\begin{equation}\label{20170514:eq1}
L=-L(\mf{sp}_N,f)\partial
=\partial^3-w_{N1}\partial-\frac18(w_{+1}W_{++}\tilde w_{+1}+2w_{+1}\tilde w_{+1}')\partial^{-1}+o(\partial^{-1})
\end{equation}
and its cube root
\begin{equation}
\begin{split}\label{spN_minimal_root3}
&B(\partial)=\partial-\frac{w_{N1}}3\partial^{-1}+\frac{w_{N1}'}{3}\partial^{-2}
\\
&-(\frac{w_{N1}^2}9+\frac1{24}w_{+1}W_{++}\tilde w_{+1}+\frac1{12}w_{+1}\tilde w_{+1}'+\frac29w_{N1}'')\partial^{-3}
+o(\partial^{-3})
\,.
\end{split}
\end{equation}
Then\begin{equation}
\begin{split}\label{spN_minimal_root4}
&B(\partial)^2=\partial^2-\frac23w_{N1}+\frac{w_{N1}'}3\partial^{-1}
\\
&-(\frac{w_{N1}^2}9+\frac1{12}w_{+1}W_{++}\tilde w_{+1}+\frac16w_{+1}\tilde w_{+1}'+\frac{w_{N1}''}9)\partial^{-2}
+o(\partial^{-2})
\,.
\end{split}
\end{equation}
%
From equations \eqref{spN_minimal_root3}, \eqref{spN_minimal_root4}, \eqref{20170514:eq1} and \eqref{eq:hn}, we get
$$
\tint h_{1,B}=\tint w_{N1}\,,
\quad
\tint h_{2,B}=0\,,
\quad
\tint h_{3,B}=\frac18\int w_{+1}W_{++}\tilde w _{+1}+2w_{+1}\tilde w_{+1}'
\,.
$$
Using equations \eqref{20170514:eq1}, \eqref{spN_minimal_root3} and \eqref{spN_minimal_root4} it is straightforward to compute
the corresponding Hamiltonian equations \eqref{eq:list3}. 
We omit the details of this computation.
%
By applying the Dirac reduction 
by the variables $w_{ij}$, for $2\leq i\leq N-1$ and $2\leq j\leq i'$, 
we get the following first non-trivial Hamiltonian equations
$$
\frac{d w_{N1}}{dt_{3,B}}
=\frac34w_{+1}\tilde w_{+1}''
\,,
\qquad
\frac{d w_{+1}}{dt_{3,B}}=
w_{+1}'''-w_{N1}w_{+1}'
\,.
$$

Finally, consider the operator
\begin{equation}\label{20170514:eq2}
L=-L(\mf{sp}_N,f)\partial^{-1}
=\partial-w_{N1}\partial^{-1}-\frac18(w_{+1}W_{++}\tilde w_{+1}+2w_{+1}\tilde w_{+1}')\partial^{-3}+o(\partial^{-3})
\,.
\end{equation}
Then
\begin{equation}
\begin{split}\label{spN_minimal_root5}
L(\partial)^2&=\partial^2-2w_{N1}-w_{N1}'\partial^{-1}
\\
&+(w_{N1}^2-\frac14w_{+1}W_{++}\tilde w_{+1}-\frac12w_{+1}\tilde w_{+1}')\partial^{-2}
+o(\partial^{-2})
\,,
\\
L(\partial)^3&=\partial^3-3w_{N1}\partial-3w_{N1}'
\\
&+(3w_{N1}^2-\frac38w_{+1}W_{++}\tilde w_{+1}-\frac34w_{+1}\tilde w_{+1}'-w_{N1}'')\partial^{-1}
+o(\partial^{-1})
\,.
\end{split}
\end{equation}
From equations \eqref{20170514:eq2}, \eqref{spN_minimal_root5} and \eqref{eq:hn}, we get
$$
\tint h_{1,L}=\tint w_{N1}\,,
\quad
\tint h_{2,L}=0\,,
\quad
\tint h_{3,L}=-\int w_{N1}^2-\frac18w_{+1}W_{++}\tilde w _{+1}-\frac14w_{+1}\tilde w_{+1}'
\,.
$$
Using equations \eqref{20170514:eq2} and \eqref{spN_minimal_root5} it is straightforward to compute
the corresponding Hamiltonian equations \eqref{eq:list3}. Again, we omit the details of this computation.
%
By applying the Dirac reduction 
by the same variables as above
we get the following first non-trivial Hamiltonian equations
$$
\frac{d w_{N1}}{dt_{3,L}}
=w_{N1}'''-6w_{N1}w_{N1}'+\frac34w_{+1}\tilde w_{+1}''
\,,
\qquad
\frac{d w_{+1}}{dt_{3,L}}=
w_{+1}'''-3(w_{N1}w_{+1})'
\,.
$$


\begin{thebibliography}{00} 

\bibitem[Adl79]{Adl79}
Adler M.,
\emph{On a trace functional for formal pseudodifferential operators and the symplectic 
structure of the Korteweg-de Vries equation},
Invent. Math. {\bf50} (1979), 219-248.

\bibitem[BDSK09]{BDSK09}
Barakat A., De Sole A., Kac V. G.,
\emph{Poisson vertex algebras in the theory of Hamiltonian equations}, 
Japan. J. Math. {\bf 4} (2009), n.2, 141-252.

\bibitem[BD82]{BD82}
Belavin A.A., Drinfeld V.G.,
\emph{Solutions of the classical Yang-Baxter equation for simple Lie algebras},
Funktsional. Anal. i Prilozhen. {\bf16} (1982), n. 3, 1-29.

\bibitem[BdGHM93]{BdGHM93}
Burruoughs N., de Groot M., Hollowood T., Miramontes L.,
\emph{Generalized Drinfeld-Sokolov hierarchies II: the Hamiltonian structures},
Comm. Math. Phys. {\bf 153} (1993), 187-215.

\bibitem[dGHM92]{dGHM92}
de Groot M., Hollowood T., Miramontes L.,
\emph{Generalized Drinfeld-Sokolov hierarchies},
Comm. Math. Phys. {\bf 145} (1992), 57-84.

\bibitem[DF95]{DF95}
Delduc F., Feh\'er L.,
\emph{Regular conjugacy classes in the Weyl group and integrable hierarchies},
J. Phys. A {\bf 28} (1995), n.20, 5843-5882 .


\bibitem[DSK13]{DSK13}
De Sole A., Kac V.G.,
\emph{Non-local Poisson structures and applications to the theory of integrable systems},
Jpn. J. Math. {\bf8} (2013), n.2, 233-347.

\bibitem[DSKV13]{DSKV13}
De Sole A., Kac V. G., Valeri D.,
\emph{Classical $\mc W$-algebras and generalized Drinfeld-Sokolov bi-Hamiltonian systems
within the theory of Poisson vertex algebras},
Comm. Math. Phys. {\bf 323} (2013), n.2, 663-711.

\bibitem[DSKV14a]{DSKV14a}
De Sole A., Kac V. G., Valeri D.,
\emph{Classical $\mc W$-algebras and generalized Drinfeld-Sokolov hierarchies
for minimal and short nilpotents},
Comm. Math. Phys. {\bf 331} (2014), n.2, 623-676.

\bibitem[DSKV14b]{DSKV14b}
De Sole A., Kac V. G., Valeri D.,
\emph{Dirac reduction for Poisson vertex algebras},
Comm. Math. Phys. {\bf 331} (2014), n.3, 1155-1190.

\bibitem[DSKV15]{AGD}
De Sole A., Kac V.G., Valeri D.,
\emph{Adler-Gelfand-Dickey approach to classical $\mathcal W$-algebras 
within the theory of Poisson vertex algebras},
Int. Math. Res. Not. {\bf21} (2015), 11186-11235.

\bibitem[DSKV15-cor]{DSKV15-cor}
De Sole A., Kac V. G., Valeri D.,
\emph{Erratum to: Classical $\mc W$-algebras 
and generalized Drinfeld-Sokolov hierarchies for minimal and short nilpotents},
Comm. Math. Phys. {\bf 333} (2015), n.3, 1617-1619.

\bibitem[DSKV16a]{DSKV16a}
De Sole A., Kac V. G., Valeri D.,
\emph{Structure of classical (finite and affine) $\mc W$-algebras},
J. Eur. Math. Soc. {\bf 18} (2016), n. 9, 1873-1908.

\bibitem[DSKVnew]{DSKVnew}
De Sole A., Kac V. G., Valeri D.,
\emph{A new scheme of integrability for (bi)-Hamiltonian PDE},
Comm. Math. Phys. {\bf 347} (2016), no.2, 449--488.

\bibitem[DSKV16b]{DSKV16b}
De Sole A., Kac V.G., Valeri D.,
\emph{Classical affine $\mc W$-algebras for $\mf{gl}_N$ and associated integrable Hamiltonian hierarchies},
Comm. Math. Phys. {\bf 348} (2016), n.1, 265-319.

\bibitem[DSKV17]{DSKV17}
De Sole A., Kac V.G., Valeri D.,
\emph{Finite $W$-algebras for $\mf{gl}_N$},
Adv. Math., to appear, arXiv: 1605.02898.

\bibitem[DSKV18]{DSKV18}
De Sole A., Kac V.G., Valeri D.,
\emph{A Lax type operator for quantum finite W-algebras},
preprint, arXiv:1707.03669

\bibitem[Dic03]{Dic03}
Dickey L. A.,
\emph{Soliton equations and Hamiltonian systems},
Advanced series in mathematical physics, World scientific, Vol. 26, 2nd Ed., 2003.

\bibitem[DS85]{DS85}
Drinfeld V.G., Sokolov V.V.,
\emph{Lie algebras and equations of KdV type},
Soviet J. Math. {\bf 30} (1985), 1975-2036.

\bibitem[FT07]{faddeev}
Faddeev L.D., Takhtajan L.A.,
\emph{Hamiltonian methods in the theory of solitons}. Translated from the 1986 Russian original by Alexey G. Reyman.
Reprint of the 1987 English edition. Classics in Mathematics. Springer, Berlin, 2007.

\bibitem[FHM93]{FHM93}
Feh\'er L., Harnad J., Marshall I.,
\emph{Generalized Drinfeld-Sokolov reductions and KdV type hierarchies},
Comm. Math. Phys. {\bf 154} (1993), n.1, 181-214.

\bibitem[FGMS95]{FGMS95}
Fern\'andez-Pousa C., Gallas M., Miramontes L., S\'anchez Guill\'en J., 
\emph{$\mc W$-algebras from soliton equations and Heisenberg subalgebras}. 
Ann. Physics {\bf 243} (1995), n.2, 372-419.

\bibitem[FGMS96]{FGMS96}
Fern\'andez-Pousa C., Gallas M., Miramontes L., S\'anchez Guill\'en J., 
\emph{Integrable systems and $\mc W$-algebras}, 
VIII J. A. Swieca Summer School on Particles and Fields (Rio de Janeiro, 1995), 475-479.

\bibitem[GGRW05]{GGRW05}
Gelfand I.M., Gelfand S.I., Retakh V. and Wilson R.L.,
\emph{Quasideterminants}
Adv. Math. {\bf 193} (2005), n.1, 56-141.

\bibitem[GD76]{GD76}
Gelfand I. M., Dickey L. A.,
\emph{Fractional powers of operators and Hamiltonian systems},
Funct. Anal. Appl. {\bf 10} (1976), 259-273.

\bibitem[GD78]{GD78}
Gelfand I. M., Dickey L. A.,
\emph{A family of Hamiltonian structures connected with integrable nonlinear 
differential equations},
Akad. Nauk SSSR Inst. Prikl. Mat. Preprint N.136 (1978).

\bibitem[Kau80]{Kau80}
Kaup D.J.,
\emph{On the inverse scattering problem for cubic eigenvalue problems of the class 
$\psi_{xxx}+6Q\psi_x+6R\psi=\lambda\psi$},
Stud. Appl. Math. {\bf62} (1980), no. 3, 189-216. 

\bibitem[Mag78]{Mag78}
Magri F.,
\emph{A simple model of the integrable Hamiltonian equation},
J. Math. Phys. {\bf 19} (1978), n.5, 1156-1162.

\bibitem[Miu68]{Miu68}
Miura R.,
\emph{Korteweg-de Vries Equation and Generalizations. I. A Remarkable Explicit Nonlinear Transformation},
J. Math. Phys. {\bf9} (1968), 1202-1204

\bibitem[Mol07]{Mol07}
Molev A.I. \emph{Yangians and classical Lie algebras},
Mathematical Surveys and Monographs, {\bf 143},
American Mathematical Society, Providence, RI, 2007.

\bibitem[MR15]{MR15}
Molev A.I., Ragoucy E., \emph{Classical W-algebras in types A, B, C, D and G},
Comm. Math. Phys. {\bf 336} (2015), n.2, 1053-1084.

\bibitem[SK74]{SK74}
Sawada K., Kotera T.,
\emph{A method for finding $N$-soliton solutions to the KdV equation and KdV-like
equation},
Progr. Theoret. Phys. B {\bf 51} (1974), 1355-1367.

\bibitem[YO76]{YO76}
Yajima N., Oikawa M.,
\emph{Formation and interaction of sonic-Langmuir solitons--inverse scattering method},
Progr. Theoret. Phys. {\bf56} (1976), n.6, 1719-1739. 

\end{thebibliography}
\end{document}